\documentclass[journal,onecolumn]{IEEEtran}
\usepackage{bbm}
\usepackage{graphicx}
\usepackage{caption}
\usepackage{lipsum}
\usepackage{nccmath}
\usepackage{mathtools}
\usepackage{cuted}
\usepackage{cleveref}
\crefname{section}{§}{§§}
\def\ci{\perp\!\!\!\perp}
\def\dperp{\perp\!\!\!\perp}
\usepackage{tikz}
\usetikzlibrary{arrows,decorations,fit,calc,shapes,backgrounds} 
\tikzstyle{box} = [draw, rectangle, rounded corners, thick, node distance=7em, text width=6em, text centered, minimum height=3.5em]
\tikzstyle{container} = [draw, rectangle, dashed, inner sep=2em]\tikzstyle{line} = [draw, thick, -latex']
\usepackage{longtable}

\usepackage{rotating}
\usepackage{amsmath,amsthm, amssymb,accents,mathalfa}
\usepackage{fancyhdr}
\usepackage{mathdots}
\usepackage{algorithm}
\usepackage{comment}
\usepackage[final,notref,notcite]{showkeys}
\usepackage[T1]{fontenc} 
\usepackage[runs=2]{auto-pst-pdf}

\usepackage{bpchem}

\usepackage{tikz}
\usetikzlibrary{arrows,decorations,fit,calc,shapes,backgrounds} 
\tikzstyle{box} = [draw, rectangle, rounded corners, thick, node distance=7em, text width=6em, text centered, minimum height=3.5em]
\tikzstyle{container} = [draw, rectangle, dashed, inner sep=2em]\tikzstyle{line} = [draw, thick, -latex']
\usepackage[latin1]{inputenc} 
\usepackage{graphics}
\usepackage{enumerate}

\def\BER{\mathrm{BER}}
\def\Ber{\mathrm{Ber}}
\def\BP{\mathrm{BP}}
\def\MAP{\mathrm{MAP}}
\def\BEC{\mathrm{BEC}}
\def\BSC{\mathrm{BSC}}
\def\BMS{\mathrm{BMS}}

\def\TV{\mathrm{TV}}
\def\eqdef{\triangleq}
\def\supp{{\mathrm{supp}}}
\def\hrank{{\mathrm{hrank}}}
\def\rank{{\mathrm{rank}}}
\def\F{\mathbb{F}}
\def\ind{\mathbbm{1}}
\def\P{{\bf P}}
\def\PP{{\bf P}}
\def\EE{{\bf E}}

\usepackage{cite}

\ifCLASSINFOpdf
\else
\fi
\usepackage{array}

\ifCLASSOPTIONcompsoc
  \usepackage[caption=false,font=normalsize,labelfont=sf,textfont=sf]{subfig}
\else
  \usepackage[caption=false,font=footnotesize]{subfig}
\fi
\theoremstyle{definition}
\newtheorem{definition}{Definition}

\newtheorem{example}{\bf Example}

\newtheorem{remark}{Remark}

\newtheorem{theorem}{\bf Theorem}
\newtheorem{lemma}{\bf Lemma}
\newtheorem{proposition}{\bf Proposition}

\newtheorem{conjecture}{\bf Conjecture}

\usepackage[noend]{algpseudocode}
\usepackage{makecell}
\begin{document}
%
\title{Low density majority codes and \\
	the problem of graceful degradation}
%
%
%

\author{Hajir~Roozbehani\quad~\IEEEmembership{}
        Yury Polyanskiy~\IEEEmembership{}
 
\thanks{H.R. is with the Laboratory of Information and Decision Systems at MIT.
e-mail: hajir@mit.edu. Y.P. is with the Department of Electrical Engineering and Computer Science, MIT.
e-mail: yp@mit.edu. This work was supported by the Center for Science of Information (CSoI), an NSF Science and Technology Center, under grant agreement CCF-09-39370 and by the NSF grant CCF-17-17842.}}
\maketitle

\begin{abstract} 
We study a problem of constructing codes that transform a channel with high bit error rate (BER) into  one with low BER
 (at the expense of rate). Our focus is on obtaining codes with smooth (``graceful'') input-output BER curves (as
 opposed to threshold-like curves typical for long error-correcting codes).

This paper restricts attention to binary erasure channels (BEC) and contains three contributions. First, we introduce the notion of Low Density Majority Codes (LDMCs). These codes are non-linear sparse-graph codes, which output majority function evaluated on randomly chosen small subsets of the data bits. This is similar to
 Low Density Generator Matrix codes (LDGMs), except that the XOR function is replaced with the majority. We show that
 even with a few iterations of belief propagation (BP) the attained input-output curves provably improve upon
 performance of any linear systematic code. The effect of non-linearity bootstraping the initial iterations of BP,
 suggests that LDMCs should improve performance in various applications, where LDGMs have been used traditionally. 
 
 Second, we establish several \textit{two-point converse bounds} that lower bound the BER achievable at one erasure
 probability as a function of BER achieved at another one. The novel nature of our bounds is that they are specific
 to subclasses of codes (linear systematic and non-linear systematic) and outperform similar bounds implied by the area
 theorem for the EXIT function.

 Third, we propose a novel technique for rigorously
 bounding performance of BP-decoded sparse-graph codes (over arbitrary binary input-symmetric channels). This is based on an extension of Mrs.Gerber's lemma to the
 $\chi^2$-information and a new characterization of the extremality under the less noisy order.
\end{abstract}
\tableofcontents

\section{Introduction}

The study in this paper is largely motivated by the following common engineering problem. Suppose that after years of 
research a certain error-correcting code (encoder/decoder pair) was selected, patented, standardized and implemented in a highly optimized
hardware. The code's error performance will, of course, depend on the level of noise that the channel applies to the
transmitted codeword. Often, for good modern codes, a very small error is guaranteed provided only that the channel's
Shannon capacity slightly exceeds code's rate. However, in practice the channel conditions may vary and engineer may
find herself in a situation where the available hardware-based decoder is unable to handle the current noise level. The
natural solution, such as the one for example being applied in optical communication~\cite{zhang2017low,barakatain2018low}, is to first use an inner code to decrease
the overall noise to the level sufficient for the preselected code to recover information. Thus, the goal of the inner
code is to match the variable real-world channel conditions to a prescribed lower level noise. 

In other words, the source bits $S_1,\ldots,S_k \in \{0,1\}$ that appear at the input of the inner code need not be
reconstructed perfectly, but only approximately. We can distinguish two cases of the operation of the inner code. In
one, it passes upstream a hard-decision about each bit (that is, an estimate $\hat S_i \in \{0,1\}$ of the true value of
$S_i$). In such a case, a figure of merit is a curve relating the channel's noise to the probability of bit-flip error
of the induced channel $S_i \mapsto \hat S_i$ that the outer code is facing.\footnote{Note that a more complete
description would be that of the ``multi-letter'' induced channel $S^k \mapsto \hat S^k$. However, in this paper we
tacitly assume that performance of the outer code is unchanged if we replace the multi-letter channel with the parallel
independent channels $\{S_i \mapsto \hat S_i\}_{i\in [k]}$. This is justified, for example, if the outer code uses a
large interleaver (with lenght much larger than $k$), or employs an iterative (sparse-graph) decoder.} In the second
case, the inner code passes upstream a soft-decision $\hat S_i$ (that is, a posterior probability distribution of $S_i$
given the channel output). In this case, the figure of merit is a curve relating the channel's noise to the capacity of
the induced channel $S_i \mapsto \hat S_i$.

Information theoretically, constructing a good inner code, thus, is a problem of lossy joint-source channel coding
(JSCC) under either the Hamming loss or the log-loss, depending on the hard- or soft-decision decoding. The crucial
difference in this paper, compared to the classical JSCC problem, is that we are interested not in a loss achieved over a 
single channel, but rather in a whole curve of achieved losses over a family of channels. What types of curves are
desirable? First of all, we do not want the loss to drop to zero at some finite noise level (since this will then be an
overkill: the outer code has nothing to do). Second, it is desirable that the loss decrease with channel improvement,
rather than staying flat in a range of parameters (which would be achievable by a separated compress-then-code scheme).
These two suggest that the resulting curve should be a smooth and monotone one. With such a requirement the 
problem is known as \textit{graceful degradation} and has been attracting 
attention since the early days of channel coding~\cite{ziv1970behavior,red79}. Despite this, no widely accepted solution
is available. This paper's main purpose is to advocate the usage of sparse-graph \textit{non-linear} codes for the
problems of graceful degradation and channel matching.

  \begin{figure}[ht]
\centering
\includegraphics[width=0.5\textwidth]{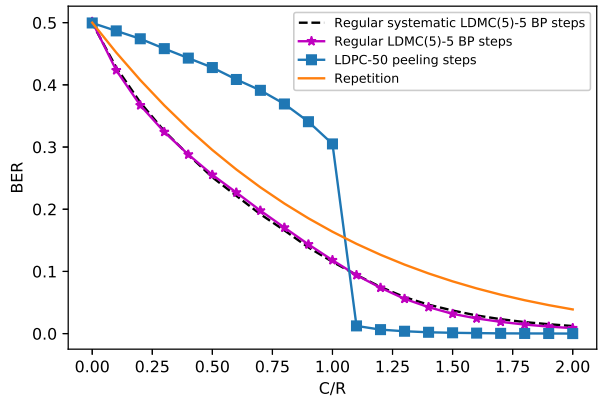}
\caption{BER performance vs. capacity-to-rate-ratio of the erasure channel for four codes with rate $R={1\over5}$ and $k=3\times 10^5$ data bits: an LDPC code using 50 iterations of peeling decoder, the repetition code, a regular (non-systematic) LDMC(5)
using 5 iterations of BP, and a regular systematic LDMC(5) using 5 iterations of BP. The LDPC code is the dual of a systematic LDGM(16) code, i.e., the variable nodes have degree 17 (and check nodes have approximately Poisson distributed degrees). The LDPC code suffers from the cliff effect, while the LDMCs and the repetition codes degrade
gracefully. The LDMC codes uniformly dominates the repetition code for all levels of the channel noise.}


\label{fig:ber_ldpc_rep}
\end{figure}

In this paper we will be considering a case of JSCC for a binary source and binary erasure channel (BEC) only.  Let
$S^k=(S_1,S_2,\cdots,S_k)\sim \mathrm{Ber}(1/2)^{\otimes k}$ be {\em information bits}. An {\em encoder}
$f:\{0,1\}^k\to\{0,1\}^n$ maps $S^k$ to a (possibly longer) sequence $X^n=(X_1,\cdots,X_n)$ where each $X_i$ is called a
{\em coded bit} and $X^n$ -- a codeword.  The rate of the code $f$ is denoted by $R=k/n$ and its bandwidth expansion by
$\rho=n/k$.  A channel $\mathrm{BEC}_{\epsilon}$ takes $X^n$ and produces $Y^n = (Y_1,\ldots, Y_n)$ where each $Y_j=X_j$
with probability $(1-\epsilon)$ or $Y_j=?$ otherwise. In this paper we will be interested in performance of the code
simultaneously for multiple values of $\epsilon$, and for this reason we denote $Y^n$ by $Y^n(\epsilon)$ to emphasize the
value of the erasure probability. Upon observing the distorted information $Y^n(\epsilon)$,  {\em decoder}\footnote{The
decoder may or may not use the knowledge of
    $\epsilon$, but for the BEC this is irrelevant.} $g$ produces $\hat S^k(\epsilon)=g(Y^n(\epsilon))$. We measure quality of the decoder by the data bit error rate (BER):
    $$ \mathrm{BER}_{f,g}(\epsilon) \triangleq \frac{1}{k} \sum_{i=1}^k \PP[S_i \neq \hat S_i(\epsilon)] = \frac{1}{k}
    \EE[d_H(S^k,\hat S^k(\epsilon))]\,,$$
    where $d_H$ stands for the Hamming distance. 

Consider now Fig.~\ref{fig:ber_ldpc_rep} which plots the BER functions for some codes. 
We can see that using LDPC code as an inner
code is rather undesirable: if the channel noise is below the BP threshold, then the BER is almost zero and hence the
outer code has nothing to do (i.e. its redundancy is being wasted). While if the channel noise is only slightly above
the BP threshold the BER sharply rises, this time making outer code's job impossible. In this sense, a simple
(blocklength 5) repetition code might be preferable. However, the \textit{low-density majority codes} (LDMCs) introduced
in this paper universally improve upon the repetition. This fact (universal domination) should be surprising for several
reasons. First, there is a famous area theorem in coding theory~\cite{ashikhmin2004extrinsic}, which seems to suggest
that BER curves for any two codes of the same rate should crossover\footnote{The caveat here is that area theorem talks
about about BER evaluated for coded bits, while here we are only interested in the data (or systematic) bits.}. Second,
a line of research in \textit{combinatorial JSCC}~\cite{kochman2012adversarial} has demonstrated that 
the repetition code is optimal~\cite{YP15-abmaps,roozbehani2018input} in a certain sense, that losely speaking
translates into a high-noise regime here. Third, Prop.~\ref{prop:linear_erasure} below shows that no linear code of rate 1/2 can
uniformly dominate the repetition code. 
Thus, the universal domination of the repetition code by the LDMC observed on 
 Fig.~\ref{fig:ber_ldpc_rep} 
 came as a surprise to us.

An experienced coding theorist, however, will object to Fig.~\ref{fig:ber_ldpc_rep} on the grounds that LDGMs, not
LDPCs, should be used for the problem of error-reduction -- this was in fact the approach
in~\cite{zhang2017low,barakatain2018low}. Indeed, the LDGM codes decoded with BP can have rather graceful BER
curves -- e.g. see Fig.~\ref{fig:LDGM_LDMC_p9}(a) below. So can LDMCs claim superior performance to LDGMs too? Yes, and
in fact in a certain sense we claim that LDMCs are superior to any linear systematic codes. This is the message of
Fig.~\ref{fig:two_user_converse_intro} and we explain the details next.

What kind of (asymptotic) fundamental limits can we define for this problem? Let us fix the rate $R={k\over n}$ of a
code. The lowest possible BER $\delta^*(R,W)$ achievable over a (memoryless) channel $W$ is found from comparing the source
rate-distortion function with the capacity $C(W)$  of the channel:
\begin{equation}\label{eq:CRD}
		R(1 - h_b(\delta^*(R,W))) = C(W)\,, \qquad h_b(x) = -x\log x - (1-x)\log(1-x)\,.
\end{equation}	
Below we call $\delta^*(R,W)$ a \textit{Shannon single-point bound}. Single-point here means that this is a fundamental limit
for communicating over a single fixed channel noise level. As we emphasized, graceful degradation is all about looking
at a multitude of noise levels. A curious lesson from the multi-user information theory shows
that it is not possible for a single code to be simultaneously optimal for two channels (for the BSC this was shown
in~\cite{kochman2018ozarow} and Prop.~\ref{prop:twopt_gen} shows it for the BEC). 

Correspondingly, we introduce a \textit{two-point fundamental limit}:
\begin{equation}\label{eq:delta2}
		\delta_2^*(R,\delta_a,W_a, W) = \limsup_{k\to\infty} \inf {1\over k} \EE_{Y^n\sim W(\cdot|X^n)}[d_H(S^k, \hat S^k)]\,,
\end{equation}	
where the infimum is over all encoders $f:S^k\to
X^n$ and decoders $g: Y^n \to \hat S^k$ satisfying
	$$ {1\over k} \EE_{Y^n\sim W_a(\cdot|X^n)}[d_H(S^k, \hat S^k)] \le \delta_a\,,$$
where $\delta_a, W_a$ is the anchor distortion and anchor channel. In other words, the value of $\delta_2^*$ shows the
lowest distortion achievable over the channel $W$ among codes that are already sufficiently good for a channel $W_a$.
Clearly, the two-point performance is related to a two-user broadcast channel~\cite[Chapter 5]{elg11} -- this is further discussed in Section~\ref{sec:prior} below.

Similarly, we can make definitions of $\delta^*$ and $\delta_2^*$ but for a restricted class of encoders, namely linear
systematic ones. We bound $\delta_2^*$ in Prop. \ref{prop:twopt_gen} for general codes, in Prop. \ref{prop:two_pt_EXIT_area} for general systematic codes, and  
 in Theorem~\ref{thm:two_point_converse} for the  subclass of linear systematic codes.

\begin{figure}[t]
\centering
\subfloat[][General codes]{\includegraphics[width=0.5\textwidth]{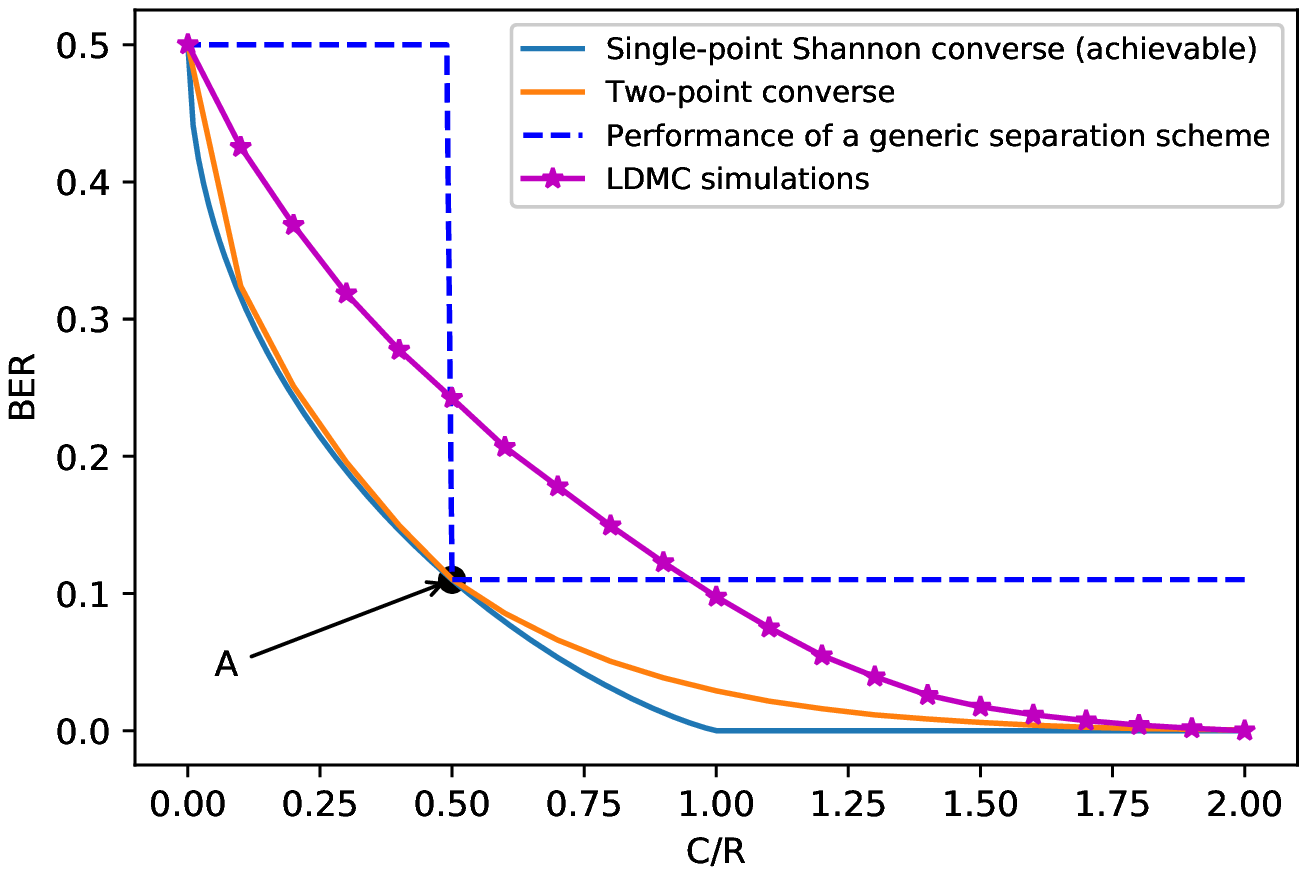}}
\subfloat[][Linear codes] { \includegraphics[width=0.5\textwidth]{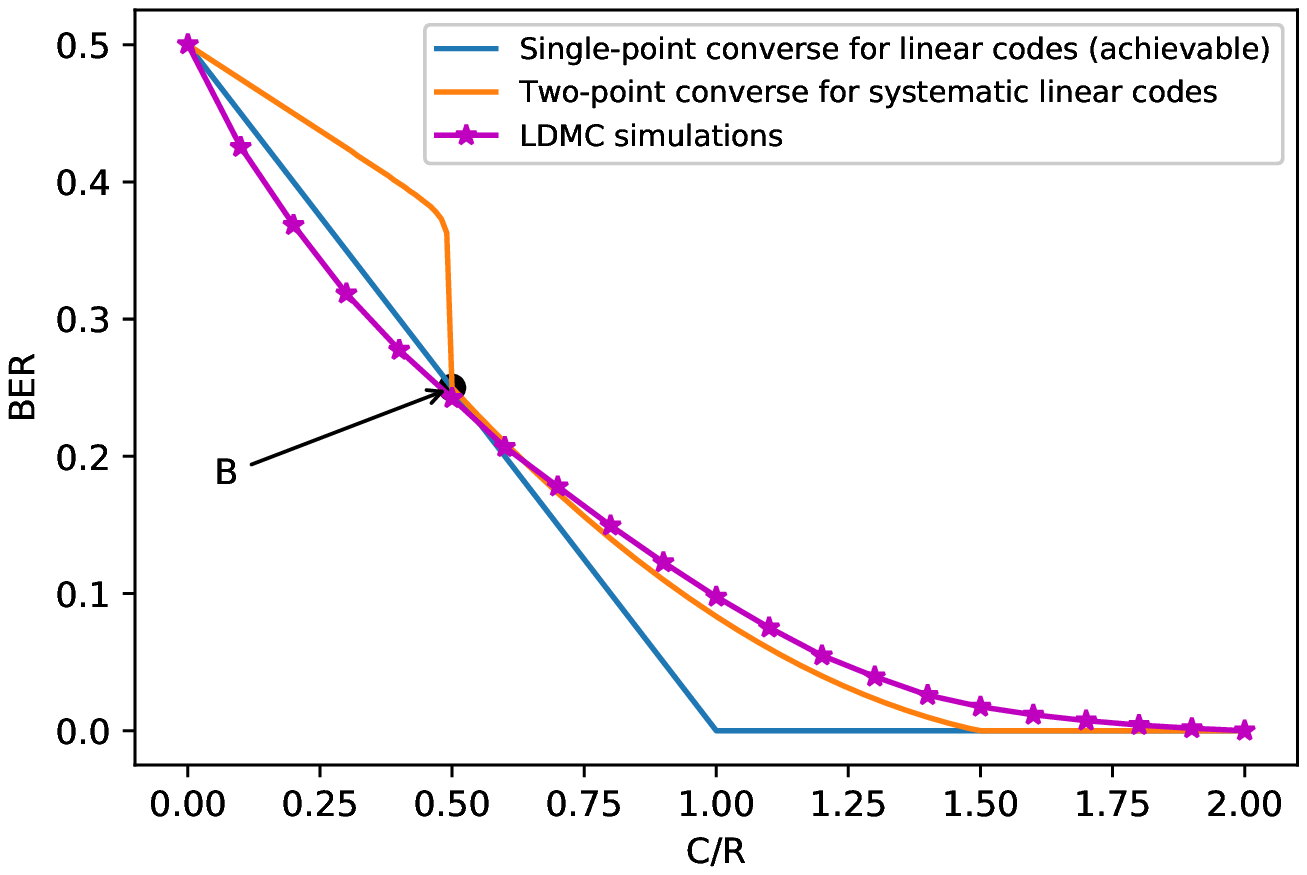}}
\caption{BER performance vs. capacity-to-noise ratio of the binary erasure channel for rate $R=1/2$ codes. The magenta curve
corresponds to simulation results of the systematic regular LDMC(9) with $k=10^5$ data bits using 3 iterations of BP. Left figure (a) compares LDMC against 
lower bounds for general codes: Shannon's single-point bound from~\eqref{eq:CRD} 
and a two-point from Prop.~\ref{prop:twopt_gen}. While every point A on the Shannon converse curve is achievable by a separation
compress-then-code architecture, the latter suffers from non-gracefulness (dashed blue curve). Right figure (b) compares
LDMC against lower bounds for linear codes. The single-point line is from Proposition \ref{prop:lin_single} below. While any point B arbitrarily close to the single-point line is
achievable by a linear systematic code, any such code will suffer from non-gracefulness. This is shown by the two-point
converse curve on the right. It plots evaluation of Theorem~\ref{thm:two_point_converse}, which lower bounds BER of any 
linear systematic code achieving $\BER\le 0.2501$ over $\BEC(0.75)$. In all, LDMCs are superior to any other linear
systematic code whose curve passes through point B.}
\label{fig:two_user_converse_intro}

\end{figure}

Armed with these definitions, we can return to Fig.~\ref{fig:two_user_converse_intro}. What it clearly demonstrates is
two things. On one hand, the single-point (left subplot) comparison shows that this specific LDMC is far from Shannon
optimality at any value of the channel noise. On the other hand, the two-point (right subplot) comparison shows LDMC outperforming 
any rate-$1\over 2$ linear systematic code among the class of those which have comparable performance at the anchor point $C(W_a)=1/4$.

We hope that this short discussion, along with Fig.~\ref{fig:ber_ldpc_rep}-\ref{fig:two_user_converse_intro}, convinces the reader that indeed the LDMCs (and in general adding non-linearity to
the encoding process) appear to be the step in the right direction for the problem of \textit{graceful degradation}. 
However, perhaps even more excitingly, performance of available LDPC/LDGM codes can be improved by adding some fraction
of the LDMC nodes -- the reasons for this are discussed in Section~\ref{sec:LDMC_construction} and~\ref{sec:comp_ldgm} below.

\textbf{Our main contributions:} 
\begin{itemize}
\item The concept of the LDMC and its favorable properties for the error-reduction and
channel matching. 
\item A new two-point lower bound for the class of systematic linear codes. (We also show that our bound is better than
anything obtainable from the area theorem.) It demonstrates that over the BEC linear codes coming close to the single-point
optimality can never be graceful -- their performance has a threshold-like behavior.
\item A two-point lower bound for general systematic codes based on a new connection between BER and EXIT functions. We show that the bound improves on the best known two-point converses for general codes in certain regimes of interest. The bound shows that at high rates general systematic codes achieving small error near capacity cannot be graceful.
\item A new method for (rigorously) bounding performance of sparse graph codes (linear and non-linear) based on the
less-noisy comparison and a new extension of the Mrs. Gerber's lemma to $\chi^2$-mutual information.
\item An idea of adding some fraction of non-linear (e.g. majority) factors into the
factor graph. It is observed that this improves early dynamics of the BP by yielding initial estimates for some
variable nodes with less redundancy than the systematic encoding would do. In particular, an LDGM can be uniformly
improved by replacing repetition code (degree-1 nodes) with LDMCs. 
\end{itemize}

\textbf{Paper organization.} 

The rest of this paper is organized as follows. We introduce the notion of LDMC codes and connections with the previous work in the remainder of this section. In Section \ref{sec:two_point} we give lower bounds for the two-point fundamental limit of linear systematic codes. These bounds improve significantly on the best known general converses \cite{kochman2019lower} in the case of linear codes. We use these bounds to show that LDMCs are superior to any linear code for error reduction in the two-point sense. We further use these bounds to show that, unlike LDMCs, linear codes of rate 1/2 cannot uniformly dominate the repetition code. Such bounds are naturally related to the area theorem of coding \cite{measson2004life}. In Section \ref{sec:bounds_area_thm}, we study the implications of the area theorem and show that our bounds are superior to those obtained from the area theorem. In Section \ref{sec:channel_comparison} we present our channel comparison lemmas and use them to construct new tools for analysis of BP dynamics in general sparse-graph codes. When applied to special classes of LDMCs, these bounds can accurately predict the performance. We study the applications of LDMCs to code optimization in Section \ref{sec:comp_ldgm}, where we show that combined LDGM/LDMC designs can uniformly dominate LDGMs both in sense of the performance as well as the rate of convergence. In Section \ref{sec:channel_transform}, we study LDMCs from the perspective of a channel transform and give bounds for soft-decision decoding using the channel comparison lemmas of Section \ref{sec:channel_transform}. These bounds give new ways to analyze BP and maybe of separate interest in inference problems outside of coding.

\subsection{The $\mathrm{LDMC}$ ensemble}
\label{sec:LDMC_construction}
We first define the notion of a check regular code ensemble generated by a Boolean function. 
\begin{definition}
Let $\P_\Delta$ be a joint distribution on $m$-subsets of $[k]$. Given a Boolean function $f:\{0,1\}^m\to\{0,1\}$, the (check regular) ensemble of codes on $\{0,1\}^k$ generated by $(f,\P_\Delta)$ is the family of random codes $f_\Delta:s\mapsto (f(s_T))_{T\in \Delta}$ obtained by sampling $\Delta\sim \P_\Delta$. 
 Here $s_T$ is the restriction of $s$ to the coordinates indexed by $T$. 
\end{definition}
Given $x\in\{0,1\}^d$, we consider the $d$-majority function 
\[
\textup{d-maj}(x)=\ind_{\{\sum_{i}x_i>\frac{d}{2}\}}.
\]
We have the following definition:
\begin{definition}
Let ${\bf U}_\Delta=\textup{Unif}^{\otimes n}(\{\textup{d-subsets of } [k]\})$ be the uniform product distribution on the $d$-subsets of $[k]$. 
The ensemble of codes generated by $(\textup{d-maj},{\bf U}_\Delta)$ is called the Low Density Majority Code (LDMC) ensemble of degree $d$ and denoted by LDMC($d$). Furthermore, define the event $A:=\cup_{ij}\{\sum_{T\in \Delta} \ind_{\{i\in T\}}=\sum_{T\in \Delta} \ind_{\{j\in T\}}\}$, i.e., the event that each $i$ appears in the same number of d-subsets $T$. Then the ensemble generated by $(\textup{d-maj},{\bf U}_{\Delta|A})$ is called a regular LDMC($d$) ensemble. 
\end{definition}

 Note that LDMC(1) is the repetition code. We shall also speak of systematic LDMCs, which are codes of the form $s\mapsto (s,f(s))$ where $f$ is picked from an LDMC ensemble. Another  ensemble of interest is that generated by the XOR function, known as the Low Density Generator Matrix codes (LDGMs).

Previously, we have introduced LDMCs in \cite{roozbehani2018triangulation} and have shown that they posses graceful degradation in the following (admittedly weak) senses:
\begin{itemize}
    \item Let 
    \begin{equation}
    \beta^*_f\triangleq\inf_{x,y} \{\frac{d(f(x),f(y))}{n}|d(x,y)\ge k-o(k)\}.
    \label{eq:betas}
    \end{equation}
    Operationally, $\beta^*$ characterizes the threshold for adversarial erasure noise beyond which the decoder cannot guaranteed to recover a single bit. Alternatively, $1-\beta_f^*$ is the fraction of equations needed such that $f$ can always recover at least one input symbol. 
    Then it was shown in \cite{roozbehani2018triangulation} that an LDMC codes achieve $\beta^*=1$ asymptotically with high probability. It was shown in \cite{roozbehani2018input} that repetition-like codes are the only linear codes achieving $\beta^*=1$. In particular, no such linear codes exist when the bandwidth expansion factor $\rho$ is not an integer. 
    \item It was shown in \cite{roozbehani2018triangulation} that LDMC codes can dominate the repetition code even with a sub-optimal (peeling) decoder. 
    A priori, it is not obvious if such codes must exist. Indeed Proposition \ref{prop:linear_erasure} below shows that no linear systematic code of rate 1/2 can dominate the repetition code. Furthermore, if we measure the quality of recovery w.r.t output distortion, i.e., the bit error rate of coded bits, then the so called area theorem (see Theorem \ref{thm:area_thm} below) states that no code is dominated by another code. However, this is not case for input distortion. Even for systematic linear codes, it is possible for one code to dominate another as the next example shows. 

\begin{example} Let $f$ be the 2-fold repetition map $x\mapsto (x,x)$. Let $g$ be a systematic code sending $x_i\mapsto (x_i,x_i,x_i)$ for all odd $i$ and $x_j\mapsto (x_j)$ for all even $i$. Then
$\BER_f(\epsilon)=\frac{1}{2}\epsilon^2$ and $\BER_g(\epsilon)=\frac{1}{2}(\frac{1}{2}\epsilon^3+\frac{1}{2}\epsilon)$. It can be checked that $f$ dominates $g$. This means that among repetition codes a balanced repetition is optimal. 
\label{ex:no_conservation}
\end{example}

\end{itemize}

As we discussed in the introduction, LDMCs can achieve trade-offs that are not accessible to linear codes. The easiest
way to illustrate this is to consider a single step of a BP decoding algorithm and contrast it for the LDPC/LDGM and the
LDMC. For the former, suppose that we are at a factor node corresponding to an observation of $X=S_1+S_2+S_3 \mod
2$. If the current messages from variables $S_2,S_3$ are uninformative (i.e. $\Ber(1/2)$ is the current estimate of
their posteriors), then the BP algorithm will send to $S_1$ a still uninformative $\Ber(1/2)$ message. For this reason,
the LDGM ensemble cannot be decoded by BP without introducing degree-1 nodes. 

\textit{On the contrary,} if we are in the LDMC(3) setting and observe $X=\mathrm{maj}(S_1,S_2,S_3)$, then even if
messages coming from $S_2,S_3$ are $\Ber(1/2)$ the observation of $X$ introduces a slight bias and the message to $S_1$
becomes $\Ber(\frac{1+2Y}{2})$ 
This explains why unlike the LDGMs, the LDMCs can be BP-decoded
without introducing degree-1 nodes.

This demonstrates the principal novel effect introduced by the LDMCs. Not only it helps with the graceful degradation, but
it also can be used to improve the BP performance at the early stages of decoding. For this, one can ``dope'' the
LDPC/LDGM ensemble with a small fraction of the LDMC nodes that would bootstrap the BP in the initial stages. 
This is the message of Section~\ref{sec:comp_ldgm} below.

\subsection{Prior work}\label{sec:prior}

Note that the two-point fundamental limit $\delta_2^*$ defined in~\eqref{eq:delta2} can be equivalently restated as the
problem of finding a joint distortion region for a two-user broadcast channel~\cite[Chapter 5]{elg11}. As such, a number
of results and ideas from the multi-user information theory can be adapted. We survey some of the prior work here with
this association in mind.

\subsubsection{Rateless codes}

To solve the multi-cast problem over the internet, the standard TCP protocol uses feedback to deal with erasures, i.e.,
each lost packet gets re-transmitted. This scheme is optimal from a data recovery point of view. From any $k$ received
coded data bits, $k$ source bits can be recovered. Hence it can achieve every point on the fundamental line of
Fig.~\ref{fig:two_user_converse_intro}. 
However, this method has poor latency and is not applicable to the cases where a feedback link is not available. Thus,
an alternative forward error
correcting code can be used to deal with data loss~\cite{luby1997practical}.

  In particular, Fountain codes have been introduced to solve the problem of multi-casting over the erasure channel
  \cite{byers1998digital}. They are a family of linear error correcting codes that can recover $k$ source bits from any
  $k+o(k)$ coded bits with small overhead.  A special class of fountain codes, called systematic Raptor codes, have been
  standardized and are used for multi-casting in 3GPP
  \cite{bouras2012evaluating,gul2016merge,ali2018raptorq,bouras2013embracing,luby2011raptorq}.  Various extensions and
  applications of Raptor codes are known \cite{chen2015protograph,chen2011protograph}.  However, as observed in
  \cite{sanghavi2007intermediate}, these codes are not able to adapt to the user demands and temporal variations in the
  network. 
  
  As less data becomes available at the user's end, it is inevitable that our ability to recover the source
  deteriorates. However, we may still need to present some meaningful information about the source to the user, i.e., we
  need to partially recover the source.  For instance, in sensor networks it becomes important to maximize the
  throughput of the network at any point in time since there is always a high risk that the network nodes fail and
  become unavailable for a long time \cite{kamra2006growth}. In such applications it is important for the codes to
  operate gracefully, i.e., to partially recover the source and improve progressively as more data comes in.  We show in
  Section \ref{sec:two_point} that Fountain codes, and more generally linear codes, are not graceful for forward error
  correction. Hence, it is not surprising that many authors have tried to develop graceful linear codes by using partial
  feedback \cite{kamra2006growth,hagedorn2009rateless,beimel2007rt,cassuto2015online}. However, we shall challenge the
  idea that graceful degradation (or the online property) is not achievable without feedback \cite{cassuto2015online}.
  Indeed LDMCs give a family of efficient (non-linear) error reducing codes that can achieve graceful degradation and
  can perform better than any linear code in the sense of partial recovery (see Fig.~\ref{fig:two_user_converse_intro}).

   Raptor codes are essentially concatenation of a rateless Tornado type error-reducing code with an outer error
   correcting pre-coder. Forney \cite{forney1965concatenated} observed that concatenation can be used to design codes
   that come close to Shannon limits with polynomial complexity. Forney's concatenated code consisted of a high rate
   error correcting (pre)-coder that encodes the source data and feeds it to a potentially complicated inner error
   correcting code.  One special case of Raptor codes, called pre-code only Raptor code is the concatenation of an error
   correcting code with the repetition code. Recently, such constructions are becoming popular in optics. In these
   applications it is required to achieve $10^{-15}$ ouput BER, much lower than the error floor of LDPC. Concatenation
   with a pre-coder to clean up the small error left by LDPCs is one way to achieve the required output BER
   \cite{sugihara2013spatially}. It was shown recently however that significant savings in  decoding complexity (and
   power) can be achieved if the inner code is replaced with a simple error reducing code and most of the error
   correction is left to the outer code \cite{zhang2017low,barakatain2018low}.
   
These codes, as all currently known examples of concatenated codes, are linear.
They use an outer linear error correcting code (BCH, Hamming, etc) and an inner error reducing LDGM. The LDGM code however operates in the regime of partial data recovery. It only produces an estimate of the source with some distortion that is within the error correcting capability of the outer code. To achieve good error reduction, however, LDGMs still need rather long block-length and a minimum number of successful transmissions. In other words, they are not graceful codes (see Fig.~\ref{fig:LDGM_LDMC}). We will show in Section \ref{sec:comp_ldgm} that LDMCs can uniformly improve on LDGMs in this regime. Thus, we expect that LDMCs appear in applications where LDGMs are currently used for error reduction.

\subsubsection{Joint Source-Channel Coding}

  The multi-user problem discussed above can be viewed as an example of broadcasting with a joint source-channel code (JSCC), which is considered one of the challenging open problems is network information theory \cite{kochman2019lower,kochman2018ozarow,reznic2006distortion,tan2013distortion,khezeli2016source}. In general it is known that the users have a conflict of interests, i.e., there is a tradeoff between enhancing the experience of one user and the others. For instance, if we design the system to work well for the less resourceful users, others suffer from significant delay. Likewise, if we minimize the delay for the privileged users, others suffer significant loss in quality. Naturally, there are two questions we are interested in: 1) what are the fundamental tradeoffs for partial recovery 2) how do we design codes to achieve them?

    Several achievability and converse bounds are available for the two user case under various noise models
    \cite{reznic2006distortion,ozarow1980source,ahlswede1985rate,el1982achievable,cover1972broadcast}. Our bounds on the
    trade-off functions give new converses for broadcasting with linear codes. These bounds are stronger than those
    inferred from the area theorem (see Section \ref{sec:bounds_area_thm}) or the general converses known for JSCC (cf.
    \cite{kochman2019lower}). The ongoing efforts to find graceful codes
    \cite{kokalj2011cliff,vitali2007multiple,goyal2001multiple,bourtsoulatze2019deep,mohr2000unequal,goyal2001quantized}
    can also be broadly categorized as JSCC solutions. (Our proposed codes, LDMCs, may prove to be useful in this
    regard as well.)
    
    In all, in most cases there is a significant gap between achievablity and converse bounds. In a sense, the theory and
    practice of partial recovery under two-point (two user) setting is far less developed compared with the cases of
    full recovery or a partial recovery for a single-user case. However, even for the former case the best achievability results are either non-constructive \cite{el1982achievable}, or
 involve complicated non-linearities (e.g., compression at different scales
 \cite{goyal2001multiple}\cite{kokalj2011cliff}). In turn, the optimal single-user (single-point) schemes for partial
 recovery (such as separation) are very suboptimal for other operating points -- see \cite{kochman2019lower} and references therein. 
    
    As we discussed in the intro, classic error correction solutions are not completely satisfactory --  they suffer 
    significant loss in recovery once the channel quality
    drops below the designed one. This sudden drop in quality is known as the ``cliff effect'' \cite{goyal2001multiple}
    and shown in Fig.~\ref{fig:ber_ldpc_rep} for LDPC codes. 
 
    Our two-point converse results show that the ``cliff effect'' persists in the range of partial
    recovery as well (for linear systematic codes over the BECs). That is, any linear code that comes close to the fundamental limits of partial recovery cannot be
    graceful. This latter result cannot be inferred in this generality from either the area theorem (see Section \ref{sec:bounds_area_thm}) or the
    general two-point converses known for the JSCC problem.

Finally, when the channel noise is modeled as an adversarial (worst case analysis), the problem becomes an instance of
the so-called combinatorial joint-source-channel-coding (CJSCC)
\cite{kochman2012results,kochman2012adversarial,YP15-abmaps,you15,maz16}.  For this problem, just like for the two-point
JSCC, we have much stronger results for the case of the linear codes. The best converse for general codes are given
in~\cite{YP15-abmaps}, while~\cite[Section 5]{PS16-uncertain} imporives this bound for linear codes. In addition, it was
shown in \cite{roozbehani2018input} that linear codes exhibit with a prescribed minimum distance cannot have a favorable
JSCC performance.

\subsubsection{Non-linear codes}
Codes whose computational graph (see Fig.~\ref{fig:comp_trees}) are sparse are known as sparse graph codes. Many such codes are known \cite{mackay2005encyclopedia} and can achieve near Shannon limit performance. With a few exceptions, these codes are mostly linear. One problem with linear codes is that BP cannot be initiated without the presence of low degree nodes. In \cite{ciliberti2005lossy}, the authors observe that non-linear functions do not have this problem and use random sparse non-linear codes to achieve near optimal compression using BP. However, using non-linear functions in this setting is mainly due to algorithmic considerations, namely, to enable the use of BP. Otherwise, similar compression results can be obtained by using LDGMs under different message passing rules\cite{wainwright2005lossy}. In \cite{montanari2008smooth}, the authors use special non-linear sparse graph codes to build optimal smooth compressors. In all of these works, however, the focus is on point-wise performance and a result the codes are optimized to operate at a particular rate. As such, they are unlikely to achieve graceful degradation. 

Another relevant work in this area is that of random constraint-satisfaction problems (CSPs) with a planted solution
\cite{krzakala2009hiding}. It appears that the CSP literature mostly focused on geometric characterization of spaces of
solutions and phase transitions thereof. These do not seem to immediately imply properties interesting to us here (such
as graceful degradation). In addition, we are not aware of any work in this area that has studied the planted CSP problem
corresponding to local majorities, and in fact our experiments suggest that the phase diagram here is very different (in
particular, the information-computation gap appears to be absent).

\section{Two-point converse bounds} \label{sec:two_point}

\subsection{General codes}

The results on broadcast channels~\cite[Theorem 1]{tan2013distortion},\cite[Section V.B]{khezeli2016source},
 give the
following bound shown in~\cite[Section II.C2]{kochman2019lower}:
\begin{proposition}[\cite{kochman2019lower,tan2013distortion,khezeli2016source}]\label{prop:twopt_gen} Consider a sequence of codes of rate $R$ encoding iid $\Ber(1/2)$ bits and achieving the per-letter
Hamming distortion $\delta^*(R,W)$, cf.~\eqref{eq:CRD}, when $W=\BEC_\epsilon$. Then over $W=\BEC_{\tau}$  their distortion is lower bounded by
	\[
	\delta \ge \ind_{\{\tau\ge \epsilon\}}\eta(\delta^{*},\epsilon,\tau)+\ind_{\{\tau< \epsilon\}}\inf_y\{y:\eta(y,\tau,\epsilon)\le \delta^*\},
	\]
where
\[
\eta(\delta^{*},\epsilon,\tau)\triangleq\sup_{q\in [0,1/2]} \frac{h_b^{-1}\left((1-(1-\tau)/R)+\frac{1-\tau}{1-\epsilon}(h_b(q\ast \delta^*)-h_b(\delta^*))\right)-q}{1-2q}.
\]
\end{proposition}
\begin{remark}
We note that the above result is rate independent in the sense that it can be re-parametrized in terms of the capacity-to-rate ratios only. Proposition \ref{prop:two_pt_EXIT_area} below gives a new (rate dependent) bound for general systematic codes that in some regimes improves the above.
\end{remark}

\subsection{Systematic linear codes}

Systematic linear codes form a vast majority of the codes that are used in practice. In this section, we work towards
proving a two-point converse bound for this class of codes. 

In the following, by $\mathrm{ker}(A)$ we refer to the left kernel of $A$, that is the subspace of vectors $x$ satisfying  $xA=0$.

\begin{definition}
Given a matrix $A$ define
\[
\mathrm{hrank}(A)\triangleq\left|\{j:\mathrm{ker}(A)\subset \{x:x_j=0\}\}\right|.
\]
\end{definition}

\begin{definition}
Given a matrix $A$, define $\tilde{A}(p,q)$ to be a random sub-matrix of $A$ that is obtained by sampling each row of $A$ with probability $p$ and each column of $A$ with probability $q$ independently of other rows/columns. 
\end{definition}
The following proposition is well known (cf. \cite{richardson2008modern}).

\begin{proposition}
Consider a system of equations $xG=y$ over $\F_2$. If $\mathrm{ker}(G)\subset \{x:x_i=0\}$, then $x_i$ is uniquely determined from solving $xG=y$. Otherwise, there is a bijection between the set of solutions $\{x:xG=y, x_i=0\}$ and $\{x:xG=y,x_i=1\}$. In particular, if exactly $t$ coordinates are uniquely determined by the above equations, then $\mathrm{hrank}(G)= t$.
\label{prop:hrank}
\end{proposition}

From this proposition, it is easy to deduce the single-point-bound for linear codes:
\begin{proposition}
Let $f$ be a linear binary code of rate $1/\rho$. Then
\begin{equation}\label{eq:singlept_lin}
\BER_f(\epsilon)\ge \frac{1-\rho(1-\epsilon)}{2}.
\end{equation}
\label{prop:lin_single}
\end{proposition}
\begin{proof}
Let $N$ be the number of unerased coded bits returned by the channel. Note that in this case we  $\hrank\left(\tilde{A}(1,1-\epsilon)\right)\le N$. Since the map is linear, by Proposition \ref{prop:hrank},  at most $N$ source bits can be recovered perfectly by any decoder and the remaining coordinates cannot be guessed better than random. In other words, 
\[
\BER_f(\epsilon)\ge \frac{\EE[1-N/k]}{2},
\]
where the expectation is taken w.r.t to the distribution of $N$. The result follows since $\EE[N]=(1-\epsilon)n$.
\end{proof}

Our next proposition relates BER and hrank.
\begin{proposition}
Let $G=[I\,\, A]$ be the generator matrix of a systematic linear code $f$ with rate $R$. Then $\mathrm{BER}_f(\epsilon)\le \delta$ 
if and only if $$\EE[\mathrm{hrank} \left(\tilde{A}(\epsilon,1-\epsilon)\right)]\ge (\epsilon-2\delta)k.$$
\label{prop:ber_hrank}
\end{proposition}
\begin{proof}
If BER is bounded by $\delta$, there are, on average, at most $2\delta k$ bits that are not uniquely determined by solving $x\tilde{G}(1,1-\epsilon)=y$. For a systematic code, the channel returns $\mathrm{Bin}(k,1-\epsilon)$ systematic bits. The remaining systematic bits $x_r$ are to be determined from solving $ x_r\tilde{A}(\epsilon,1-\epsilon)=\tilde{y}$ where $\tilde{y}$ is some vector that depends on the channel output $y$ and the returned systematic bits. If $t$ additional systematic bits are recovered, then $\mathrm{hrank}(\tilde{A}(\epsilon,1-\epsilon))= t$ by Proposition \ref{prop:hrank}. Since on average at least $(\epsilon-2\delta)k$ additional systematic bits are recovered, the claim on the average hrank follows. 
\end{proof}

The next proposition shows how matrices with positive hrank behave under row sub-sampling. Our main observation is that row sub-sampled matrices of a (thin) matrix with large hrank have bounded rank. In particular, if a (thin) matrix has full hrank, its sub-sampled matrices cannot have full rank. 

\begin{proposition}
Consider an arbitrary field $\F$ and let $\epsilon_1> \epsilon_2$. Given a $k\times m$ matrix $A$, 
$$\EE[\rank \left(\tilde{A}(\epsilon_2,1)\right)] \le \rank(A)-(1-\frac{\epsilon_2}{\epsilon_1})\EE[\mathrm{hrank} \left(\tilde{A}(\epsilon_1,1)\right)],$$
and
\[
\EE[\mathrm{hrank} \left(\tilde{A}(\epsilon_2,1)\right)] \ge \frac{\epsilon_2}{\epsilon_1}\EE[\mathrm{hrank} \left(\tilde{A}(\epsilon_1,1)\right)].
\]
Therefore, if $\EE[\rank \left(\tilde{A}(\epsilon_2,1)\right)]=\rank(A)-o(k)$, then $\EE[\rank \left(\tilde{A}(\epsilon_1,1)\right)]=o(k)$.\\
\label{prop:hrank_subsampling}
\end{proposition}
\begin{proof}
Suppose that $\hrank\left(\tilde{A}(\epsilon_1,q)\right)= t$. This means that there are at least $t$ rows $a_j$ in $\tilde{A}(\epsilon_1,q)$ such that $a_j$ is not in the span of $\{a_i:i\neq j\}$. Let $B$ be the row-submatrix of $\tilde{A}(\epsilon_1,q)$ associated to these $t$ rows, and $B^c$ be its compliment, i.e., the matrix with rows $\{a_j:a_j\in \tilde{A}(\epsilon_1,q), a_j\not\in  B\}$. We claim that the compliment of $B$ is a matrix of rank $\rank(A)-t$. To see this, note that $\mathrm{Im}(B)\cap \mathrm{Im}(B^c)=\{0\}$, for otherwise we get linear dependencies of the form $h=\sum_i \alpha_i b_i\neq 0$ where $b_i\in B$ and $h\in \mathrm{Im}(B^c)$, which contradicts the construction of $B$. This means that  $\rank(B^c)+\rank(B)= \mathrm{rank}(A)$. The claim now follows since $\rank(B)=t$. Under row sub-sampling, each row of $B$ is selected with probability $\epsilon_2/\epsilon_1$ independently of other rows. Thus,
\[
\EE[\hrank\left(\tilde{A}(\epsilon_2,q)\right)|\hrank\left(\tilde{A}(\epsilon_1,q)\right)= t]\ge \frac{\epsilon_2}{\epsilon_1}t
\]
The rows selected from $B^c$ can contribute at most $\rank(A)-t$ to the rank of $\tilde{A}(\epsilon_2,q)$. Hence 
\[\EE[\rank \left(\tilde{A}(\epsilon_2,q) \right)|\hrank(\tilde{A}(\epsilon_1,q))=t]\le\frac{\epsilon_2}{\epsilon_1}t+\rank(A)-t
\]
Taking the average over the hrank of $\tilde{A}(\epsilon_1,q)$ proves the first two results. The last inequality follows by re-arranging the terms. 


\end{proof}
\begin{remark}
In general the above bound cannot be improved up to $o(k)$ deviations. Indeed we can partition the matrix $\tilde{A}(\epsilon_1,1)$ in the form
\[
\left[\begin{array}{c}B\\O\\F \end{array}\right]
\]
where $B$ is a basis with $\hrank(\tilde{A}(\epsilon_1,1))$ many rows, $O$ is the zero matrix, and $F$ is a redundant frame with $f>1-\epsilon_1-t$ rows that span the co-kernel of $B$. This means that any $1-\epsilon_1-t$ rows in $F$ form a basis for the image of $F$. Now for any $\epsilon_2<\epsilon_1$, if $\frac{\epsilon_2}{\epsilon_1}f=1-\epsilon_1$, then we sub-sample a basis from $f$ with high probability. Thus the $hrank$ of the sub-sampled matrix $\tilde{A}(\epsilon_2,1)$ can jump up with high probability for large $k$. 
\end{remark}

The next Proposition shows that rank is well behaved under column sub-sampling.
\begin{proposition}
Consider an arbitrary field $\F$ and let $p> q$. Given a $k\times m$ matrix $A$ over $\F$, $$\EE[\rank \left(\tilde{A}(1,p)\right)]\le \min\{pm, \frac{p}{q}\EE[\rank \left(\tilde{A}(1,q)\right)]\}.$$
\label{prop:rank_subsampling}
\end{proposition}

\begin{proof}
Pick a column basis for $\tilde{A}(1,p)$. We can realize $\tilde{A}(1,q)$ by sub-sampling columns of $\tilde{A}(1,p)$. In this way, each column in the basis of $\tilde{A}(1,p)$ is selected with probability $q/p$ independently of other columns. In other words, $$\EE[\rank \left(\tilde{A}(1,q)\right)]\ge  \frac{q}{p}\EE[\rank \left(\tilde{A}(1,p)\right)].$$ The desired result follows.
\end{proof}

We are now ready to prove our main result.
\begin{theorem}
Let $f:x\mapsto xG$ be a systematic linear code of rate $1/\rho$ with generator matrix $G=[I\,|\,A]$ over $\F_2$. Fix $\epsilon_1>\epsilon_2$ and $\delta_1\le \frac{\epsilon_1}{2}$. If $\mathrm{BER}_f(\epsilon_1)\le \delta_1$, then
\begin{equation}\label{eq:twpt_lin_1}
\mathrm{BER}_f(\epsilon_2)\ge \kappa(\rho,\delta_1,\epsilon_2,\epsilon_1)\triangleq \frac{\epsilon_2-\frac{1-\epsilon_2}{1-\epsilon_1}\left[\frac{\epsilon_2}{\epsilon_1}\gamma+(\rho-1)(1-\epsilon_1)-\gamma \right]}{2}
\end{equation}
with $\gamma=\epsilon_1-2\delta_1$. If $\epsilon_2>\epsilon_1$ then
\begin{equation}\label{eq:twpt_lin_2}
\mathrm{BER}_f(\epsilon_2) \ge {\epsilon_2\over 2} - {\epsilon_2 \over \epsilon_2 - \epsilon_1} {1\over 1-\epsilon_1}
\left(\delta_1 - {1\over 2} (1-\rho(1-\epsilon_1))\right)\,.
\end{equation}
In particular, if $\mathrm{BER}(\epsilon_1)=\frac{1}{2}(1-\rho(1-\epsilon_1))+o(1)$, then $\mathrm{BER}(\epsilon_2)=
\frac{\epsilon_2}{2}-o(1)$. 

\label{thm:two_point_converse}
\end{theorem}
\begin{proof}
By Proposition \ref{prop:ber_hrank}, we have $\EE[\hrank\left(\tilde{A}(\epsilon_1,1-\epsilon_1)\right)]\ge \gamma k$. By Proposition \ref{prop:hrank_subsampling}, we have 
\[
\EE[\rank\left(\tilde{A}(\epsilon_2,1-\epsilon_1)\right)]\le (\frac{\epsilon_2}{\epsilon_1}\gamma+(\rho-1)(1-\epsilon_1)-\gamma)k. 
\]
By Proposition \ref{prop:rank_subsampling}, we have
\[
\EE[\rank\left(\tilde{A}(\epsilon_2,1-\epsilon_2)\right)]\le \frac{1-\epsilon_2}{1-\epsilon_1} (\frac{\epsilon_2}{\epsilon_1}\gamma+(\rho-1)(1-\epsilon_1)-\gamma)k.
\]
The first result now follows from Proposition \ref{prop:ber_hrank} upon observing that $\hrank(\tilde{A})\le \rank(\tilde{A})$. 

For the second case, we trivially have the following estimate (by interchanging the roles of $\epsilon_1$ and
$\epsilon_2$ and applying the first part):
\[
\BER_f(\epsilon_2)\ge \inf_{\delta_2}\{\delta_2: \kappa(\rho,\delta_2,\epsilon_1,\epsilon_2)\le \delta_1\}.
\]
The estimate~\eqref{eq:twpt_lin_2} then follows by evaluating this infimum (which is a minimization of a linear
function).

\end{proof}

One simple application of Theorem \ref{thm:two_point_converse} demonstrates that no linear systematic code can
uniformly dominate 2-fold repetition.

\begin{proposition}
Let $g$ be the $2$-fold repetition code with bit-MAP decoder and $f$ be a linear systematic code of rate $1/2$. If there exists $\epsilon_2$ such that $\BER_f(\epsilon_2)<\BER_{g}(\epsilon_2)$, then there exists some $\epsilon^*>\epsilon_2$ such that for all $\epsilon_1\in (\epsilon^*,1)$ we have $\BER_f(\epsilon_1)>\BER_{g}(\epsilon_1)$. Moreover, if $\BER_f(\epsilon_2)=t\BER_g(\epsilon_2)$ for some $t<1$, then we can pick $\epsilon^*=\max(\epsilon_2,1-\frac{(1-t)\epsilon_2^2}{1-\epsilon_2})$.
\label{prop:linear_erasure}
\end{proposition}
\begin{proof}
The key idea of our proof could be explained as follows. The
repetition codes achieves the single-point converse line for linear codes~\eqref{eq:singlept_lin} with optimal slope at
$C\to 0$. However, Theorem 1 shows that no optimal slope is possible for any other systematic linear code.

We prove an equivalent claim.
Suppose that we have a code such that $\BER_f(\epsilon_1)\le {1\over 2}\epsilon_1^2$ for all $\epsilon_1$ in the
neighborhood of 1. We claim that then $\BER_f(\epsilon_2) \ge {1\over 2}
\epsilon_2^2$ for all $\epsilon_2$. Indeed, by (\ref{eq:twpt_lin_1}) in Theorem \ref{thm:two_point_converse}, for all $\epsilon_2<\epsilon_1$ we must have
\begin{equation}
\BER_f(\epsilon_2)\ge \frac{\epsilon_2}{2}-\frac{1}{2}(1-\epsilon_2)(\epsilon_2-\epsilon_1+1).
\label{eq:main_ineq1}
\end{equation}
Now suppose to the contrary that $\BER_f(\epsilon_2)=t\epsilon_2^2/2$ for some $\epsilon_2< \epsilon_1$ and $t< 1$. 
Then the above inequality can be written as
\begin{equation}
t\epsilon_2\ge \epsilon_2-(1-\epsilon_1)\frac{1-\epsilon_2}{\epsilon_2}.
\label{eq:main_ineq2}
\end{equation}
Clearly, this can not be satisfied for
\begin{equation}
\epsilon_1 >1 - \frac{(1-t)\epsilon_2^2}{1-\epsilon_2}.
\label{eq:main_ineq3}
\end{equation}
Thus, it must be that inequality $\BER_f(\epsilon_1)\le {1\over2}\epsilon_1^2$ is violated for every such $\epsilon_1$.

\end{proof}

\section{Bounds via Area Theorem}
\label{sec:bounds_area_thm}
The lower bound of Theorem \ref{thm:two_point_converse} states that a linear systematic code cannot have small BER for all erasure probabilities. In this sense, it has the flavor of a ``conservation law''. In coding theory, it is often important to understand how a code behaves over a family of parametrized channels. The main existing tool in the literature to study such questions is the so called area theorem. Here we introduce the theorem and study its consequences for two point bounds on BER. It turns out that the bound in Theorem \ref{thm:two_point_converse} is tighter than what can be inferred from the area theorem. 

Following \cite{richardson2008modern}, we define the notion of an {\em extrinsic information transfer} (EXIT) function.
\begin{definition}
Let $X$ be a codeword chosen from an $(n,k)$ code $C$ according to the uniform distribution. Let $Y(\epsilon)$ be obtained by transmitting $X$ through a BEC($\epsilon$). Let $$Y_{\sim i}(\epsilon)=(Y_1(\epsilon),\cdots,Y_{i-1}(\epsilon),?,Y_{i+1}(\epsilon),\cdots,Y_n(\epsilon))$$ be obtained by erasing the $i$-th bit from $Y(\epsilon)$. The $i$-th EXIT function of $C$ is defined as
\[
h_i(\epsilon)=H(X_i|Y_{\sim i}(\epsilon))
\]
The average EXIT function is
\[
h(\epsilon)=\frac{1}{n}\sum_{i=1}^nh_i(\epsilon)
\]
\end{definition}
The area theorem states that
\begin{theorem}[Area Theorem] The average EXIT function of a binary code of rate $R$ satisfies the following property
\[
R=\int_0^1 h(\epsilon)d\epsilon.
\]
\label{thm:area_thm}
\end{theorem}
Let $g$ be a decoder acting on $Y(\epsilon)$. Then the output bit error rate associated to $g$ can be defined as 
\begin{equation}
p^g_b(\epsilon)=\frac{\EE[d(X,g(Y(\epsilon))]}{n}
\label{eq:ber_output}
\end{equation}

where the expectation is taken w.r.t to both the input distribution and channel realizations at erasure probability $\epsilon$. By Proposition \ref{prop:hrank}, the MAP decoder $g^*$ either fully recovers a bit or leaves it completely unbiased. Thus the $i$-th EXIT function can be written as 
\[
H(X_i|Y_{\sim i}(\epsilon))=H(X_i|Y_{\sim i}(\epsilon),g_i^*(Y_{\sim i}(\epsilon)))=\P(g_i^*(Y_{\sim i}(\epsilon))=?).
\]
This gives 
\begin{equation}
p_b^{g^*}(\epsilon)=\frac{1}{2n} \sum_i{\epsilon\P(g_i^*(Y_{\sim i}(\epsilon))=?)}=\frac{\epsilon h(\epsilon)}{2}.
\label{eq:pb_heps}
\end{equation}
Let us now find the implications of the area theorem for the input BER of linear systematic codes. To this end we define the average systematic EXIT function $$h^{\mathrm{sys}}(\epsilon)=\frac{1}{k}\sum_{i=1}^k h_i(\epsilon).$$ Likewise we can define the non-systematic EXIT function as follows:
$$h^{\mathrm{non-sys}}(\epsilon)=\frac{1}{n-k}\sum_{i=k+1}^n h_i(\epsilon).$$
We first prove a lemma to show that the coded bit error rate converges to $0$ continuously as the input bit error rate vanishes.
\begin{lemma}[Data BER vs EXIT function]
Fix $\epsilon<\epsilon_0$. Let $f$ be a binary code of rate $R$.
\begin{enumerate}[(a)]
    \item If $f$ is linear, then
\begin{equation}
h(\epsilon)\le \frac{2 R}{\epsilon_0-\epsilon} \mathrm{BER}_f(\epsilon_0).    
\end{equation}
\item If $f$ is general we have
\[
h(\epsilon)\le \frac{R}{\epsilon_0-\epsilon} h_b(\mathrm{BER}_f(\epsilon_0)),
\]
where $h_b$ is the binary entropy function.
\end{enumerate}
In particular, if BER$(\epsilon_0)\to 0$ for a sequence of binary codes, then $h(\epsilon)\to 0$ for all $\epsilon<\epsilon_0$.
\label{lem:h_nonsys}
\end{lemma}
\begin{proof}

\begin{enumerate}[(a)]
\item Let $X$ be an input codeword $X\in\{0,1\}^n$ and denote by $Y(\epsilon)$ and $Y(\epsilon_0)$ outputs of degraded binary erasure channels, i.e.:
$$ X \to Y(\epsilon) \to Y(\epsilon_0)\,.$$
Notice that
$$ I(X_i; Y(\epsilon_0) | Y_{\sim i}(\epsilon)) = I(X_i; Y_i(\epsilon_0) | Y_{\sim i}(\epsilon)) = (1-\epsilon_0) H(X_i|Y_{\sim i}(\epsilon))\,,$$
where the first equality follows from degradation and the second is a property of erasure channels. Rewriting this identity and summing over $i$ we obtain
\begin{equation}\label{eq:ar_1}
\sum_{i=1}^{n} H(X_i | Y_{\sim i}(\epsilon), Y(\epsilon_0)) = \epsilon_0 \sum_{i=1}^n H(X_i | Y_{\sim i}(\epsilon)) = \epsilon_0 n h(\epsilon)\,,
\end{equation}
where $h(\cdot)$ is an EXIT function of the code $X$.

We now interpret the left-hand side sum in~\eqref{eq:ar_1} as another EXIT function (a conditional one). Indeed, given $Y(\epsilon_0)$ denote by $T_0$ the set of erasures in $Y(\epsilon_0)$. Conditioned on $Y(\epsilon_0)=y_0$ we have that the joint distribution $P_{X,Y(\epsilon)|Y(\epsilon_0)=y_0}$ can be understood as follows: $X_{T_0}$ is sampled from the distribution $P_{X_{T_0}|X_{T_0^c}}$ and then each of the $|T_0|$ entries of $X_{T_0}$ is erased independently with probability $\omega = \frac{\epsilon}{\epsilon_0}$. Denote by $h^0(\omega; y_0)$ the EXIT function of the code $X_{T_0}$ (note that this is a random function, dependent on values of $y_0$ on a set $T_0^c$). This discussion implies 
\begin{equation}\label{eq:ar_2}
   h^0(\omega; y_0) = \frac{1}{|T_0|} \sum_{i=1}^{n} H(X_i | Y_{\sim i}(\epsilon), Y(\epsilon_0) = y_0 )
\end{equation}
(note that terms corresponding to $i \not \in T_0$ are zero.)
From the area theorem and monotonicity of the EXIT function we obtain
\begin{equation}\label{eq:ar_3}
   h^0(\omega; y_0) (1-\omega) \le \frac{1}{|T_0|} H(X|Y(\epsilon_0)=y_0)\,, 
\end{equation}
where the right-hand side is an effective rate of the code. In all, from~\eqref{eq:ar_1}-\eqref{eq:ar_3} we obtain (after taking expectation over $y_0$)
\begin{equation}\label{eq:ar_4}
    n h(\epsilon) \le \frac{1}{\epsilon_0-\epsilon} H(X|Y(\epsilon_0))\,.
\end{equation}
 So far we have not used the fact that the code is binary linear, but now we will. Let $k(T_0) \le nR$ be the number of unrecoverable information bits given a set $T_0$ of erasures. Notice that 
$$ H(X|Y(\epsilon_0)=y_0) \le k(T_0)\,,$$
and thus taking the expectation, we obtain
$$ H(X|Y(\epsilon_0)) \le \EE[k(T_0)] = 2nR\times \mathrm{BER}(\epsilon_0)\,.$$
Together with~\eqref{eq:ar_4} this completes the proof of part (a).
\item Now let $f:S\mapsto X$ be any code. From rate-distortion we have
	$$ I(S; Y(\epsilon_0)) \ge k (1-h_b(\BER_f(\epsilon_0)))\,,$$
	or, equivalently,
	$$ H(X|Y(\epsilon_0)) \le kh_b(\BER_f(\epsilon_0))\,.$$
	The proof is concluded by~(\ref{eq:ar_4}).
\end{enumerate}
\end{proof}
We are now ready to prove the following consequence of the area theorem.
\begin{proposition}
Let $\epsilon_2<\epsilon_1$.  Let $f$ be a binary code of rate $R$ with $\mathrm{BER}(\epsilon_2)\le \delta_2$. Define
\begin{equation}
\zeta(x,\epsilon_2,\epsilon_1)\triangleq \sup_{\{\epsilon_0:\epsilon_0<\epsilon_2\}}\frac{1}{R}\left(\frac{1}{(\epsilon_1-\epsilon_0)}(R-(1-\epsilon_1)-\epsilon_0\frac{x R}{\epsilon_2-\epsilon_0})-1+R\right).
\label{eq:zeta}
\end{equation}
The following hold:
\begin{enumerate}[(a)]

    \item If $f$ is linear systematic then
\[
\mathrm{BER}(\epsilon_1)\ge \frac{\epsilon_1}{2}\zeta(2\delta_2,\epsilon_2,\epsilon_1).
\]
In particular, if $\mathrm{BER}(\epsilon_2)=o(1)$, then
\[
\mathrm{BER}(\epsilon_1)\ge \frac{\epsilon_1}{2R}\left(\frac{1}{(\epsilon_1-\epsilon_2)}(R-(1-\epsilon_1))-1+R\right)+o(1).
\]
    \item If $f$ is systematic (but possibly non-linear), then
    \[
    \BER(\epsilon_1)\ge \epsilon_1h_b^{-1}\left( \zeta(h_b(\delta_2),\epsilon_2,\epsilon_1)\right).
    \]
\end{enumerate}
\label{prop:two_pt_EXIT_area}
\end{proposition}
\begin{proof}
To prove the lower bound on $h(\epsilon_2)$, we may approximate $h(\epsilon_1)$ in a worst-cast fashion as a piece-wise constant function. To do this, note that  $h(\epsilon)\le h(\epsilon_2)$ for all $\epsilon\le \epsilon_2$, and $h(\epsilon)\le h(\epsilon_1)$ for all $\epsilon\in (\epsilon_2,\epsilon_1]$, and $h(\epsilon)\le 1$ for all $\epsilon>\epsilon_1$. Then the area theorem gives that
\[
1-\epsilon_1+h(\epsilon_1)(\epsilon_1-\epsilon_2)+h(\epsilon_2)\epsilon_2\ge R
\]
We note that 
\begin{equation}
h(\epsilon)=Rh^{\mathrm{sys}}(\epsilon)+(1-R)h^{\mathrm{non-sys}}(\epsilon)
\label{eq:EXIT_decomp}
\end{equation}
Using the above two relations, we have
\[
R h^{\mathrm{sys}}(\epsilon_1)\ge \frac{1}{\epsilon_1-\epsilon_2}(R-(1-\epsilon_1)-h(\epsilon_2)\epsilon_2)-(1-R)h^{\mathrm{non-sys}}(\epsilon_1)
\]
Using $h^{\mathrm{non-sys}}\le 1$, we get
\begin{equation}
h^{\mathrm{sys}}(\epsilon_1)\ge \frac{1}{R(\epsilon_1-\epsilon_2)}(R-(1-\epsilon_1)-h(\epsilon_2)\epsilon_2)-(\frac{1}{R}-1).
\label{eq:hsys_general}
\end{equation}
\begin{enumerate}[(a)]
    \item 
    Applying Lemma \ref{lem:h_nonsys}a to bound $h(\epsilon_2)$ we get from~\eqref{eq:hsys_general}
	\[
	h^{\mathrm{sys}}(\epsilon_1)\ge \zeta(2\delta_2,\epsilon_1,\epsilon_2).
	\]
    The bounds on BER follow from noticing that for a linear systematic code
\[
\mathrm{BER}(\epsilon)=\frac{\epsilon h^{\mathrm{sys}}(\epsilon)}{2}.
\]
\item Applying Lemma \ref{lem:h_nonsys}b to bound $h(\epsilon_2)$ in (\ref{eq:hsys_general}) gives
\begin{align}\label{eq:hsys_gen2}
h^{\mathrm{sys}}(\epsilon_1)\ge \zeta(h_b(\delta_2),\epsilon_1,\epsilon_2).
\end{align}
Let $\tilde X_i = \tilde X_i(Y_{\sim i}(\epsilon_1))$ be the MAP decoder of $X_i$ given $Y_{\sim i}(\epsilon_1)$.
From Fano's and Jensen's inequalities we have 
$$ h^{\mathrm{sys}}(\epsilon_1) = {1\over k} \sum_{i=1}^k H(X_i|Y_{\sim i}(\epsilon_1)) \le {1\over k}\sum_{i=1}^k h_b(\PP[X_i \neq \tilde X_i]) \le h_b\left({1\over k}
\sum_{i=1}^k \PP[X_i \neq \tilde X_i]\right)\,.$$
Now, notice that $\PP[X_i \neq \hat X_i] = \epsilon_1 \PP[X_i \neq \tilde X_i]$ and, thus,
$$  h^{\mathrm{sys}}(\epsilon_1) \le h_b({1\over \epsilon_1} \BER(\epsilon_1))\,.$$
The proof is concluded by applying~\eqref{eq:hsys_gen2}.
\end{enumerate}

\end{proof}

\subsection{Comparison of different two-point converse bounds}

We presented four different two-point converse bounds: Prop.~\ref{prop:twopt_gen} applies to all codes,
Prop.~\ref{prop:two_pt_EXIT_area}b applies to all (non-linear) systematic codes, Prop.~\ref{prop:two_pt_EXIT_area}a
and Theorem~\ref{thm:two_point_converse} apply to linear systematic codes. The difference between the latter two is that
Prop.~\ref{prop:two_pt_EXIT_area} is a consequence of the Area Theorem, while Theorem~\ref{thm:two_point_converse} is
new.

In Fig.~\ref{fig:area_vs_KOP} we compare the two-point converse bounds for
general (non-linear) systematic codes (Prop.~\ref{prop:two_pt_EXIT_area}b) and the one for all possible codes
(Proposition \ref{prop:twopt_gen}). We see that depending on the chosen parameters either bound can be a stronger one. 
We also see that once the error at the anchor point vanishes, the systematic codes exhibit a threshold like behavior. We
do not expect a significant difference between general codes and systematic codes at low rates. The plots suggest that
any systematic linear code achieving low BER at rates closed to capacity will not be degrade gracefully and hence is not
suitable for error reduction. However, we cannot rule out the existence of non-systematic nonlinear codes that would
simultaneously be graceful and almost capacity-achieving.

In Fig.~\ref{fig:comparison_area_thm} we compare two-point converse bounds for linear systematic codes (Theorem \ref{thm:two_point_converse} and
Proposition \ref{prop:two_pt_EXIT_area}a). We see that the bounds of Theorem \ref{thm:two_point_converse} are tighter
and much more stable as $\BER$ moves away from $0$.


\begin{figure}[ht]
\centering
\subfloat[Codes with $R=10/11$ and $\mathrm{BER}(0.09)=0$]{\includegraphics[width=0.5\textwidth]{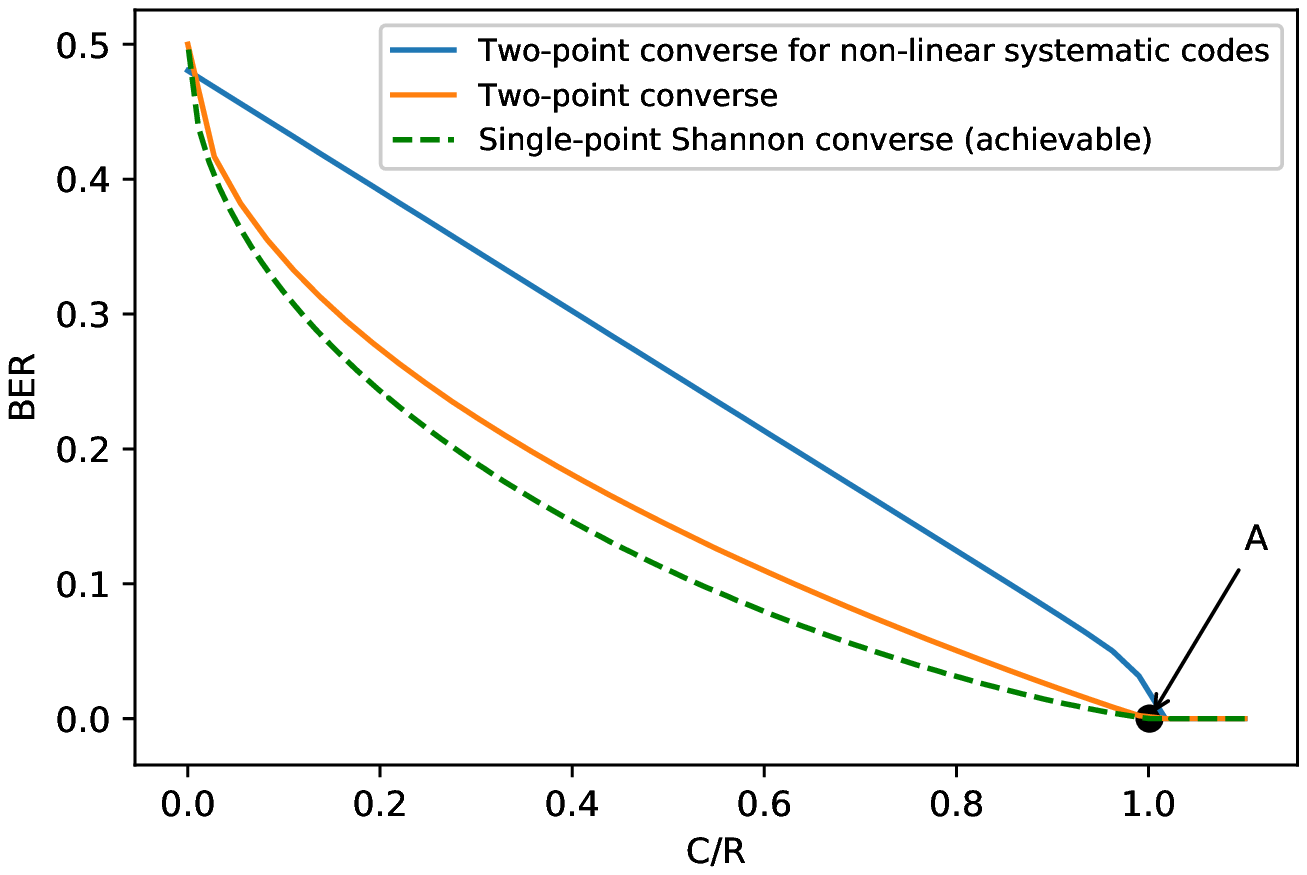}}
\subfloat[Codes with $R=1/20$ and $\mathrm{BER}(0.94995)\le 0.0001$]{\includegraphics[width=0.5\textwidth]{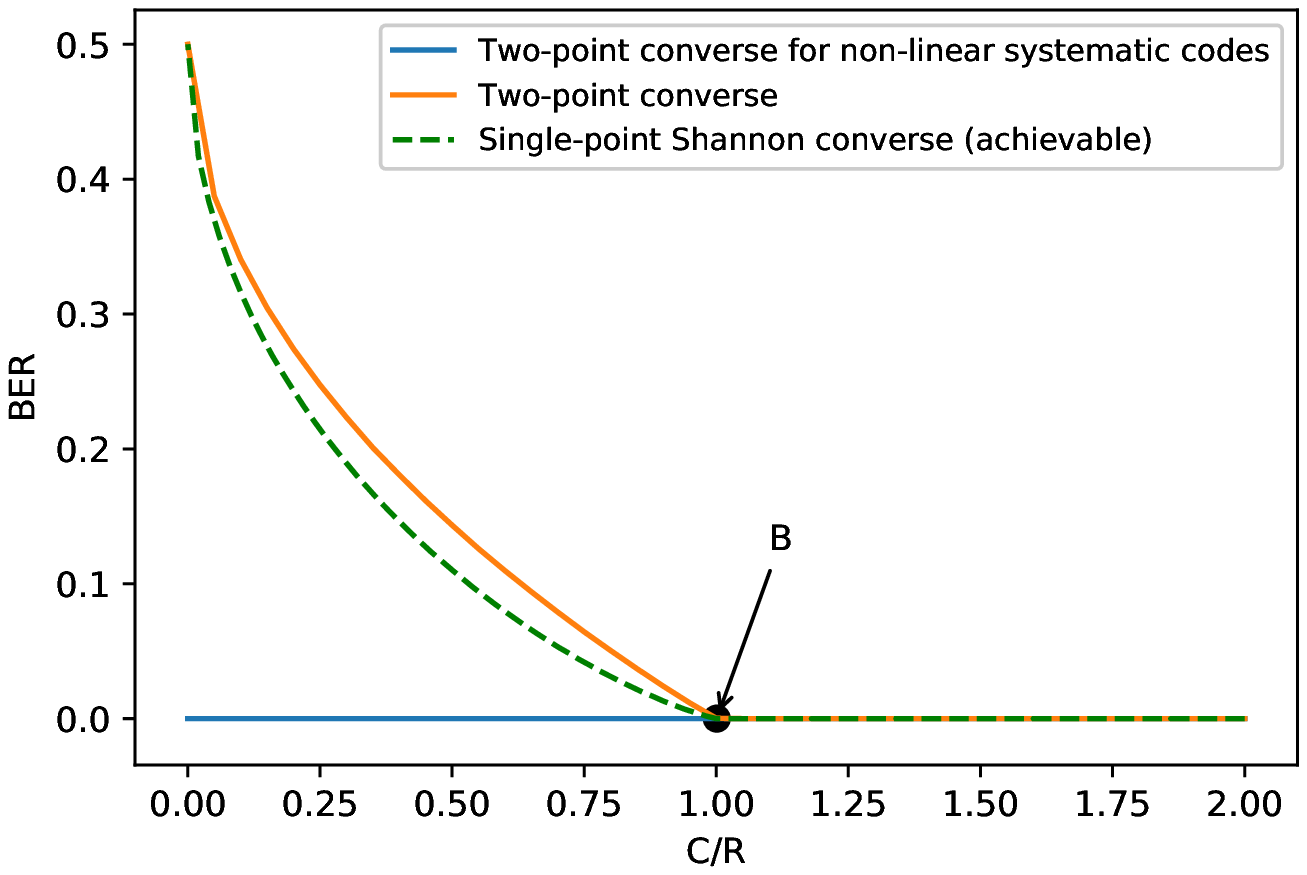}}
\\
\subfloat[Codes with $R=1/20$ and $\mathrm{BER}(0.949995)=0$]{\includegraphics[width=0.5\textwidth]{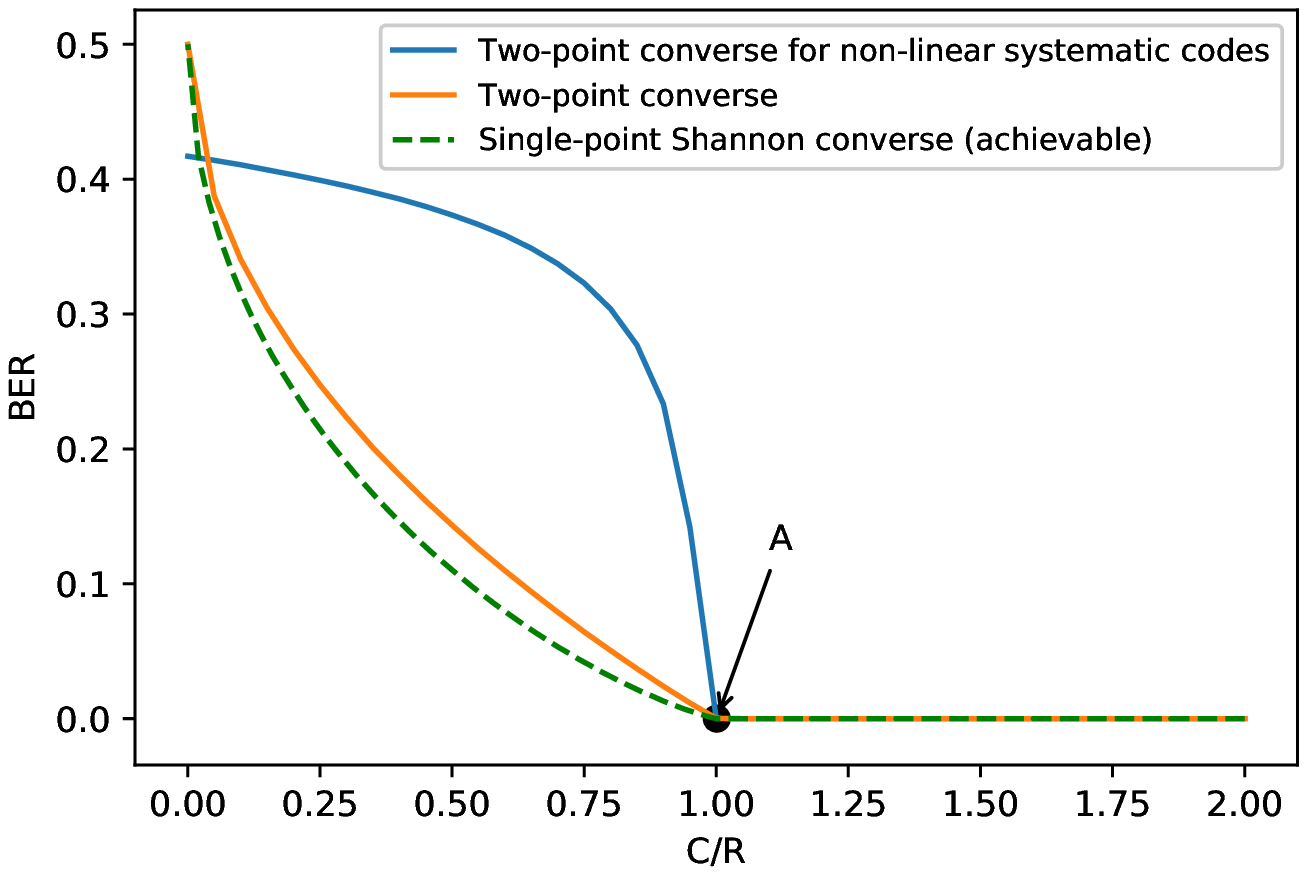}}
\caption{Comparison of the two-point converse for general (Prop.~\ref{prop:twopt_gen}) and systematic codes
(Prop.~\ref{prop:two_pt_EXIT_area}b) of rate $1/2$. Anchor points are $A=(1.001,0.0)$ and $B=(1.001,0.0001)$. }
\label{fig:area_vs_KOP}
\end{figure}

\begin{figure}[ht]
\centering
\subfloat[Lower bounds for codes with $\mathrm{BER}(0.475)=0$]{\includegraphics[width=0.5\textwidth]{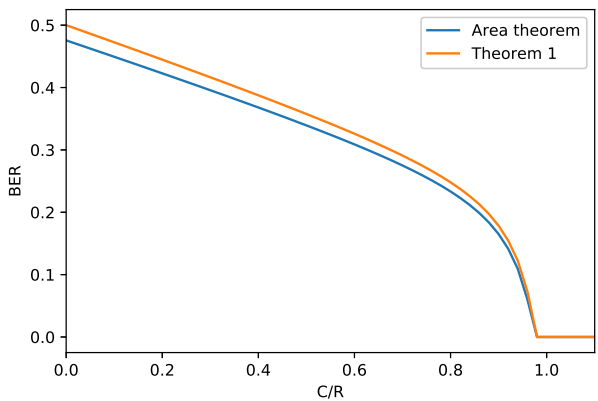}}
\subfloat
[Lower bounds for codes with $\mathrm{BER}(0.495)=0$]{
\includegraphics[width=0.5\textwidth]{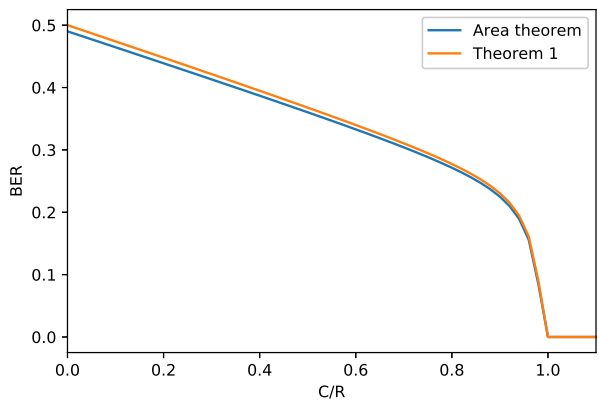}}
\\
\subfloat
[Stability of the bounds around $\mathrm{BER}(0.495)=0$]{
\includegraphics[width=0.5\textwidth]{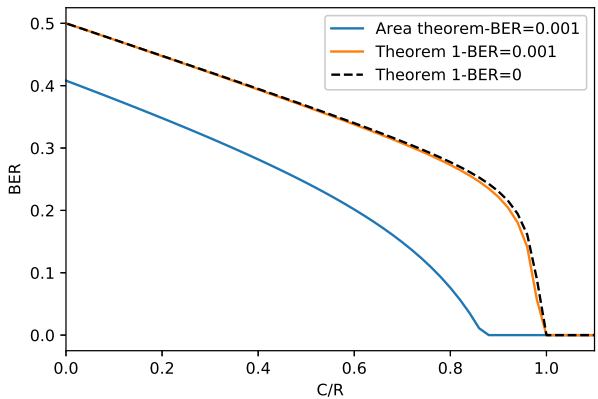}}

\caption{Comparing two-point converse bounds for linear systematic codes (Theorem \ref{thm:two_point_converse} and
Proposition \ref{prop:two_pt_EXIT_area}a). Rate is fixed at $1/2$ satisfying and achor points are a) $\mathrm{BER}=0$ at
$\epsilon=0.475$; b) $\mathrm{BER}=0$ at $\epsilon=0.495$;  c) compares $\BER=0$ vs $\BER=0.001$ constraint at $\epsilon=0.495$. }
\label{fig:comparison_area_thm}
\end{figure}

\section{Channel comparison method for analysis of BP}
\label{sec:channel_comparison}
In this section we provide a new tool for the study of error dynamics under BP for general codes and apply them to provide upper
and lower bounds on the BER of the LDMC codes. Our tools are completely general and apply also to other sparse
graph codes (LDPC, LDGM, LDMC and arbitrary mixtures of these). Furthermore, they will also be adapted (Section~\ref{sec:channel_transform})
to study not only the BER but also the mutual information (for individual bits), which is more relevant if the LDMC
decoder produces soft information for the outer decoder. Finally, any upper bound for the BP implies an upper bound on
the bit-MAP decoding (obviously), but our tools also provide lower bounds on the bit-MAP decoding.

Before going into the details of these tools we review the BP.

\subsection{Review of BP}
We recall the notion of a code ensemble generated by a Boolean function $f:\{0,1\}^m\to \{0,1\}$ from Section \ref{sec:LDMC_construction}. 
We also briefly review the notion of a (bipartite) factor graph associated with a code from the ensemble (cf.\cite{richardson2008modern}, Chapter 2). Consider a code defined on $\{0,1\}^k$. To every coordinate $i\in [k]$, we associate a {\em variable node} and represent it by a circle. We further associate random variables $S_i\stackrel{i.i.d}{\sim}\mathrm{Ber}(1/2)$ to the variable nodes. Likewise, to every subset $\Delta_j\in \Delta$, we associate a {\em check node} and represent it by a square. Every such node represents a constraint of the form $x_{j}=f(S_{\Delta_j})$, where $x_j$'s are the realized (unerased) coded bits and $S_{\Delta_j}$ is the restriction of $S$ to the coordinates in $\Delta_j$. We connect a variable node $i$ to a check node $\Delta_j$ if and only if $i\in \Delta_j$ (see Fig.~\ref{fig:factor_graph}a). We remark that most references associate a separate node with $y_j$'s to model the channel likelihoods \cite{kschischang2001factor,richardson2008modern,wainwright2008graphical}. In the language of \cite{kschischang2001factor}, our description is a cross section of the full factor graph parametrized by $x_j$'s. We do not make this distinction in the sequel as our primary interest is to analyze the decoding error for erasure noise. In this case, we can simply restrict to the sub-graph associated with the observed bits and do not need to consider the channel likelihoods. 

Given a target bit $S_0$, the decoding problem is to estimate (or approximate) the marginal probabilities $p_{S_0|X}(\cdot|x)$ for a realization $x$ of the (observed) coded bits. Here we denote such an estimate by the function $\pi_{S_0}$ and refer to it as a message. A message should be thought of as an approximation to the true marginal computed by the decoder.  To study the behavior of iterative decoding methods, it is helpful to consider the notion of a local neighborhood. Given a target bit $S_i$, we denote by $\Delta(i)$ the set of its neighbor nodes among the factors, that is, the set of check nodes whose constraint involves $S_i$. We further define the local neighborhood $\partial(i)$ among the variable nodes to be the set of variables (other than $i$) that appear in $\Delta(i)$ (see Fig.~\ref{fig:factor_graph}b). Given a vector $v\in \{0,1\}^k$, we define $\partial v_i:=v_{\partial(i)}$ to be the restriction of $v$ to the coordinates in $\partial(i)$. Likewise, if $v\in\{0,1\}^n$, then $\Delta v_{i}:=v_{\Delta(i)}$ denotes the restriction of $v$ to $\Delta(i)$. The  $j$-th node in $\Delta(i)$ is denoted by  $\Delta_j(i)$. The variable nodes other than $i$ that are connected to $\Delta_j$ are denoted by $\partial_j(i)$. Similarly, we define the $j$-th order local neighborhood $\partial^j(i)$ by recursively unfolding the local neighborhoods at the boundary $\partial^{j}(i):=\partial(\partial^{j-1}(i))-\partial^{j-1}(i)$. In other words, the $\ell$-th order boundary is the set of nodes (not in $\partial^{j-1}(i)$) that are in the local neighborhood of $\partial^{j-1}(i)$. Likewise, $\Delta^j(i):=\Delta(\partial^{j-1}(i))-\Delta^{j-1}(i)$ (see Fig.~\ref{fig:comp_trees}a). The compliment of $\Delta(i)$ inside $\Delta$ is denoted by $\Delta^\sim(i)$. Finally, we define $\Delta^{(j)}(0):=\cup_{i=1}^j\Delta^i(0)$. 

With this notation we can describe a generic iterative algorithm to compute $\pi_{S_i}$. Let $\pi_{\partial S_0}$ be the message (or approximation) for $p_{\partial S_0|{\Delta^\sim X_0}}$. This is the conditional estimate of the random variables in the local neighborhood of $S_0$ given all the observed bits outside the neighborhood. By graph separation, the computation of the marginals for $S_0$ can be decomposed as follows:
\begin{equation}
\pi_{S_0}(s_0)= \sum_{\partial s_0}\pi_{\partial S_0}(\partial s_0)p_{S_0|\partial s_0}(s_0|\partial s_0)\propto \sum_{\partial x_0} \pi_{\partial  S_0}(\partial s_0)\prod_{\Delta_j\in \Delta(0)}\ind_{\{y_j=f(x_0,\partial_j x_{0})\}}.
\label{eq:decomposed_posterior}
\end{equation}
 In this way, we obtain an iterative procedure where the messages $\pi_{\partial S_0}$ flow into the local neighborhood and the posterior estimates $\pi_{S_0}$ flow out to the target node (see Fig.~\ref{fig:factor_graph}b). To iterate such a procedure $\ell$-times, one needs to first approximate the marginals at the $\ell$-th order boundary. Once this is done, (\ref{eq:decomposed_posterior}) can be applied iteratively to compute $\pi_{S_0}$. The factor (sub-)graph obtained after $\ell$ unfoldings represents the natural order of recursive computations needed to compute $\pi_{S_0}$, and hence, we refer to it as a {\em computational graph}. Fig.~\ref{fig:comp_trees} shows the case where the computational graph is a tree.  

Belief propagation (BP) is a special case of such iterative procedure where the input messages are assumed to factorize into a product
\[
\pi_{\partial S_0}=\prod_{i}\pi_{\partial_i S_0}.
\]

 The number of iterations of BP determine the depth of the computational graph, i.e., the order of the local neighborhood on which we condition. We denote by $\pi^\ell$ the message corresponding to $p_{S_0|\Delta^{(\ell)}X_0}$. This is the approximate marginal given observations revealed in the computational graph of depth $\ell$. After $\ell$ iterations, the marginals under BP can be written 
more efficiently (compared with  (\ref{eq:decomposed_posterior})) as

\begin{equation}
\pi^{\ell}_{S_0}(s_0)\propto  \prod_{\Delta_j\in\Delta(0)} \sum_{\partial_j s_0} \pi^{\ell-1}_{\partial_j S_0}(\partial_j s_0)\ind_{\{x_j=f(s_0,\partial_j s_{0})\}},
\label{eq:decomposed_posterior_BP}
\end{equation}
with the initial conditions $\pi^0_{S_i}(0)=\pi^0_{S_i}(1)=1/2$ for all bits.

It can be checked that when the computational graph is a tree, BP is exact, i.e., it computes the correct marginals $p_{S_0|\Delta^{(\ell)}X}$ given the observations in the depth $\ell$ graph. We also refer to the correct marginal $p_{S_0|X}$ as the (bitwise) MAP estimate of $S_0$. When the computational graph is a tree, the only difference between MAP and BP estimates is the input messages into the $\ell$-th order local neighborhood. In other words, if the initial messages along the boundary are the correct marginals $p_{\partial^{l}S_0|\Delta^{\sim (\ell)}X_0}$, then BP iterations recover the (bitwise) MAP estimate. Here $\Delta^{\sim (\ell)}$ is the set of check nodes in $\Delta$ that do no appear in the computational tree of depth $\ell$.

\subsection{Channel comparison method}\label{sec:channel_comp_method}

The key problem in analyzing the BP algorithm is that the incoming messages after a few iterations have a very
complicated distribution. Our resolution is to apply channel comparison methods from information theory to replace these
complicated distributions with simpler ones on each iteration. This way BP operation only ever acts on a simple
distribution and hence we can apply single-step contraction analysis to figure out asymptotic convegence. (This is
reminiscent of ``extremes of information combining'' method in the LDPC literature, cf.~\cite[Theorem 4.141]{richardson2008modern} -- see
Remark~\ref{rmk:extinfo} for discussion.) 

We start with reviewing some key information-theoretic notions. 

\begin{definition}[{~\cite[\S 5.6]{el2011network}}]
Given two channels $P_{Y|X}$ and $P_{Y\sp{\prime}|X}$ with common input alphabet, we say that $P_{Y\sp{\prime}|X}$ is
\begin{itemize}
    \item {\em less noisy} than $P_{Y|X}$, denoted by $P_{Y|X} \preceq_\mathrm{l.n.} P_{Y\sp{\prime}
|X}$, if for all joint distributions $P_{UX}$ we have
\[
I(U;Y)\le I(U;Y')
\]
\item {\em more capable} than $P_{Y|X}$, denoted by $P_{Y|X}\preceq_{\mathrm{m.c.}} P_{Y\sp{\prime}
|X}$, if for all marginal distributions $P_{X}$ we
have
$$I(X; Y ) \le I(X; Y\sp{\prime}).$$ 
\item {\em less degraded} than $P_{Y|X}$, denoted by $P_{Y|X}\preceq_{\mathrm{deg}}P_{Y\sp{\prime}|X}$, if there exists a Markov chain $Y-Y'-X$.
\end{itemize}
\end{definition}
We refer to 
~\cite[Sections I.B, II.A]{makur2018comparison} and~\cite[Section 6]{polyanskiy2017strong} for alternative useful characterizations of the less-noisy order. In particular, it is known (cf. ~\cite[Proposition 14]{polyanskiy2017strong},\cite{korner1977comparison}) that
\begin{equation}
P_{Y|X} \preceq_\mathrm{l.n.} P_{Y\sp{\prime}
|X} \iff D(P_{Y}\| Q_Y )\le D(P_{Y\sp{\prime}}\| Q_{Y\sp{\prime}})
\label{eq:less_noisy_kl}
\end{equation}
where the output distributions correspond to common priors $P_X,Q_X$. The following implications are easy to check  
\[
P_{Y|X}\le_{\mathrm{deg}}P_{Y\sp{\prime}|X}\implies P_{Y|X}\le_{\mathrm{l.n.}}P_{Y\sp{\prime}|X}\implies P_{Y|X}\le_{\mathrm{m.c.}}P_{Y\sp{\prime}|X}. 
\]
Counter examples for reverse implications are well known,~\cite[Problem 15.11]{csiszar2011information}, and even
possible in the class of BMS channels as follows from Lemma~\ref{lem:lemma_A} below. Nevertheless, we give another
example which will be instructive for the later discussion after Lemma~\ref{lem:lemma_B}.
\begin{example}
Fix $\delta \in (0,1/2)$ and $\epsilon < h_b(\delta)$. 
By Lemma \ref{lem:lemma_A} (below),  $\BEC_{\epsilon}$ is more capable than $\BSC_\delta$. Let $X_1,X_2$ be some
binary random variables (not necessarily independent or unbiased) and $\tilde Y_i=\BEC_{\epsilon}(X_i)$, $Y_i=\BSC_{\delta}(X_i)$ be their observations. By
~\cite[Problem 6.18]{csiszar2011information}, the property of being more capable tensorizes,
implying that
\begin{equation}\label{eq:bcc_1}
	I(X_1, X_2; \tilde Y_1, \tilde Y_2) > I(X_1, X_2; Y_1, Y_2)\,.
\end{equation}
Somewhat counter-intuitively, however, there exists functions of $(X_1,X_2)$ for which this inequality is reversed
(demonstrating thus, that $\BEC_\epsilon \not \succeq_{l.n.}\BSC_\delta$). Indeed, consider $X_1 \dperp X_2 \sim \Ber(1/2)$ and $U=\mathrm{XOR}(X_1,X_2)$. Then a simple computation shows
$$ I(U;Y_1,Y_2)=1-h_b(2\delta(1-\delta)), \qquad I(U; Y_1, Y_2) = (1-\epsilon)^2\,.$$
Since $h_b(\delta) > 1-\sqrt{1-h_b(2\delta(1-\delta))}$ by taking $\epsilon$ in between these two quantities we
ensure that both~\eqref{eq:bcc_1} and the following hold
	$$ I(U; \tilde Y_1, \tilde Y_2) < I(U; Y_1, Y_2)\,.$$
\label{ex:counter_ex}
\end{example}

\begin{definition}[{BMS ~\cite[\S 4.1]{richardson2008modern}}] Let $W$ be a memoryless channel with binary input
alphabet $\mathcal{X}$ and output alphabet $\mathcal{Y}$. We say that $W$ is a binary memoryless symmetric channel
($\BMS$) if there exists a measurable involution of $\mathcal{Y}$, denoted $-$, such that $W(y|0)=W(-y|1)$ for all $y\in
\mathcal{Y}$.
\end{definition}

We also define the total variation distance (TV) and  $\chi^2$-divergence between two probability measures $P$ and $Q$ as follows:
\begin{align*}
\TV(P,Q)&\triangleq \frac{1}{2}\int |dP-dQ|,\\
\chi^2(P\| Q)&\triangleq\int (\frac{dP}{dQ})^2dQ-1.
\end{align*}
For an arbitrary pair of random variables we define
$$ I_{\chi^2}(X;Y) = \chi^2(P_{X,Y} \| P_X \otimes P_Y)\,,$$
where $P_X \otimes P_Y$ denotes the joint distribution on $(X,Y)$ under which they are independent.

Let $W$ be a $\BMS$ channel, $X\sim \Ber(1/2)$ and $Y=W(X)$ be the output induced by $X$. We
define $W$'s probability of error, capacity, and $\chi^2$-capacity as follows
	\begin{align}
		P_e(W) &= \frac{1-\TV(W(\cdot|0), W(\cdot|1))}{ 2},\\
	   C(W) &= I(X;Y)\\
	   C_{\chi^2}(W) &= I_{\chi^2}(X;Y)\,.
	\end{align}
The reason for naming $C$ and $C_{\chi^2}$ capacity is because of their extremality properties\footnote{Both statements
can be shown by first explicitly checking the case of $W=\BSC_\delta$ and then by representing a $\BMS$ as a mixture of
BSCs.}
	$$ C = \sup_{P_X} I(X;Y)\,, \quad C_{\chi^2} = \sup_{P_X} I_{\chi^2}(X;Y)\,.$$
For the BMS channels there is the following comparison ($C(W)$ is measured with $\log$ units to arbitrary base):
	$$ {\log e\over 2} C_{\chi^2}(W) \le C(W) \le C_{\chi^2}(W) \log 2\,. $$
We also mention that for any pair of channels $W\le_{ln}W'$ we always have $I(X;W(X))\le I(X;W'(X))$ and
$I_{\chi^2}(X;W(X)) \le I_{\chi^2})(X; W'(X))$, and in particular $C(W)\le C(W')$ and $C_{\chi^2}(W)\le C_{\chi^2}(W')$.

\begin{lemma}\label{lem:lemma_A}
  The following holds:
  \begin{enumerate}
  \item Among all $\BMS$ channels with the same value of $P_e(W)$ the least
  degraded is $\BEC$ and the most degraded is $\BSC$, i.e.
\begin{equation}\label{eq:la_deg}
    	\BSC_{\delta} \preceq_{deg} W \preceq_{deg} \BEC_{2\delta}\,,
\end{equation}    
  where $\preceq_{deg}$ denotes the (output) degradation order. 

  \item Among all $\BMS$ with the same capacity $C$ the most capable is $\BEC$ and the least capable is
  $\BSC$, i.e.:
\begin{equation}\label{eq:la_mc}
    	\BSC_{1-h_b^{-1}(C)} \preceq_{mc} W \preceq_{mc}  \BEC_{1-C}\,,
\end{equation}    
  where $\preceq_{mc}$ denotes the more-capable order, and $h_b^{-1}:[0,1]\to[0,1/2]$ is the functional inverse of the (base-2) binary entropy function $h_b:[0,1/2]\to[0,1]$.

  \item Among all $\BMS$ channels with the same value of $\chi^2$-capacity $\eta=I_{\chi^2}(W)$ 
  the least noisy is $\BEC$ and the most noisy is $\BSC$, i.e.
\begin{equation}\label{eq:la_ln}
  	\BSC_{1/2-\sqrt{\eta}/2} \preceq_{ln} W \preceq_{ln} \BEC_{1-\eta}\,, 
\end{equation}  
  where $\preceq_{ln}$ denotes the less-noisy order.
  \end{enumerate}
\end{lemma}

The next lemma states that if the incoming messages to BP are comparable, then the output messages are comparable as well. 
\begin{lemma}\label{lem:lemma_B}
Fix some random transformation $P_{Y|X_0,X_1^m}$ and $m$
  $\BMS$ channels $W_1,...,W_m$. Let $W: X_0\mapsto (Y,Y_1^m)$ be a (possibly non-$\BMS$) channel 
  defined as follows. First, $X_1,..., X_m$ are generated as i.i.d $\Ber(1/2)$. Second, 
  each $Y_j$ is generated as an observation of $X_j$ over the $W_j$, i.e. $Y_j=W_j(X_j)$ (observations are all
  conditionally independent given $X_1^m$). Finally, $Y$
  is generated from all $X_0,X_1^m$ via $P_{Y|X,X_1^m}$ (conditionally independent of $Y_1^m$ given $X_1^m$). Define the $\tilde W$ channel similarly, but with $W_j$'s replaced with $\tilde W_j$'s. The following statements hold:
  \begin{enumerate}
   \item If $\tilde W_j \preceq_{deg} W_j$ then $\tilde W \preceq_{deg} W$
   \item If $\tilde W_j \preceq_{ln} W_j$ then $\tilde W \preceq_{ln} W$
   \end{enumerate}
\end{lemma}
\begin{remark}
An analogous statement for more capable channels does not hold. To see this, let $Y=X_0+X_1+X_2 \mod 2$.
Then the channel $X_0\mapsto (Y,Y_1,Y_2)$ is equivalent to $U\mapsto (Y_1,Y_2)$ in the setting of example
\ref{ex:counter_ex}, thus implying that $I(X_0; Y, Y_1,Y_2)$ decreases while replacing $W_j$ with more capable channels.
We pose the following \textit{open question:} In the setup of Lemma \ref{lem:lemma_B}, assume further that
$Y=f(X_0,X_1^m)$ where $f:\{0,1\}^{m+1}\to \{0,1\}^r$. For what class of functions can the above lemma be extended to
more capable channels? Fig.~\ref{fig:soft_info_bounds} suggests that it holds for the case where each of the $r$
components of $f$ is a majority.
\label{remark:mc_vs_ln}
\end{remark}

The lemmas are proved in Appendix~\ref{apx:complem}. Equipped with them we can prove rigorous upper/lower bounds on the
BER and mutual information -- this will be executed in Propositions~\ref{prop:approx_gives_lower_bound} and~\ref{prop:approx_gives_lower_bound_soft} below. Here we wanted to
pause and discuss several issues pertaining to the generality of the method contained in
Lemmas~\ref{lem:lemma_A}-\ref{lem:lemma_B}.

\begin{remark}[LDPC as LDGM] In the formulation of Lemma~\ref{lem:lemma_B} we enforced that ``source'' bits
$X_0,X_1,\ldots,X_m$ be independent. This is perfectly reasonable for the analysis of the LDMC and LDGM codes, for which
we take $Y$ as a noisy observation of the majority, or XOR of $X_0,\ldots,X_m$. 
However, for the LDPC codes the ``source'' bits are not independent -- rather they are
chosen uniformly among all solutions to a parity-check $X_0+X_1+\cdots+X_m = 0 \mod 2$. This, however, can be easily
modeled by taking $X_0,X_1,\ldots,X_m$ independent $\Ber(1/2)$ and defining $Y=X_0+X_1+\cdots+X_m \mod 2$. In other
words, codes such as LDPC which restrict the possible input vectors can be modeled as LDGM codes with noiseless
observations of some outputs. In this way, our method applies to codes defined by sparse constraints.

Similarly, although Lemma~\ref{lem:lemma_B} phrased in terms of a ``factor node'', it can also be applied to a
``variable node'' processing by modeling the constraint $X_1=\cdots=X_m=X_0$ as a sequence of noiseless observations
$X_1+X_0, X_m+X_0$.
\end{remark}

\begin{remark}[Relation to ``extremes of information combining'']\label{rmk:extinfo} In the original LDPC literature it
was understood that the analysis of the BP performance of sparse linear codes (LDGM/LDPC) can be done rigorously
by a method known as density evolution. However, analytically implementing the method is only feasible over the erasure
channel, leading to a notorious open question of whether LDPCs can attain capacity of the BSC. However, several
techniques were invented for handling this difficulty. All of these are based on replacing the complicated $\BMS$
produced after a single iteration of BP by a BEC/BSC with the same value of some information parameter. The key
difference with our results is that all of the previous work focused on a very \textit{special case} of
Lemma~\ref{lem:lemma_B} where observation $Y$ is obtained as a noisy (or noiseless, see previous remark) observation of
the sum $X_0+\cdots+X_m \mod 2$. As such, those methods do not apply to LDMCs. 

Specifically, the first such method~\cite[Theorem 4.2]{khandekar2003graph} considered parameter known as Bhattacharya distance $Z(W) =
\sum_y \sqrt{W(y|0)W(y|1)}$. Then it can be shown that a (special case of) of Lemma~\ref{lem:lemma_B} holds for it --
see~\cite[Problem 4.62]{richardson2008modern}. Namely, unknown BMSs replaced with $Z$-matched BEC (BSC) result in a
channel with a worse (better) value of $Z$. Note the reversal of the roles of BEC and BSC compared to our comparison --
this is discussed further in the following remark.

Next, the more natural choice of information parameter is the capacity, $C(W)$. Here again one can prove that BEC/BSC
serve as the worst and best channels (in terms of capacity), however their roles are reversed at the variable and (linear) check nodes --
see~\cite[Theorem 4.141]{richardson2008modern}. 

In all, unlike $C_{\chi^2}$-based comparison proposed by us, neither the $Z(W)$ nor the $C(W)$-based comparisons hold 
in full generality of Lemma~\ref{lem:lemma_B} and only apply to linear codes. Furthermore, in a subsequent paper we will
show that the $C_{\chi^2}$-based comparison in fact yields stronger results in many problem.
\end{remark}

\begin{remark}[Weak universal upper bounds on LDPC thresholds] The method of $Z(W)$-comparison~\cite[Theorem
4.2]{khandekar2003graph} allows one to make the
following statement. Consider an infinite ensemble of irregular LDPCs (or IRA or any other linear, locally tree-like
graph codes) and let $\epsilon^*_{\mathrm{BEC}}$ and $\delta^*_{\mathrm{BSC}}$ be their BP decoding thresholds (that is,
the ensemble achieves vanishing error under BP decoding over any
$\BEC_\epsilon$ with $\epsilon < \epsilon^*_{\mathrm{BEC}}$ but not for any $\epsilon>\epsilon^*_{\mathrm{BEC}}$; and
similarly for the $\BSC$). Then we always have
	$$ \delta^*_{\mathrm{BSC}} \ge {1-\sqrt{1-(\epsilon^*_{\mathrm{BEC}})^2}\over 2}\,.$$
This bound is universal in the sense that it is a special case of the more general result that BP error converges to
zero on any $\BMS$ $W$ with $Z(W)< Z(\BEC_{\epsilon^*_{\mathrm{BEC}}})$. 

For the $(3,6)$ regular LDPC we have $\epsilon^*_{\mathrm{BEC}} \approx 0.4294$ so the previous bound yields
$\delta^*_{\mathrm{BSC}} \ge 0.0484$. One can also execute the universal $C(W)$-comparison (although it requires knowledge of full
degree distribution rather than only $\epsilon^*_{\mathrm{BEC}}$), which yields, as an example, 
that for the $(3,6)$ code $\delta^*_{\mathrm{BSC}} \ge 0.0484$ -- see~\cite[\S 4.10]{richardson2008modern}. 

However, it has been observed that no universal \textit{upper bounds} are usually possible via these classical methods, cf.~\cite[Problem
4.57]{richardson2008modern}. Similarly, our Lemmas~\ref{lem:lemma_A}-\ref{lem:lemma_B} immediately imply that the BP error does not
converge to zero over any $\BMS$ $W$ with $C_{\chi^2}(W)<C_{\chi^2}(\BEC_{\epsilon^*_{\mathrm{BEC}}})$. This yields a
universal upper bound on the BP threshold of general ensembles, and for the BSC implies
\begin{equation}\label{eq:univ_ub}
		\delta^*_{\mathrm{BSC}} \le {1-\sqrt{1-\epsilon^*_{\mathrm{BEC}}}\over 2}\,.
\end{equation}
Unfortunately, for the regular $(3,6)$ ensemble this evaluates to $\delta^*_{\mathrm{BSC}} \le 0.1223$, which is worse
than the trivial bound $\delta^*_{\mathrm{BSC}} \le 1-h^{-1}(R)$, where $R$ is the rate of the code ($R=1/2$ for
the $(3,6)$ example). We may conclude that bound~\eqref{eq:univ_ub} is only ever useful for those ensembles whose BEC
BP-threshold $\epsilon^*_{\mathrm{BEC}}$ is very far from $1-R$, i.e. for codes with a large gap to (BEC) capacity.
\end{remark}

\subsection{E-functions}
We recall that, in general, a computational graph of small depth ($o(\log(k))$) corresponding to a (check-regular) code ensemble is with high probability a tree (cf. \cite{richardson2008modern}, Exercise 3.25). For such ensembles, we want to study the dynamics of the decoding error along the iterations of BP.  Hence, we need to understand how the error flows in and out of the local neighborhood of a target node. In other words, we want to understand how the BP dynamics contract the input error.  

The notions of E-functions are useful for this purpose. They can be viewed as a mapping of the input error density at the leaf nodes (in the beginning of a decoding iteration) to the output error density at the target node (at the end of the iteration). There are two types of E-functions studied in this work: the erasure functions and the error functions. 

\begin{figure*}[!t]
\centering
\subfloat[Factor graph representation of the (observed) equations]{\includegraphics[width=0.65\textwidth]{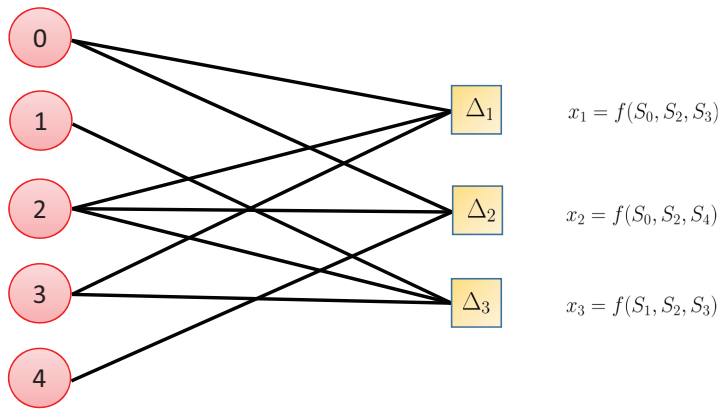}}
\label{fig:factor_graph_a}
\hfil
\subfloat
[Local neighborhood of $0$ in the unfolded factor graph]{
\includegraphics[width=0.7\textwidth]{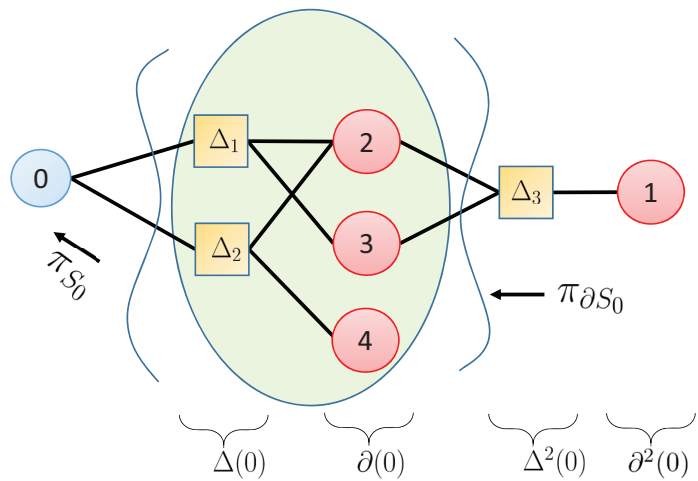}}
\label{fig:factor_graph_b}
\caption{The factor graph of a code and the local neighborhood of a target node are shown. a) The check nodes correspond to observed (unerased) coded bits and represent the constraints imposed by such observations. b) The factor graph can be unfolded with respect to a target node. The immediate (variable) neighbors of the target nodes in such unfolding form its local neighborhood. A recursive algorithm can first estimate the marginal probabilities $\pi_{\partial S_0}$ for the local neighborhood and then compute the posterior $\pi_{S_0}$ using (\ref{eq:decomposed_posterior}). Here we recall that $\partial {S_0}=S_{\partial(0)}$. }
\label{fig:factor_graph}
\end{figure*}

\begin{figure*}[!t]
\centering
\subfloat[Computational tree of depth $\ell$ for a (check-regular) ensemble of degree 3]{\includegraphics[width=\textwidth]{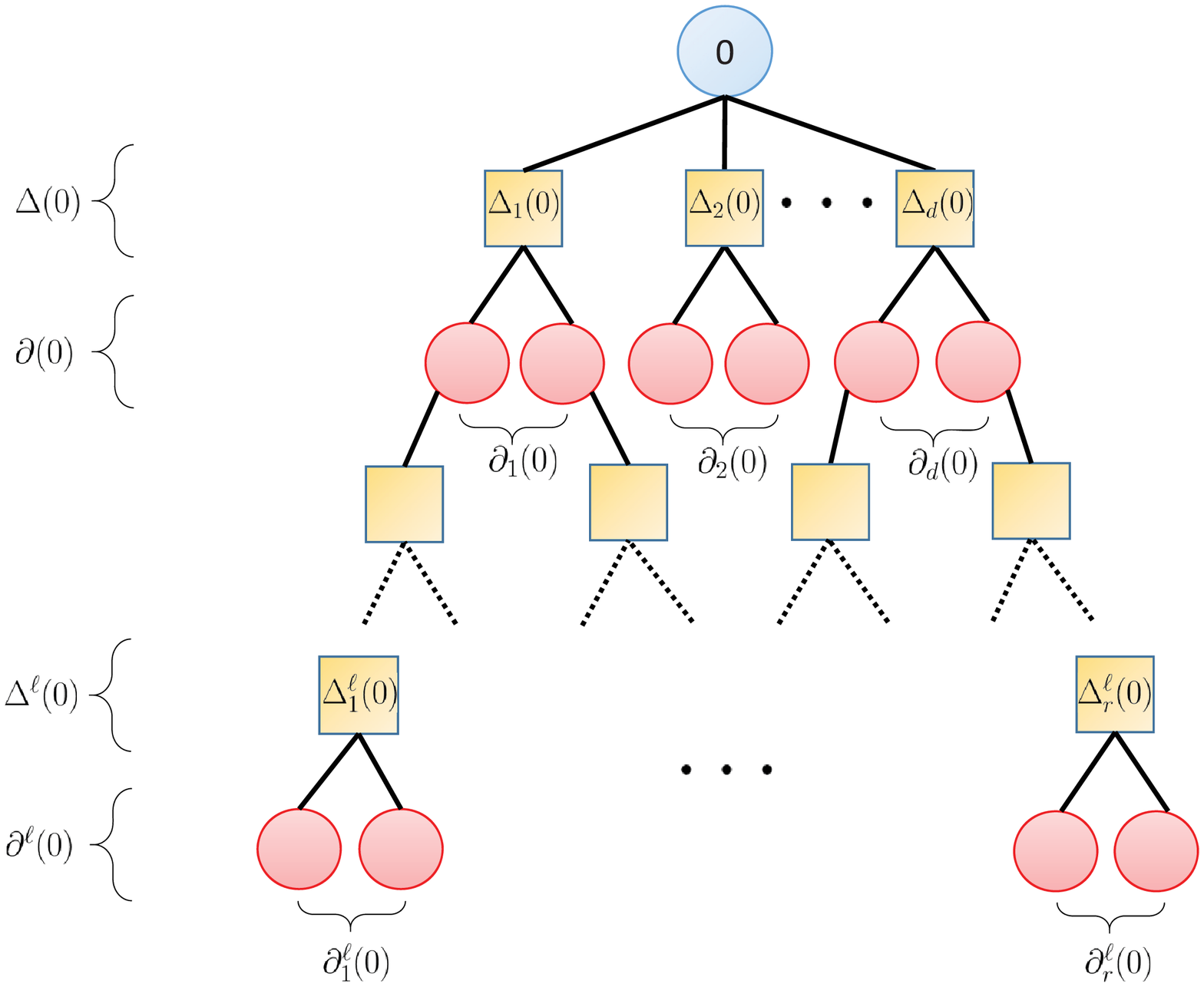}}
\label{fig:comp_trees_a}
\hfil
\subfloat
[Local neighborhood of a variable node with BEC inputs]{
\includegraphics[width=\textwidth]{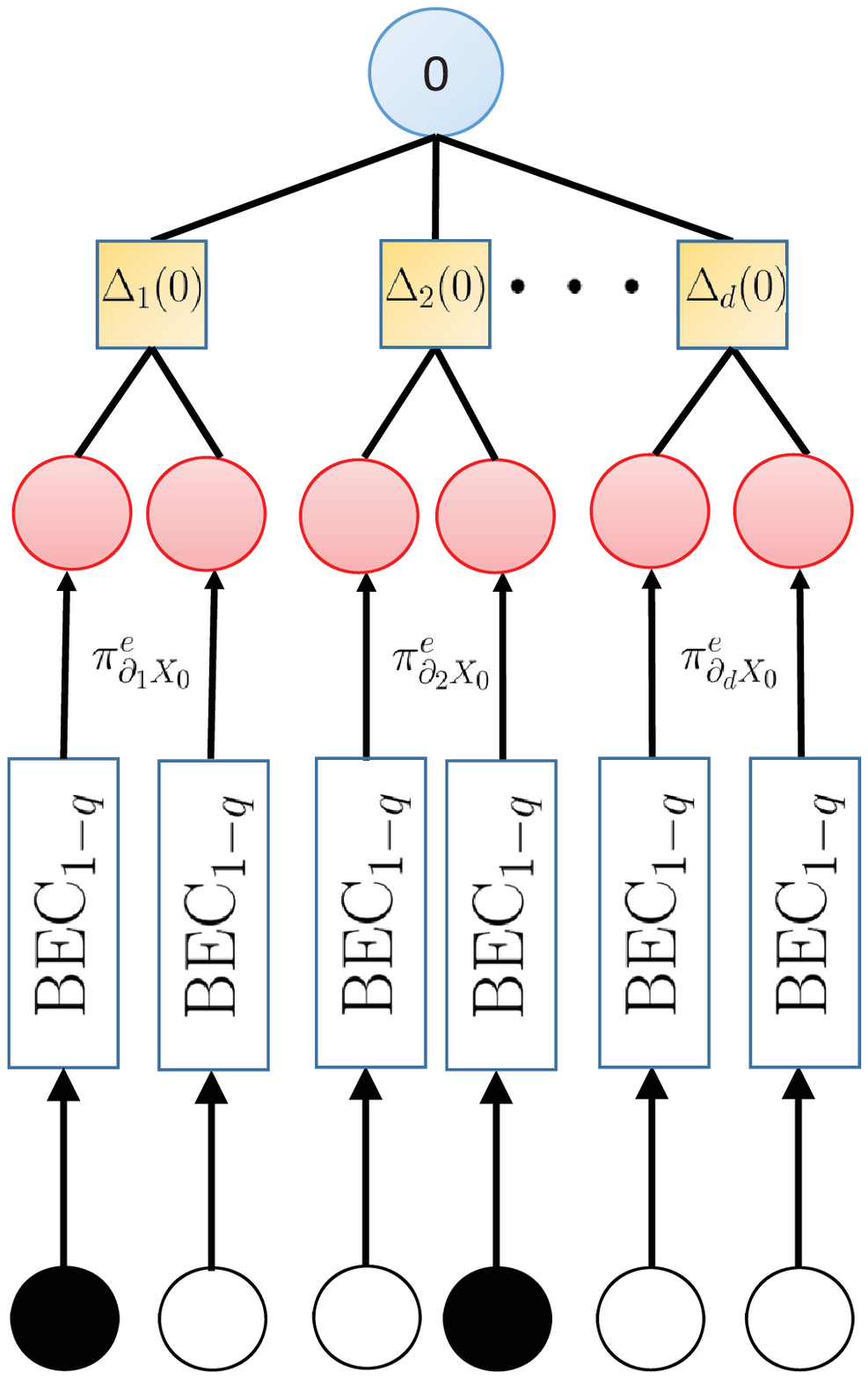}}
\label{fig:comp_trees_b}
\caption{ a) A computational tree for a (check-regular) ensemble of degree 3 obtained after $\ell$ unfoldings w.r.t a target node along with the (local) indexing used in the analysis of BP. We refer to $\partial^j(0)$ as the $j$-th order neighborhood of $0$ and $\partial^{\ell}(0)$ as the boundary of the tree. b) The local neighborhood of the target node with leaves observed through BEC channels. This local graph is used to define the erasure function. }
\label{fig:comp_trees}
\end{figure*}

\begin{definition}[Erasure function]
Consider a code ensemble generated by a Boolean function $f:\{0,1\}^m\to\{0,1\}$ with variable node degrees sampled from $\mathrm{Deg}$. Fix $\alpha=C/R$ and consider a computational tree of depth 1 as in Fig.~\ref{fig:comp_trees}b corresponding to the target bit $X_0$. Let $M_j=f(S_0,S^{(j)})$, $j=1,\cdots,d$, where $S^{(j)}\stackrel{}{\sim}\mathrm{Ber}(1/2)^{\otimes(m-1)}$ are the boundary nodes. Suppose that each boundary node is observed through a (memoryless) $\BEC$ channel, i.e., $Y^{(j)}=\BEC_{\bar{q}}(S^{(j)})$ where $\bar{q}=1-q$ is the probability of error. The function $$E^{\BEC}_d(q)\triangleq\EE [\P(S_0=1|M_1,\cdots,M_d,Y^{(1)},\cdots,Y^{(d)})|S_0=0].$$ is called the $d$-th erasure polynomial of the ensemble. Here the expectation is taken with respect to the  randomization over bits as well as the noise in the observations. The erasure function is defined as the expectation of $E_d^\BEC$ (over the ensemble):
\[
E^{\BEC}(\alpha,q)\triangleq\sum_{k} \P(\mathrm{Deg}=k)E^\BEC_k(q).
\]
The $d$-th truncated easure polynomial is
\[
E^{\BEC}_{\le d}(\alpha,q)\triangleq\sum_{k\le d} \P(\mathrm{Deg}=k)E^\BEC_k(q).
\]
\label{def:erasure_function}
\end{definition}
Similarly, we can define the notion of an error function.
\begin{definition}[Error function]
In the setup of \textup{Definition \ref{def:erasure_function}}, let $Y^{(j)}=\BSC_{q}(S^{(j)})$ be the result of passing $S^{(j)}$ through a (memoryless) $\BSC$ channel with crossover probability $q$. The function $$E^{\BSC}_d(q)\triangleq\EE [\P(S_0=1|M_1,\cdots,M_d,Y^{(1)},\cdots,Y^{(d)})|S_0=0]$$ is called the $d$-th error polynomial of the ensemble. Likewise, the error function is defined as
\[
E^{\BSC}(\alpha,q)\triangleq\sum_{k} \P(\mathrm{Deg}=k)E^\BSC_k(q).
\]
The $d$-th truncated error polynomial is
\[
E^{\BSC}_{\le d}(\alpha,q)\triangleq\sum_{k\le d} \P(\mathrm{Deg}=k)E^\BSC_k(q)+\frac{1}{2} \P(\mathrm{Deg}>d).
\]
\label{def:error_function}
\end{definition}

\begin{remark}
We briefly discuss the effect of truncating the E-functions here. Clearly
$E^\BEC\ge E^\BEC_{\le d}$
holds pointwise since we drop some non-negative terms from $E^\BEC$ to obtain $E^\BEC_{\le d}$. Likewise
$
E^\BSC\le E^\BSC_{\le d}
$
since we assume all the high degree nodes are in error when computing $E^\BSC_{\le d}$. In fact, due to monotonicity, a better upper bound on $E^\BSC$ would be
\[
E^{\BSC}_{\le d}(\alpha,q)\le \sum_{k\le d} \P(\mathrm{Deg}=k)E^\BSC_k(q)+E_{d+1}^\BSC(q)\sum_{k> d} \P(\mathrm{Deg}=k).
\]
In practice, we choose the truncation degree to be large enough that makes this adjustment not so crucial. In either case, if an error probability is lower bounded by $E^\BEC$ it is also lower bounded by $E^\BEC_{\le d}$. Likewise, if it is upper bounded by $E^\BSC$, it is also upper bounded by $E^\BSC_{\le d}$.
\label{rmk:truncation}
\end{remark}
\begin{remark}
It is possible to study iterative decoding in terms of the input-output entropy instead of error probability. For linear codes (over the erasure channel), the two methods are equivalent as the EXIT function is proportional to the probability of error. For general codes, however, we would need to invoke a Fano type inequality to relate the two and this step is often lossy. For instance, in the case of LDMCs, we can obtain much better bounds by analyzing the probability of error  directly as shown in Section \ref{sec:fano}.  
\end{remark}

\subsection{Bounds via channel comparison lemmas}

Armed with E-functions and the channel comparison lemmas we can proceed to stating our iterative bounds.

\begin{proposition}
Consider the dynamical system 
\begin{equation}
q^\BEC_{t+1}(x_0)=1-2E_{\le d}^\BEC(\alpha,q^\BEC_{t})
\label{eq:qbec}
\end{equation}
initialized at $q^\BEC_0=x_0$ with $\alpha=C/R$. Similarly, define 
\begin{equation}
q^\BSC_{t+1}(x_0)=E_{\le d}^\BSC(\alpha,q^\BSC_{t})
\label{eq:qbsc}
\end{equation}
with $q^\BSC_0=x_0$.
Let $\delta_\ell^{\mathrm{BP}}$ be the $\BER$ of a  ensemble under $\BP$ after $\ell$ iterations. Likewise, let $\delta^{\mathrm{MAP}}$ be the $\BER$ under the optimal (bitwise $\mathrm{MAP}$) decoder. Then 
\[
\frac{1-q^\BEC_\ell(1)}{2}-o(1)\le \delta^{\mathrm{MAP}}\le \delta_\ell^{\mathrm{BP}}.
\]
Furthermore, 
\[
\frac{1-q^\BEC_\ell(0)}{2}-o(1)\le \delta_\ell^{\mathrm{BP}}\le q^\BSC_\ell(1/2)+o(1)
\]
with $o(1)\to 0$ and $k\to \infty$.
\label{prop:approx_gives_lower_bound}
\end{proposition}

\begin{proof}
We sample codes from the family and consider the (local) computational graph of a fixed bit $X_0$ with depth $\ell$.  It is known that for large codes, the computational graph of depth $\ell$ has a tree structure with high probability. Hence, we may assume that the graph is a tree. 

Consider the depth $\ell$ tree emanating from $S_0$. The channel
   $T_\ell:S_0\mapsto (\textup{computational tree of depth } \ell, \Delta^{(\ell)}X_{0})$ is a BMS (recall that $\Delta^{(\ell)}X_{0}=X_{\Delta^{(l)}(0)}$ denotes all the coded bits observed in the tree of depth $\ell$). We note that running $\ell$-steps of BP is equivalent to
   decoding $S_0$ from the output of $T_\ell$. In other words $\delta^\BP=P_e(T_\ell)$ is the error we want to bound. Further note that the structure of the tree is included as part of the channel, so that $P_e(T_\ell)$ is computed by randomizing over possible realizations of the graph as well.

   Now condition on the first layer of the tree. If the number of variable nodes in $\partial(0)$ is $m$, then the restriction of $T_{\ell}$ to the first layer has the
   structure of Lemma \ref{lem:lemma_B} (with $P_{\Delta X_0|S_0,\partial S_0}$ being the Boolean functions of various
   subsets in $\Delta(0)$) and each $W_j = T_{\ell-1}$ being the channels corresponding to the trees emanating from
   $X_j$'s (with $j \in \partial(0)$). More explicitly,
   for each choice of $S_0=s_0,\partial S_0=\partial s_0$,  $P_{\Delta X_0|S_0,\partial S_0}$ simply indicates whether or not all the constraints in the local neighborhood are satisfied:  $P_{\Delta X_0|S_0,\partial S_0}(\Delta x_0|s_0,\partial s_0)=\prod_{j\in \Delta(0)} \ind_{\{x_j=f(s_0,\partial_j s_0)\}}$. Furthermore, if we set $W_j=T_{\ell-1}$ to be the corresponding tree channel emanating from $S_j$'s (with $j\in\partial(0)$), then due to the locally tree assumption their observations are independent. 

   Now assume by induction that $T_{\ell-1} \preceq_{\mathrm{deg}} \BEC_{\bar{q}_{\ell-1}}$.
   Then by Lemma \ref{lem:lemma_B}, we have $T_\ell \preceq_{\mathrm{deg}} \tilde T_\ell$ where $\tilde T_\ell$ is the tree of depth $\ell$ in which the channels $W_j$ are replaced with $\BEC_{\bar{q}_{\ell-1}}$. Note that if we condition on the degree $d$ of $S_0$, then the $\tilde{T}$  channel has error  $E_d^\BEC(\alpha,q_{\ell-1}).$ By averaging over the degrees, we obtain
   \[
    P_e(\tilde T_\ell)=E^\BEC(\alpha,q_{\ell-1})\ge  E_{\le d}^\BEC(\alpha,q_{\ell-1})=\bar{q_\ell}/2,
   \]
    where the inequality is due to  truncation (recall that in $E^\BEC_{\le d}$ all nodes of degree larger than $d$ are assumed to have zero error--see Remark \ref{rmk:truncation}). 
    To complete the induction step, note that $\tilde T_\ell\preceq_{\mathrm{deg}}\BEC_{\bar{q}_\ell}$ by Lemma \ref{lem:lemma_A}. We thus have $P_e(T_\ell) \ge \bar{q}_{\ell}/2$ as desired. 
    
    The proof of the BSC upper bound is obtained in a similar manner after replacing the input channels to $\tilde T_\ell$ with BSCs and invoking the reverse sides of Lemmas \ref{lem:lemma_B},\ref{lem:lemma_A} again. 
    
    Finally, $\BP$ and $\mathrm{MAP}$ decoding differ only by the initialization of beliefs at the leaf nodes. Since the MAP channel at the leaves is a degradation of $\BEC_0$, the lower bound on MAP follows as well.

\end{proof}

\subsection{Computing E-functions for $\mathrm{LDMC(3)}$}
\label{sec:Ecurve_ldmc3}
In the rest of this section, we provide an algorithm to compute the E-functions for LDMC(3) and use Proposition \ref{prop:approx_gives_lower_bound} to obtain upper and lower bounds for BP and bit-MAP decoders for this family of codes. The degree distribution of LDMC(3) is asymptotically $\mathrm{Poi}(3\alpha)$ distributed where $\alpha=C/R$. In this case, the truncated erasure polynomial is
\[
E^\BEC_{\le d}(\alpha,q)=\sum_{k=0}^d \P(\mathrm{Poi}(3\alpha)=k)E^\BEC_k(q).
\]


Computing the erasure polynomials is more involved for LDMC(3) than LDGMs since the BP messages are more complicated. For LDGMs, the messages are trivial in the sense that every uncoded bit remains unbiased after each BP iteration. This does not hold for LDMCs, and it is in fact this very principle that allows BP decoding to initiate for LDMCs without systematic bits. Hence, to analyze BP locally, we need to randomize over all possible realizations of the bits in the local neighborhoods. This is a computationally expensive task in general, but one that can be carried out in some cases by properly taking advantage of the inherent symmetries in the problem.

 The BP update rules are easy to derive for LDMCs. Let $X_j$ be the majority of 3 bits $S_0,S_1,S_2$. Then if $X_j=0$, the check to bit message is
\begin{equation}
m_j=\frac{\P(S_0=0|X_j=0)}{\P(S_0=1|X_j=0)}=1+\frac{1}{r_1}+\frac{1}{r_2},
\label{eq:BP_law_maj3}
\end{equation}
where $r_i=\PP(S_i=0)/\PP(S_i=1)$ are the priors (or input messages to the local neighborhood). The posterior likelihood ratio for $S_0$ is $r_0=\prod_{j\in\Delta(0)}m_j$. 
We now use these update rules to compute the E-polynomials for LDMC(3). We start with the erasure polynomials.

 Let $\bar{q}=1-q$ be the probability of erasure at the boundary. For bits of degree zero, the probability of error is clearly $\frac{1}{2}$ and for bits of degree 1 the probability of error is $\frac{1}{4}$ independent of $q$. To see this, consider the computational tree of a degree $1$ bit $S_0$ at depth 1. There are two leaf bits in tree. Suppose that neither of the leaf bits is erased. This happens with probability $q^2$. Conditioned on this, only when the two leaf bits take different values can $S_0$ be fully recovered and this conditional probability is $\frac{1}{2}$. Otherwise, the bit remains unbiased and must be guessed randomly. The overall contribution of this configuration to the probability of error for $S_0$ is $q^2/2$. One other possible configuration is when only one leaf bit is erased. In this case the target bit is determined whenever the unerased bit disagrees with the majority, which happens with probability $\frac{1}{4}$. When the unerased bit agrees with the majority, it weakens the (likelihood ratio) message sent from the majority to the target bit. In this case, the message passing rule in (\ref{eq:BP_law_maj3}) shows that the probability of error is $\frac{1}{3}$.  Overall, the contribution of this configuration to the probability of error is $2q(1-q)/4$. Finally, if both bits are erased, which happens with probability $(1-q)^2$, then the probability of error is again $\frac{1}{4}$. Adding up all the error terms, we see that $E_1(q)=\frac{1}{4}$. It is true for any monotonic function that $E_1(q)$ is a constant. Indeed if $f$ is monotonic, then the decision rule for estimation of any degree 1 node depends in a deterministic fashion on the value of $f$ and not on the distribution of local beliefs. It can be checked that $E_2(q)$ depends on $q$ non-trivially. 

For the general case, the ideas are the same. Consider the message sent from the a majority check to a target bit modulo inversion. This means that we identify a message $m$ and its inverse $1/m$ as one group of messages. This is a random variable that depends on the erasure patterns as well as the realized values at the leaves. Let us first condition on the erasure patterns. In this case the message is either in $\{0,\infty\}$, $\{1\}$, $\{2,1/2\}$, or $\{3,1/3\}$. In the first case, the conditional error is zero, hence, we assume that one of the latter messages is sent. Let $M_i$ be the message sent from the $i$-th majority to the target bit modulo inversion. If we represent $\{1\}$ with a constant, $\{2,1/2\}$ with variable $s$, and $\{3,1/3\}$ with variable $t$, then the distribution of $M_i$ (modulo inversion) can be represented by the following polynomial
\begin{equation}
f(s,t,q)=q^2/2+2q(1-q)s+(1-q)^2t
\label{eq:fbec3}
\end{equation}
where $1-q$ is the erasure probability at the leaves. For a target node of degree $d$, the joint distributions of messages  $M_1,\cdots,M_d$ is given by a product distribution $\prod_{i}p_{M_i}$. Modulo permutation of messages, these can be represented by 
\begin{equation}
f(s,t,q)^d=\sum_{j,k:j+k\le d} f^d_{jk}(q)s^jt^k.
\label{eq:poly_expansion}
\end{equation}
Define $R_i=\ind_{\{M_i\in \{2,1/2\}\}}$ and $T_i=\ind_{\{M_i\in \{3,1/3\}\}}$ to be the indicators that either $\{2,1/2\}$ or $\{3,1/3\}$ are sent, respectively. Let $R=\sum_{i=1}^d R_i,T=\sum_{i=1}^d T_i$. Note that $\P(R=j,T=k)=f^d_{jk}(q)$, i.e., 
the coefficient of $s^jt^k$ in the above expansion of $f(s,t,q)^d$ is the probability of the event $\{R=j,T=k\}$. If we find the conditional error $E_{jk}$ associated with each monomial term in $f$, then we can conveniently represent the erasure polynomial as follows
\begin{equation}
E^{\BEC}_d(q)=\sum_{j,k:j+k\le d} f^d_{jk}(q) E_{jk}.
\label{eq:Ed}
\end{equation}
To this end, define $M(j,k)=(M_i(j,k))$ for all $i\le d$ with
\begin{equation}
M_i(j,k)=\left\{\begin{array}{cc}2 & 0\le i\le j\\ 3 & j< i\le j+k\\ 1& \textup{otherwise},  \end{array}\right.
\label{eq:Mjk}
\end{equation}
to map $R,T$ back to a realization of the incoming messages to $S_0$. By symmetry
\[
\P(\hat{S_0}\neq S_0|R=j,T=k)=\P(\hat{S_0}\neq S_0|M(j,k)).
\]
Let $A=(A_i)$ with $A_i=\ind_{\{\Delta_iX_0=S_0\}}$ being the indicator that the $i$-th majority agrees with the target bit. Let $p_{a|jk}=\P(A=a|M(j,k))$ be the conditional probability that $a$ is realized given the incoming messages.  Since the events $\{\Delta_iX_0=S_0\}$ are independent conditioned on $M_i$'s we have
\begin{equation}
p_{a|jk}=\prod_i \P(A_i=a_i|M_i(j,k))=\prod_i \frac{1}{1+M_i(j,k)^{2a_i-1}}.
\label{eq:agree_prob}
\end{equation}
The conditional probability of error given the joint realization of messages and majority votes is given by
\begin{equation}
E_{jk|A}=\min(\frac{1}{1+\prod M_i(j,k)^{1-2a_i}},\frac{\prod M_i(j,k)^{1-2a_i}}{1+\prod M_i(j,k)^{1-2a_i}})
\label{eq:error_prob}
\end{equation}
It is  convenient to define
\begin{equation}
E_{jk}=\sum_{a\in \{0,1\}^d} p_{a|jk} E_{jk|A}
\label{eq:Ejk}
\end{equation}
and think of it as the error associated to the monomial $y^jz^k$ in (\ref{eq:poly_expansion}). Algorithm \ref{algo:Ed} summarizes the proposed procedure to compute the erasure polynomial. 
For instance, for degree $4$ nodes we have the following erasure polynomial:
\begin{align*}
E^{\BEC}_4(q)&=0.03125q^8 + 0.25q^7(-q + 1) + 1.25q^6(-q + 1)^2 \\
&+2.875q^5(-q + 1)^3 + 4.6875q^4(-q + 1)^4 \\
&+4.4375q^3(-q + 1)^5+ 2.84375q^2(-q + 1)^6 \\
&+ 0.9375q(-q + 1)^7+ 0.15625(-q + 1)^8.
\end{align*} 
\begin{algorithm}[t]
\caption{Compute $E_d(q)$}
\begin{algorithmic}
\Function{ErrorPoly}{$d$}:
\State Define $f(s,t,q)=q^2/2+3/2q(1-q)s+(1-q)^2t$
\State Expand the $d$-th power of $f$
\[
f^d(q)=\sum_{j,k} f^d_{jk} s^jt^k
\]
\State Initialize $E:= 0$
\For{k:=1 to $d$ and $j\le k$}{}
\State Compute $E_{jk}$ using (\ref{eq:Mjk})-(\ref{eq:Ejk})
\State Update $E:= E+E_{jk}f^d_{jk}$\\
\EndFor
\textbf{return} $E$ 
\EndFunction
\end{algorithmic}
\label{algo:Ed}
\end{algorithm}

Fig.~\ref{fig:exit_deg48} compares $E^\BEC_d$ with the empirical BER of  degree d nodes across samples from its depth 1 computational tree with $\BEC$ inputs for d=4,8. 
For many code ensembles an exact computation of $E^\BEC_d$ is often computationally prohibitive. In such cases, one can sample from the computational tree and find $E^\BEC_d$'s by solving a regression problem. Such functions are useful in optimizing codes as we will see in Section \ref{sec:comp_ldgm}. 
\begin{figure*}[!t]
\centering
\subfloat[$d=4$]{\includegraphics[width=0.5\textwidth]{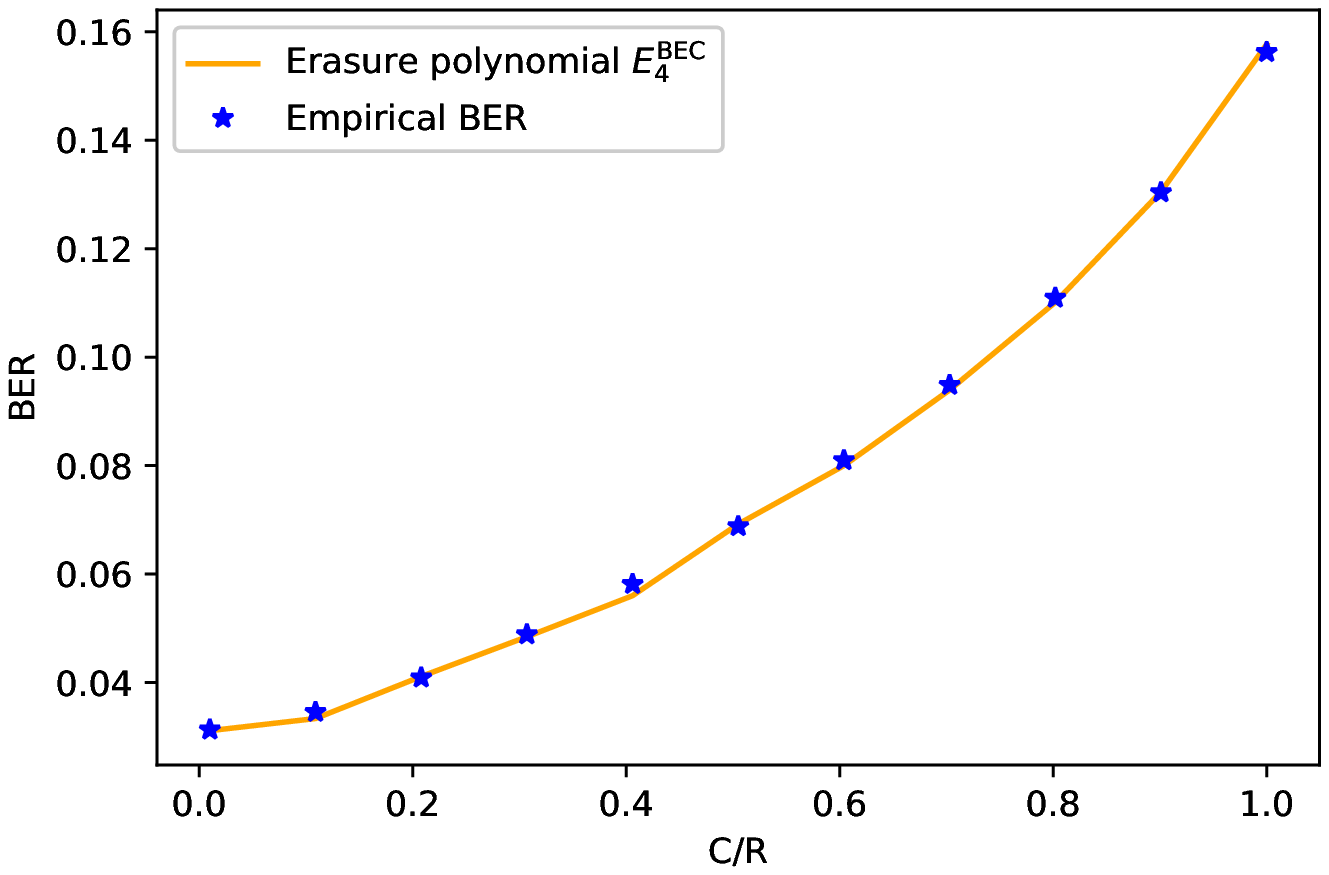}%
}
\hfil
\subfloat[$d=8$]{\includegraphics[width=0.5\textwidth]{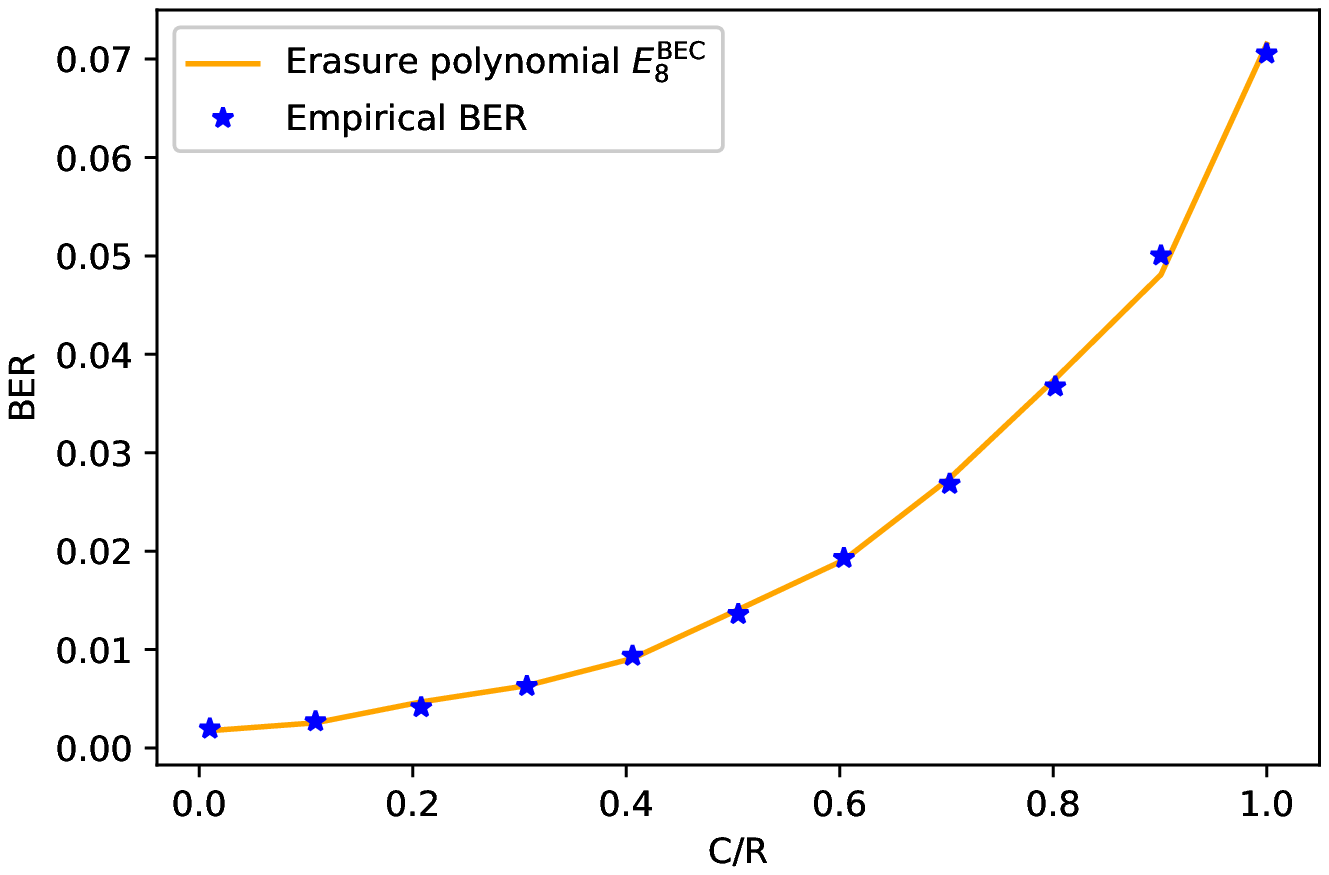}}%
\caption{Comparing the erasure polynomials $E^\BEC_4,E^\BEC_8$ with their empirical means. The empirical curves are obtained using $50000$ samples from the computational trees of depth 1 for target nodes of degrees 4 and 8, respectively, with leaves observed through $\mathrm{BEC}_{\epsilon}$ as in Fig.~\ref{fig:comp_trees}b.}
\label{fig:exit_deg48}
\end{figure*}

Recall the definitions of $q_t^\BEC(x_0)$ and $q_t^\BSC(x_0)$ from (\ref{eq:qbec})-(\ref{eq:qbsc}). Once we compute the E-polynomials, we iterate the dynamical system in (\ref{eq:qbec})-(\ref{eq:qbsc}) to find bounds on the decoding error. We compare the bounds with the empirical performance of BP in Fig.~\ref{fig:BP_fixed} for LDMC(3). We see a good agreement between the two. In particular, we see that the lower bound for LDMC(3) is almost tight. To explain this, we need to consider the distribution of posterior beliefs in LDMC(3). As shown in Fig.~\ref{fig:hist}, the empirical histogram of beliefs after convergence of BP at $C/R=1$ has three major spikes: two spikes at $p=0,1$ and one at $p=0.5$. The rest of the beliefs are almost uniformly distributed across the range $[0,1]$. 
It thus seems reasonable to approximate the posteriors obtained by BP as if they were induced by erasure channels. We emphasize that this phenomenon is specific to ensembles of degree $3$. For larger degrees, the histogram has a pronounced uniform component (see Fig.~\ref{fig:hist_d5}). Thus one cannot expect a similar agreement between the BEC lower bound and the BER performance (see Fig.~\ref{fig:ldmc_d5vsd3}). 
\begin{figure}[ht]
\centering
\subfloat[$C/R=1$]{\includegraphics[width=0.5\textwidth]{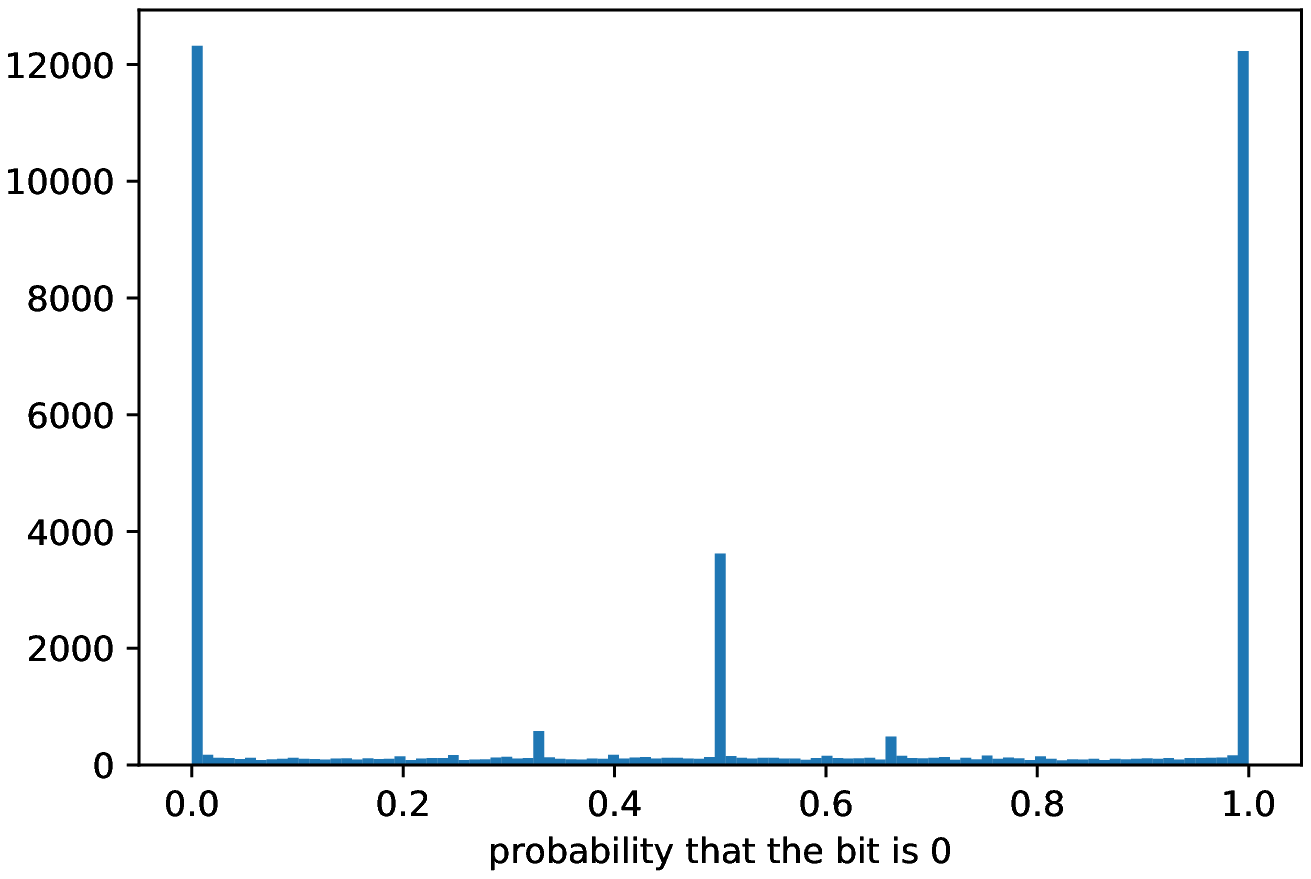}}
\subfloat
[$C/R=0.25$]{
\includegraphics[width=0.5\textwidth]{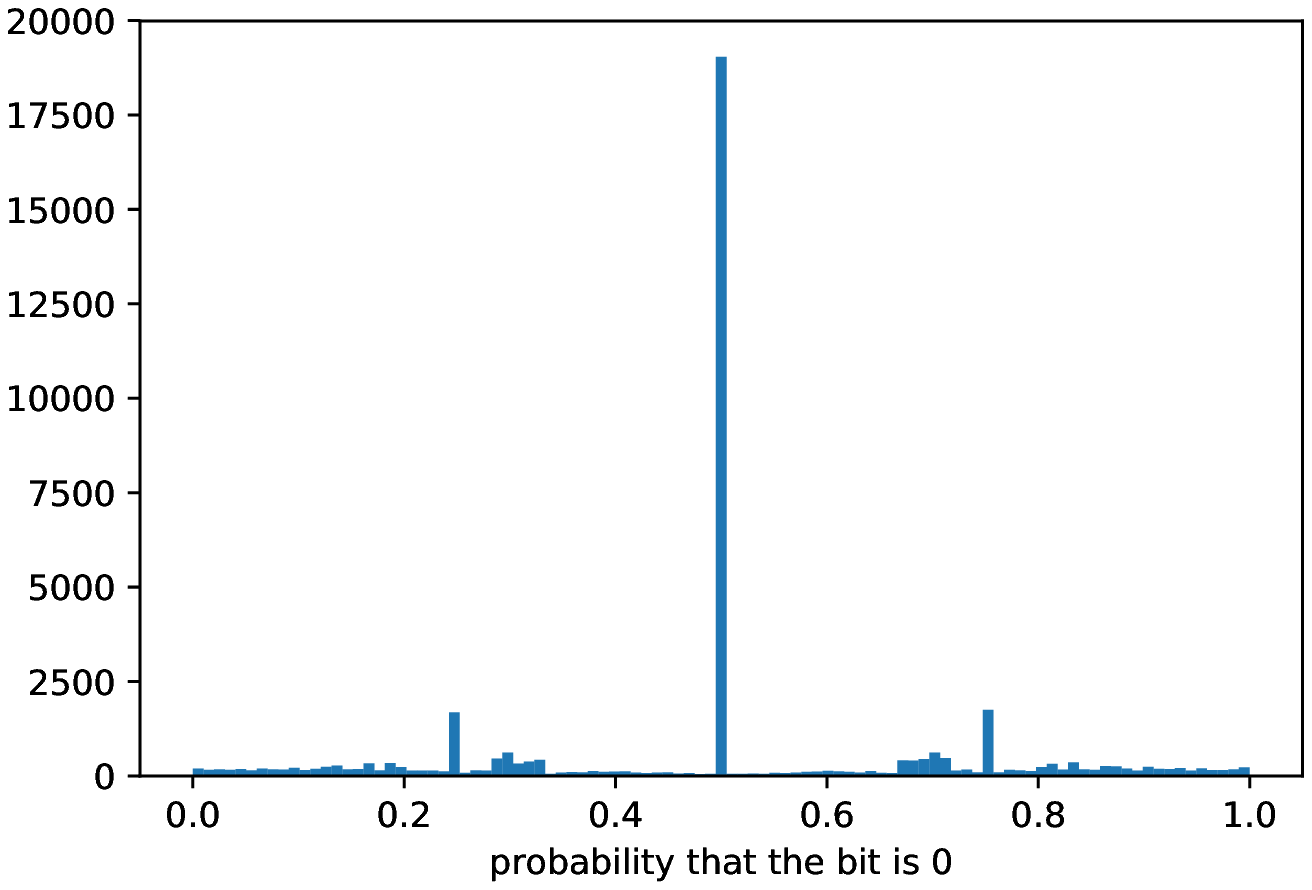}}
\caption{The empirical histogram of belief distributions for LDMC(3) with $k=40000$ bits. The number of bits that are $0$ with probability close to $p$ are shown as a function of $p$ for a) $C/R=1$ b) $C/R=0.25$.}
\label{fig:hist}
\end{figure}

\begin{table*}[t]
\centering
\begin{tabular}[width=0.75\textwidth]{|l|l|l|l|l|l|l|l}
\hline
$C/R$&$(E^\BEC_2,\BER_2)$ & $(E^\BEC_3,\BER_3)$  &$(E^\BEC_4,\BER_4)$&$(E^\BEC_5,\BER_5)$ \\
 \hline
 0.25& (0.194,0.202) &(0.127,0.146)  &   (0.097,0.117)&(0.068,0.093)\\
 \hline
 0.5& (0.166,0.177)  &(0.106,0.124)  & (0.070,0.090)& (0.047,0.066) \\
 \hline
 1& (0.137,0.139)&(0.077,0.081)&(0.044,0.047)&(0.025,0.028)\\
\hline 
\end{tabular}
\caption{Comparing $\BER_d$, the empirical bit error rate of degree $d$ nodes after 10 iterations of BP, with the theoretical lower bounds $E^\BEC_d$ at various $C/R$'s. The lower bounds are computed at $1-2\BER$ for each $C/R$ where $\BER$ is obtained empirically. }
\label{table:E_A_BER}
\end{table*}

\begin{figure}[ht]
\centering
     \subfloat[Density evolution lower bounds for MAP and BP]{\includegraphics[width=0.5\textwidth]{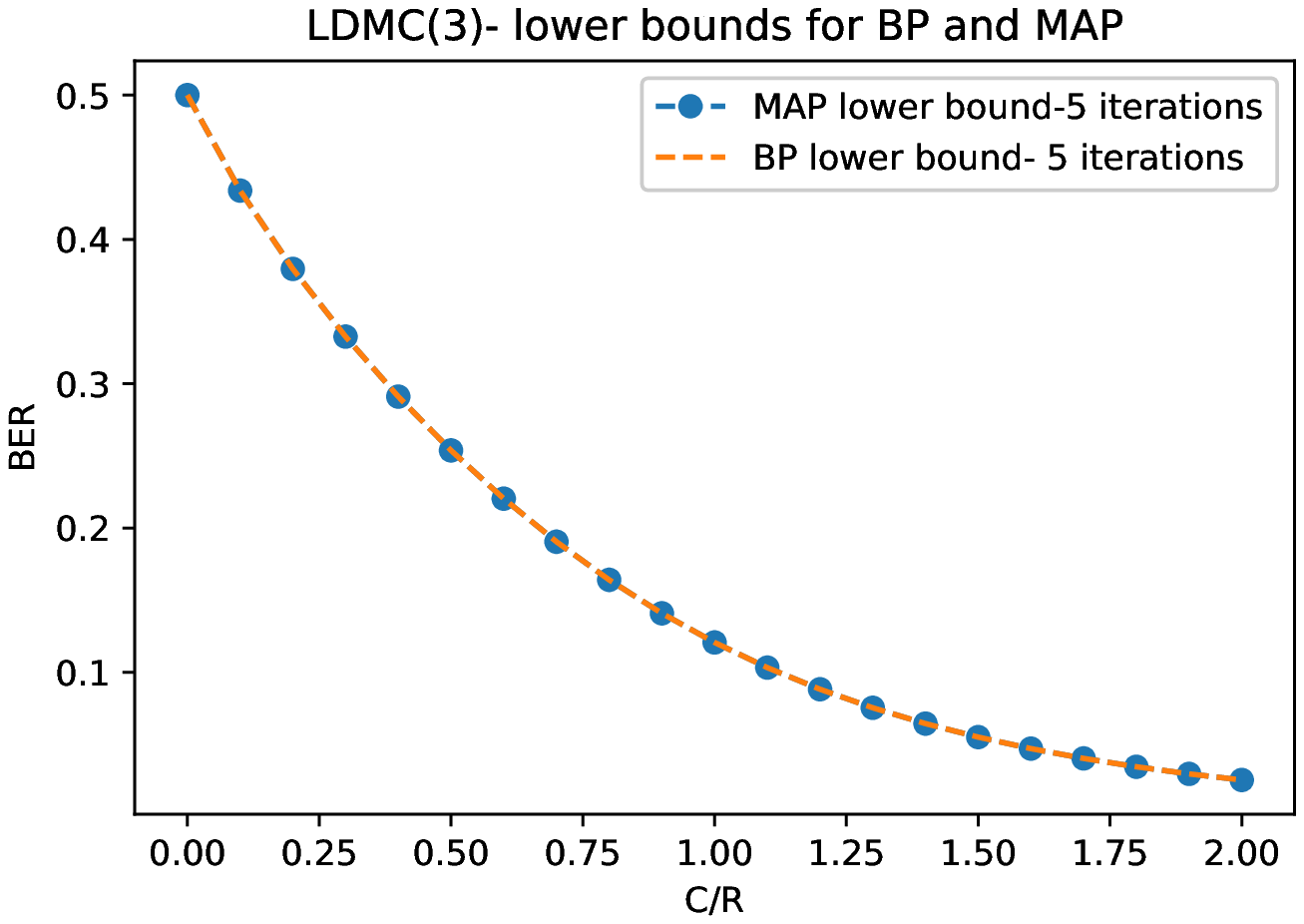}}
     \subfloat[Bounds from density evolution vs empirical BER]{\includegraphics[width=0.5\textwidth]{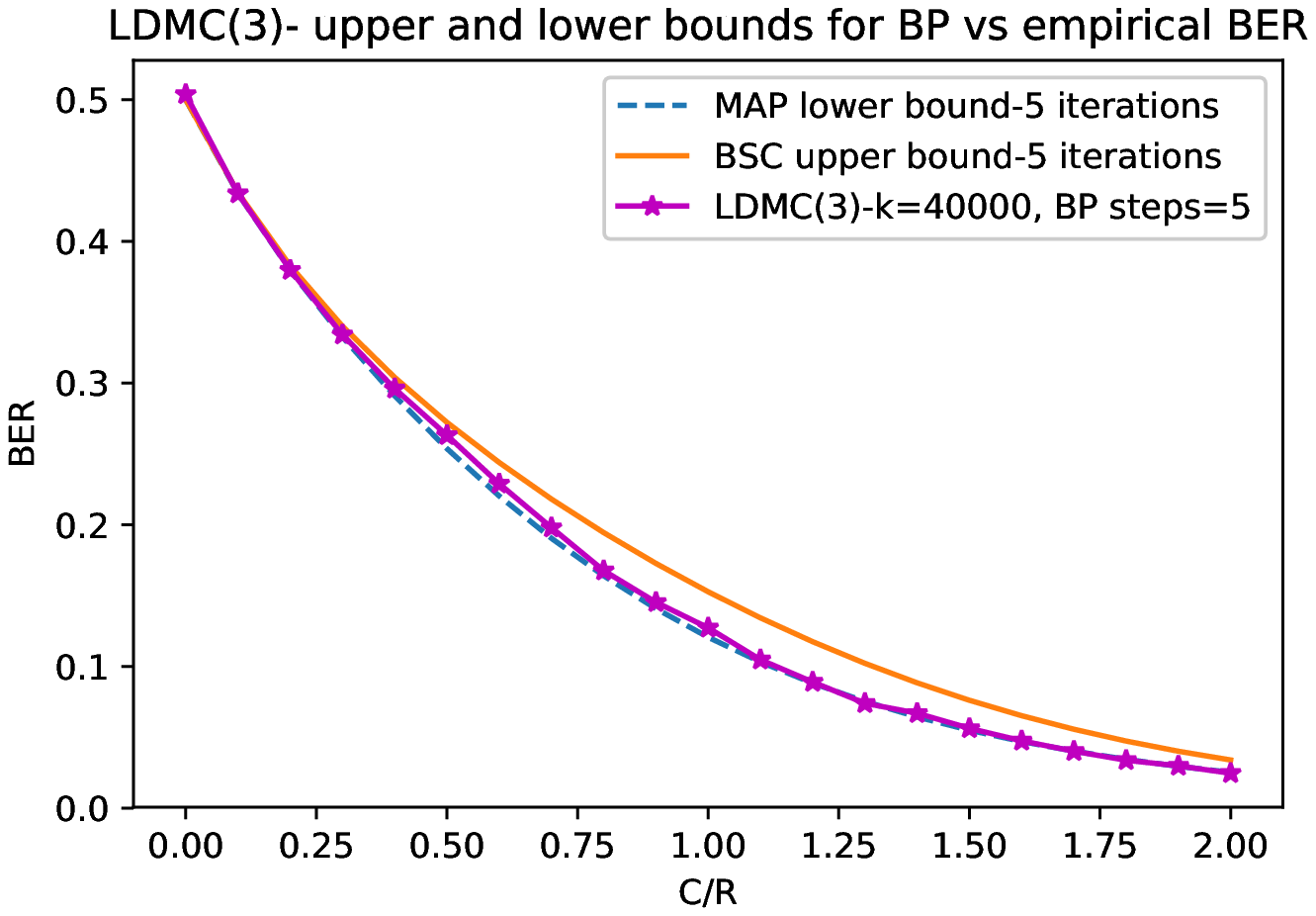}}

\caption{The LDMC(3) performance with 5 iterations of BP along with the bound of Proposition  \ref{prop:approx_gives_lower_bound} using $\ell=5$ and $E_{d\le 10}$-functions. The BP lower bound is $q_5^\BEC(0)$ from (\ref{eq:qbec}) using the erasure functions associated with the polynomial in (\ref{eq:fbec3}). The MAP lower bound is $q_5^\BEC(1)$. The BSC upper bound is $q_5^\BSC(1/2)$ from (\ref{eq:qbsc}) using the error polynomial in (\ref{eq:fbsc3}). a) The density evolution dynamics of (\ref{eq:qbec}) has a unique fixed point. Hence $q_\ell^\MAP$ and $q_\ell^\BEC$  converge to the same point. We remark that this property does not hold for general codes (see Conjecture \ref{conj:monotone} and Remark \ref{remark:q0vsq1} below).  b) The BP performance is compared against the bitwise MAP lower bound. The lower bound is almost tight since the empirical histogram of beliefs in LDMC(3) is much closer to one induced by an erasure channel than BSC (see Fig.~\ref{fig:hist}). }
\label{fig:BP_fixed}
\end{figure}

The erasure polynomials needed to compute $E^\BEC_{\le 10}$ are listed in Appendix \ref{apx:polys} in Python form. Table \ref{table:E_A_BER} compares the values of $E^\BEC_d(1-2\BER)$ with the empirical BER of degree $d$ nodes in the LDMC(3) ensemble after 10 iterations of BP. 

The ideas to compute the BSC upper bound are similar. Recall that in (\ref{eq:Ejk}), $E_{jk}$ is the error associated to the monomial $s^jt^k$ (meaning that j of type 1 and k of type 2 messages are received) for LDMC(3). In general we can re-write (\ref{eq:Ejk})
 in the form
\[
\sum_{} E_{i_1,\cdots_{i_s}}f^d_{i_1,\cdots_{i_s}}(q)
\]
where again $E_{i_1,\cdots_{i_s}}$ is a channel-independent term that corresponds to the conditional  error given the input types at the boundary. The only term that depends on the channel is $f^d_{i_1,\cdots_{i_s}}(q)$. Thus for any channel once we find the corresponding $f$-polynomial we can construct upper/lower bounds as above.

Let us construct the $f$-polynomial of LDMC(3) for BSC. Again consider the local neighborhood of a target node connected to one majoiry. Note that for the two leaf nodes in the boundary, each realization $00,01,10,11$ is equally likely (after possible flips by BSC). We need to compute the likelihood that they agree with their majority given the realization. Let $\partial S_0$ be the boundary bits, $\partial Y$ be their observations through $\BSC_p$, and $X$ be the majority. We proceed as follows. 
\begin{itemize}
\item The observed value is $ Y=00$.
We can check that 
\[
\P(X=0|\partial Y=00)=
(1-p)^2+p(1-p).
\]
The message corresponding to this event is
\[
r_0=1+2/\rho
\]
with $\rho:=\frac{1-p}{p}$. The complimentary event $\P(X=1|\partial Y=00)$ has probability $p(1-p)+p^2$ and the corresponding message sent to the target node is
\[
r_0=\frac{1}{1+2\rho}.
\]
Let $s$ represent $1+2/\rho$ and $t$ represent $1+2\rho$ (modulo inversion). Overall, the outgoing message for the assumed observed value is 
\[
r_0= t(p(1-p)+p^2)+s((1-p)^2+p(1-p)).  
\]
\item Suppose that $\partial Y=11$ is observed. By symmetry
\[
\P(X=0|\partial Y=11)=\P(X=1|\partial Y=00)=p(1-p)+p^2.
\]
The corresponding message is $1+2\rho$. Likewise
\[
\P(X=1|\partial Y=11)=\P(X=0|\partial Y=00)=p(1-p)+(1-p)^2
\]
with message $\frac{1}{1+2/\rho}.$ 
The outgoing message is 
\[
r_0= t(p(1-p)+p^2)+s((1-p)^2+p(1-p)).  
\]
\item Suppose that $\partial Y=01$ or $\partial Y=10$ is observed. Then 
\[
\P(X=0|\partial Y=10)=\P(X=1|\partial Y=10)=1/2
\]
by symmetry. The corresponding messages in each case are, $1+\rho+1/\rho$ for $X=0$ and $\frac{1}{1+\rho+1/\rho}$ for $ X=1$,
which we represent by $z$. 
\item Adding up all the terms, we get the following $f$-polynomial to compute $E^\BSC$:
\begin{equation}
f=\frac{z}{2}+\frac{1}{2}(t(p(1-p)+p^2)+s((1-p)^2+p(1-p))).
\label{eq:fbsc3}
\end{equation}
\end{itemize}
We use this polynomial to compute $E^\BSC_{\le 10}$ and obtain an upper bound on BP error using Proposition \ref{prop:approx_gives_lower_bound}. The upper bound is compared with the simulation results in  Fig.~\ref{fig:BP_fixed}. The above polynomial is also used in Section \ref{sec:channel_transform} to derive bounds on soft information.

\subsection{Comparing $\mathrm{LDMC(3)}$ with $\mathrm{LDMC}(5)$}
It is natural to ask how the BER curves behave for LDMC(d) as $d$ grows. This question is in general computationally difficult to answer. The girth of the computational graph grows exponentially fast with $d$ and BP iterations do not seem to stabilize quickly enough when $d$ is large. Hence, one needs to consider codes of large length and more iterations of BP.  Nevertheless, the case of degree 5 can still be worked out with modest computational means. Here we compare the performance of LDMC(5) with LDMC(3). We also compute the erasure function of LDMC(5) and compare the corresponding bound with simulations. As mentioned before, the spiky nature of histogram observed in Fig.~\ref{fig:hist} is specific to the ensembles of  degree $3$ and hence one cannot expect the BEC lower bound of Proposition \ref{prop:approx_gives_lower_bound} to give equally good predictions on BER for higher degrees. 

We first work out the computation of $E^\BEC$ for LDMC(5). As before, we need to consider various cases for realization of erasures at the input layer:

\begin{itemize}
\item No input bits are erased. This case occurs with probability $q^4$. If the input bits are balanced, no error occurs. The complimentary event in which the bits are not balanced has probability $5/8$, in which case the message to the target bit is  $r_0=1$ and error is $1/2$. The corresponding term is $5/8q^4$.
\item One bit is erased. This happens with probability $4q^3(1-q)$. There are two cases to consider in which an error may occur: 1) all three unerased bits agree with the majority; this happens with probability 1/4, and the corresponding message is $1$. 2) Two unerased bits agree with the majority; this happens with probability 9/16; the corresponding message is $2$, which we represent with $u$. Overall, the error term is $4q^3(1-q)(1/4+9/16u)$.
\item Two bits are erased. This happens with probability $6q^2(1-q)^2$. There are two cases in which an error may occur: 1) both unerased bits agree wit the majority; this happens with probability $7/16$, and the corresponding message is $4/3$, denoted by $w$. 2) one unerased bit agrees with the majority; this happens with probability $1/2$, and the corresponding message is $3$, denoted by $z$. Overall, the error term is $6q^2(1-q)^2(7/16w+1/2z)$. 
\item Three bits are erased. This happens with probability $4q(1-q)^3$. Two cases need to be considered: 1) the unerased bit agree with the majority, which happens with probability 11/16, in which the message is $7/4$; we represent this message by $v$. 2) the unerased bit disagrees with the majority, which happens with probability $5/16$ and gives a message of $4$, represented here by $s$. The corresponding term is $4q(1-q)^4(11/16v+5/16s)$.

\item All bits are erased. This happens with probability $(1-q)^4$ in which case the message is $11/5$. We represent this message by $t$. The corresponding term is $(1-q)^4t$.
\item Adding up all the terms, we get the following polynomial
\begin{equation}
f(q)=5/8q^4+4q^3(1-q)(1/4+9/16u)+6q^2(1-q)^2(7/16w+1/2z)+4q(1-q)^3(11/16v+5/16s)+(1-q)^4t.
\label{eq:fbec5}
\end{equation}
\end{itemize}
\begin{figure}[ht]
\centering
     \subfloat[Density evolution lower bounds for MAP and BP]{\includegraphics[width=0.5\textwidth]{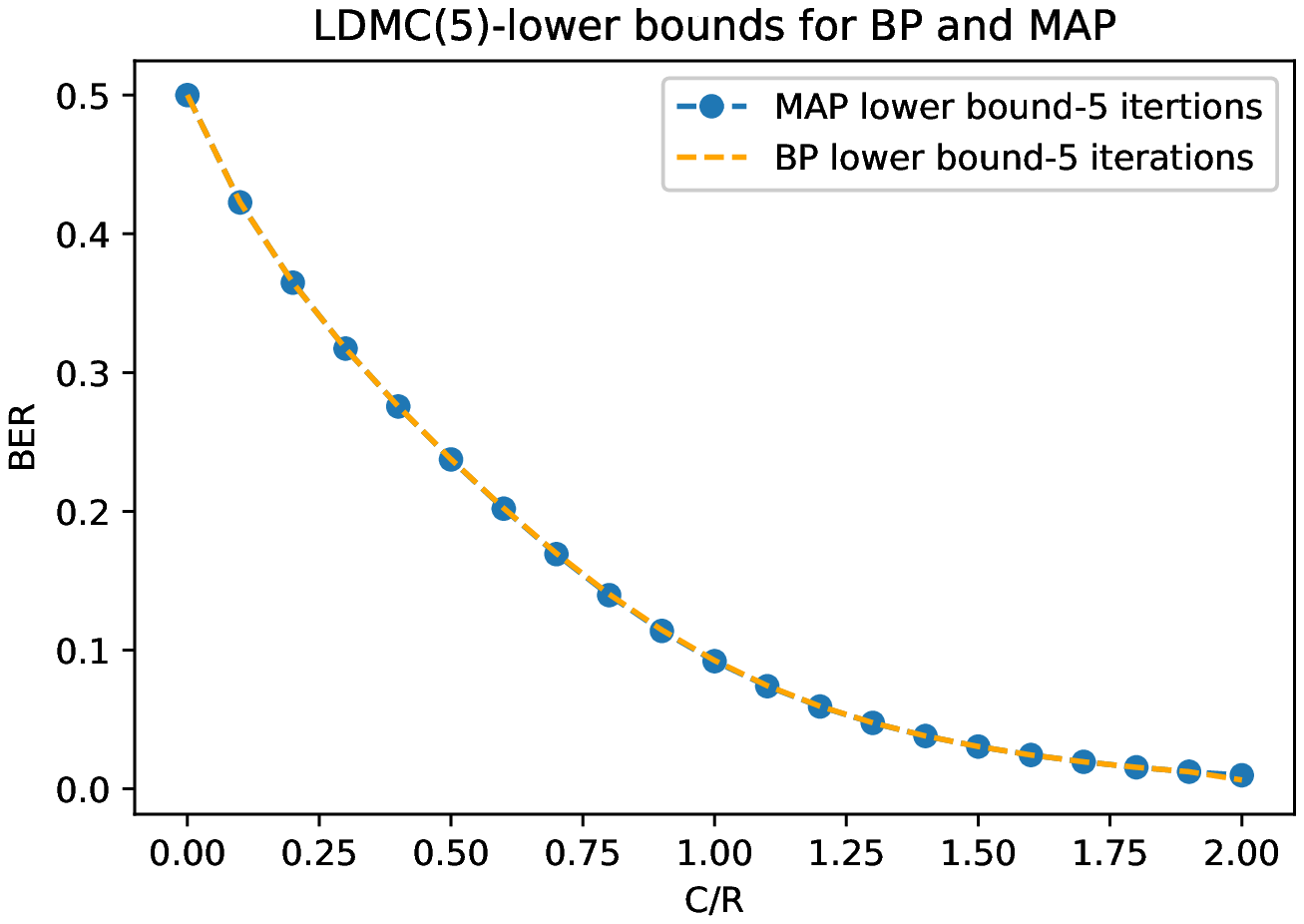}}
\subfloat[Bounds from density evolution vs empirical BER]{\includegraphics[width=0.5\textwidth]{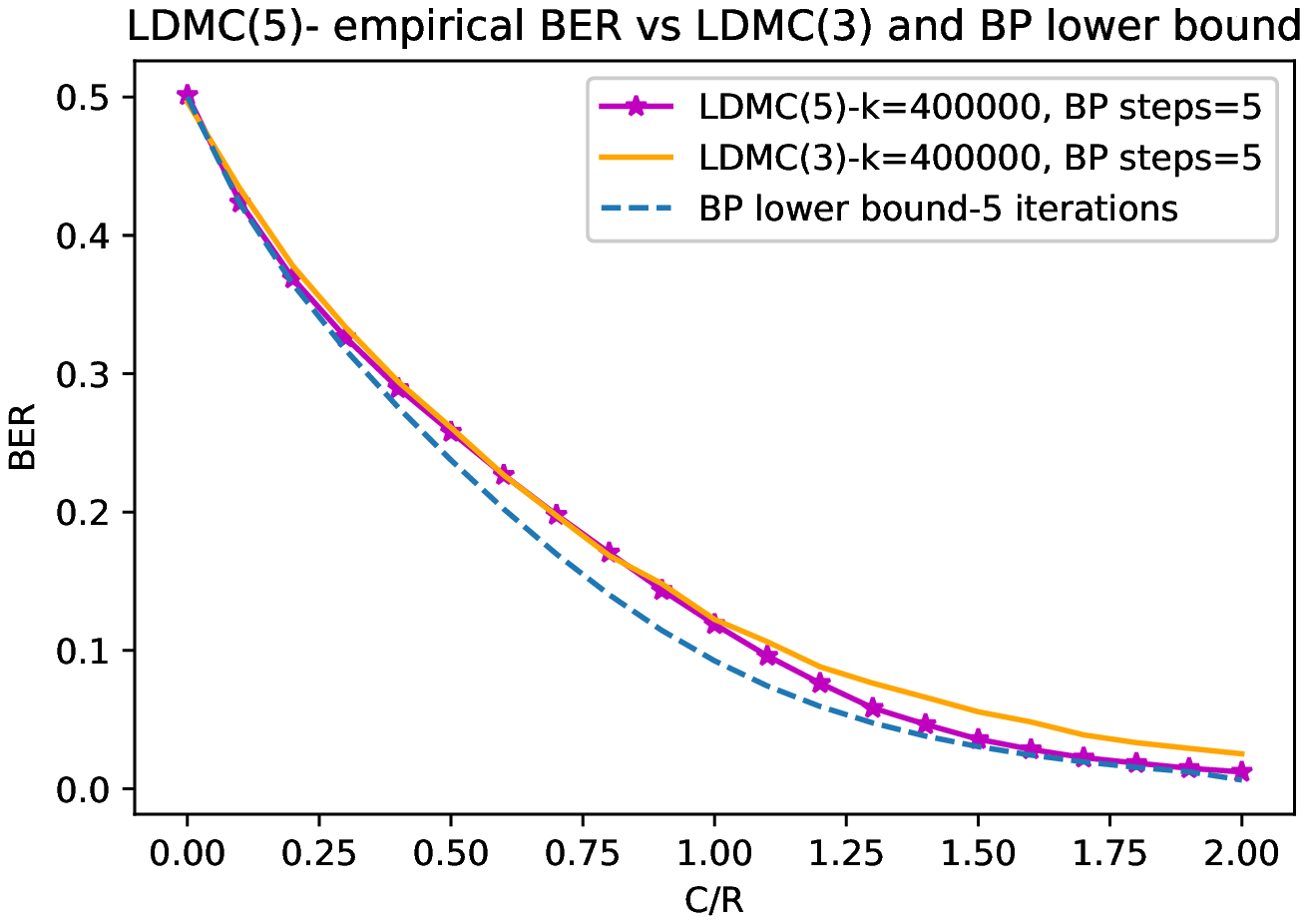}}

\caption{The LDMC(5) performance with 5 iterations of BP along with the lower bound of Proposition  \ref{prop:approx_gives_lower_bound} using $\ell=5$ and $E^\BEC_{\le 10}$-function. The BP lower bound is for LDMC(5), and is equal to $q^\BEC_5(0)$ in (\ref{eq:qbec}) with the erasure function obtained from (\ref{eq:fbec5}). The MAP lower bound is $q^\BEC_5(1)$. a) The density evolution dynamics of (\ref{eq:qbec}) has a unique fixed point. Hence both bitwise MAP and BP lower bounds converge to the same point. b) The BP performance for LDMC(5) is compared against its lower bound $q^\BEC_5(0)$  and the performance of LDMC(3). The lower bound can be seen to be looser compared with the scenario in Fig.~\ref{fig:BP_fixed}. This can be attributed to the fact that the empirical histogram of beliefs in LDMC(5) shown in Fig.~\ref{fig:hist_d5} is less spiky compared with Fig.~\ref{fig:hist} and can no longer be well approximated by one parameter (i.e., the erasure probability). Furthermore, LDMC(5) can be seen to perform better than LDMC(3) for all erasure probabilities. However, we note that longer codes are needed to avoid short cycles in the computational graph and achieve good decoding performance for $d=5$. }
\label{fig:ldmc_d5vsd3}
\end{figure}

\begin{figure}[ht]
\centering
\subfloat[$C/R=1$]{\includegraphics[width=0.5\textwidth]{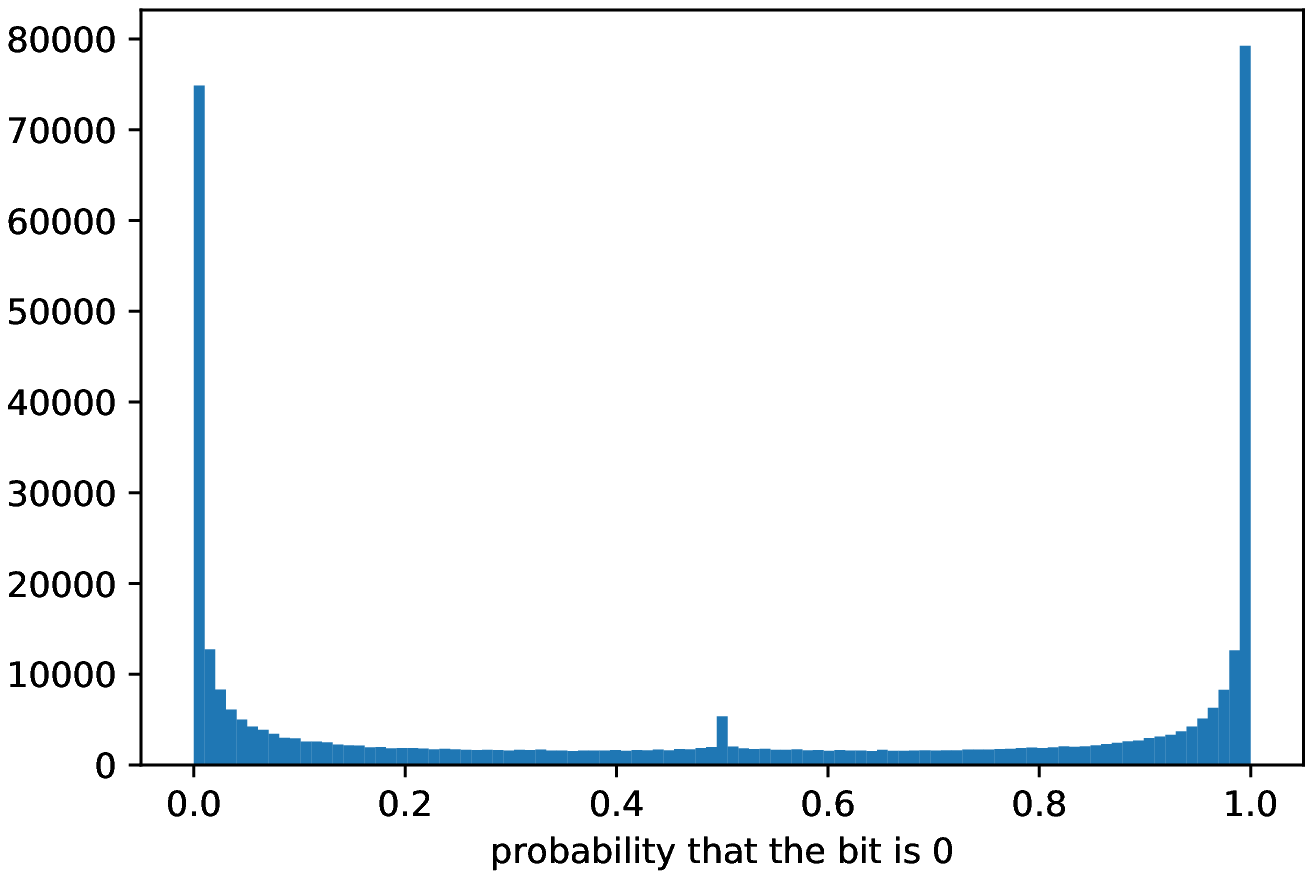}}
\subfloat
[$C/R=0.25$]{
\includegraphics[width=0.5\textwidth]{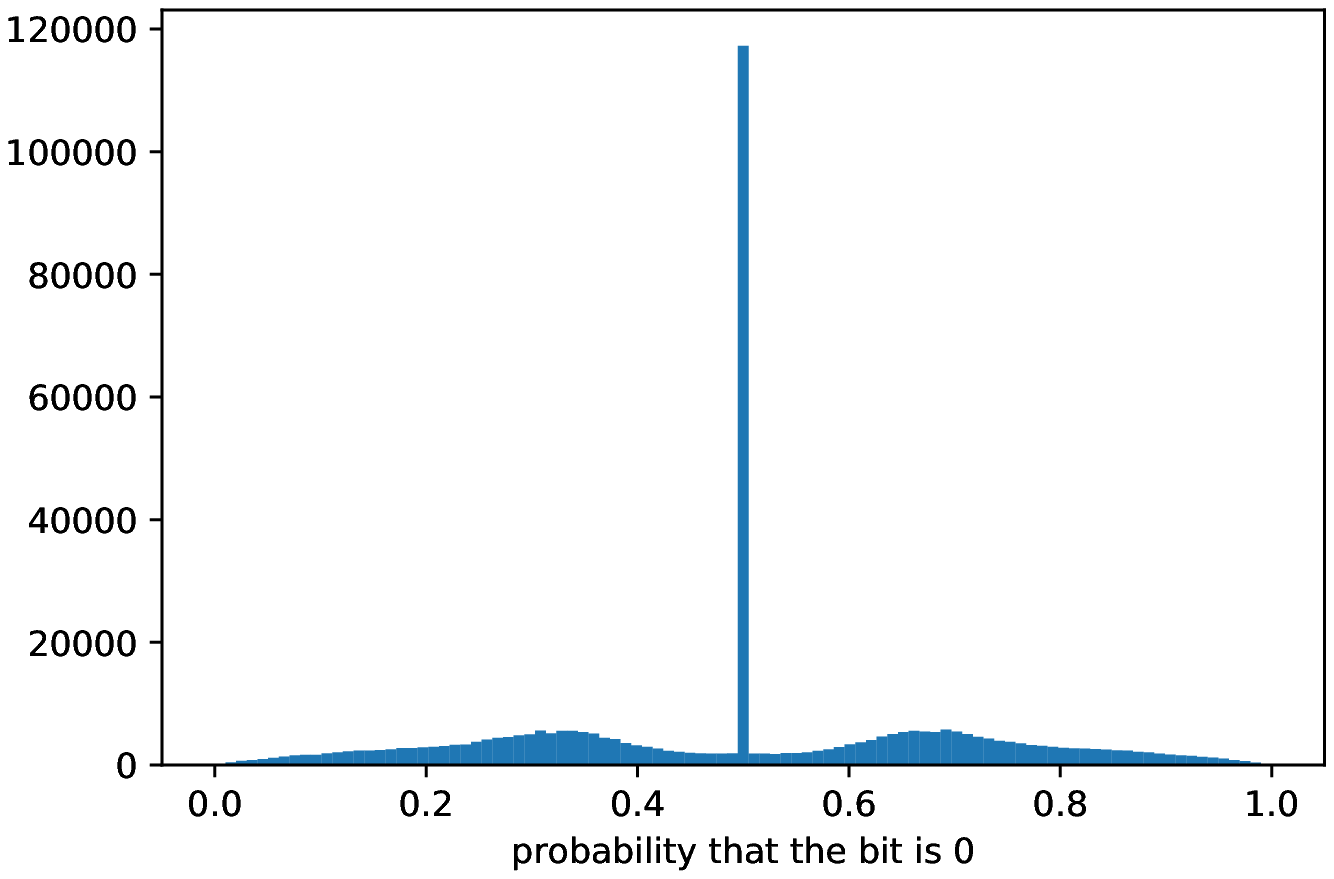}}
\caption{The empirical histogram of belief distributions for LDMC(5) with $k=400000$ bits after 5 iterations of BP. The number of bits that are $0$ with probability close to $p$ are shown as a function of $p$ for a) $C/R=1$ b) $C/R=0.25$.}
\label{fig:hist_d5}
\end{figure}
Using (\ref{eq:Ed}) we compute $E_{\le 10}^\BEC$ for LDMC(5) and then apply Proposition \ref{prop:approx_gives_lower_bound} to compute a lower bound on BER. The results are shown in Fig.~\ref{fig:ldmc_d5vsd3} along with comparisons between ensembles of degree $3$ and $5$ for 5 iterations of BP. We note that the effect of truncation is of a lower order than the scale of the plots in Fig.~\ref{fig:ldmc_d5vsd3}. Since $E_d(q)$ is monotonically decreasing in $d$ and $q$, we can deduce for all $\alpha\le 1$ that
\[
|E^\BEC(\alpha,q)-E^\BEC_{\le 10}(\alpha,q)|\le E^\BEC_{10}(0)\PP(\mathrm{Poi}(5)>10)=0.001.
\]
Thus the gap between $q_\ell$ and $\delta^{BP}$ for the degree 5 ensemble cannot be attributed to the truncation, but rather to the role of the ``uniform'' component of the belief histogram shown in Fig.~\ref{fig:hist_d5}. 

We still see in Fig.~\ref{fig:ldmc_d5vsd3}a that $q^{\BEC}_l$ converges to a unique point regardless of the initial condition for LDMC(5). We remark that the same holds for the error dynamics of the large $d$ limit obtained below in (\ref{eq:dynamics_large_d}). In the view of these observations, we put forth the following conjecture:

\begin{conjecture}
For any ensemble generated by a monotone function, $q^\BEC_\ell(x)$ converges to a unique fixed point independent of $x$.
\label{conj:monotone}
\end{conjecture}

\begin{remark}
We note that the conjecture does not hold for general ensembles. For instance, we have $q_\ell^\BEC(1)=1$ for LDGM(d) whereas $q_\ell^\BEC(0)=0$ for all $\ell$. In fact, Mackay showed in ~\cite{mackay1999good} that the ensembles generated by $\mathrm{XOR}$ are {\em very good}, meaning that for large enough degree they can asymptotically achieve arbitrarily small error for rates close to capacity under $\MAP$ decoding. Evidently, such performance cannot be achieved by $\BP$ since for any degree larger than $1$, $q=0$ is a fixed point of $\BP$, i.e., $\BER$ is $1/2$ for all $\ell$. This point shall be explained further in Section \ref{sec:comp_ldgm} (see Fig.~\ref{fig:Dq}).
\label{remark:q0vsq1}
\end{remark}

\subsection{Tighter bounds for systematic $\mathrm{LDMC}(d)$ with $d=3,5$}
It is possible to obtain tighter bounds for systematic ensembles. Here we study the case of systematic regular LDMC(3). The next section extends the analysis to the large $d$ limit for LDMC(d). 

Consider a regular ensemble of (check) degree $d$. Let $\epsilon$ be the probability of erasure and $R$ be the rate of the code with variable degree $1+d(1-R)/R$. Note that we need $d(1-R)/R\in \mathbb{Z}$ to ensure that a regular systematic code exists. As before let $\alpha=C/R$. For a regular systematic ensemble of rate $R$, we have the following erasure function:
\begin{equation}
E^\BEC(q,\alpha)=(1-\alpha R)\sum_{i\le d(1-R)/R} \PP(\mathrm{Bin}(\frac{d(1-R)}{R},\alpha R)=i)E_i^\BEC(q,\alpha).
\label{eq:sys_reg}
\end{equation}
The key observation is that BP can be initially loaded with the information obtained from the systematic bits. In other words we can iterate the dynamical system in (\ref{eq:qbec}) with $x_0=1-\alpha R$ and $E^\BEC$ as above. Clearly, $q^\BEC_1(x_0)$ gives an exact estimate for the first iteration of BP and can serve as an upper bound for the error $\delta_\ell^\BP$. The results are shown in Fig.~\ref{fig:ldmc_tight_bounds} for $R=1/2$ and $d=3$. The bounds can be seen to be rather tight. We note that the accuracy of these bounds depend primarily on the rate and the check degree of the ensemble and not the regularity assumption.  
\begin{figure}[ht]
\centering
\subfloat[LDMC(3) vs bounds]{\includegraphics[width=0.5\textwidth]{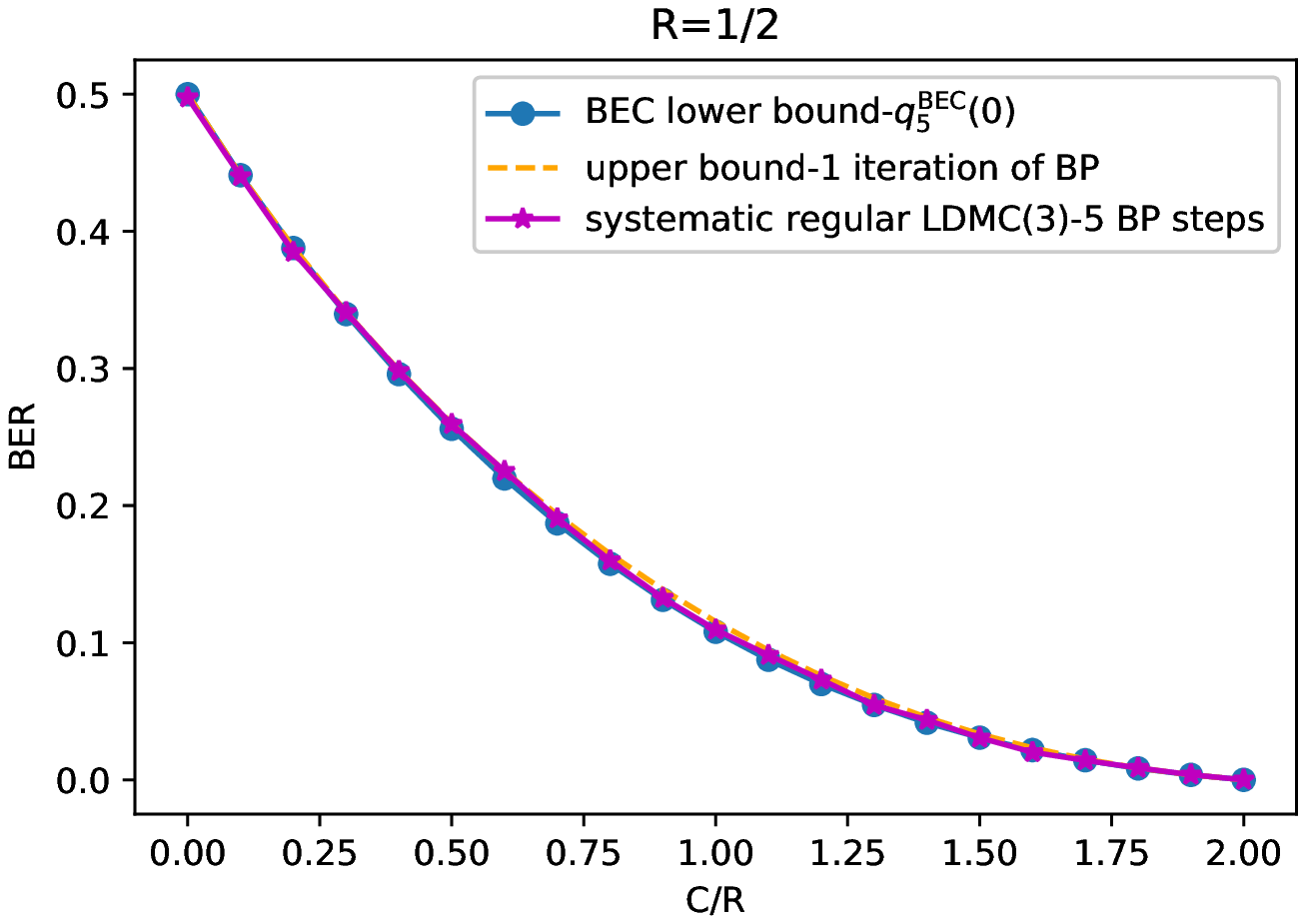}}
\subfloat
[LDMC(5) vs bounds]{
\includegraphics[width=0.5\textwidth]{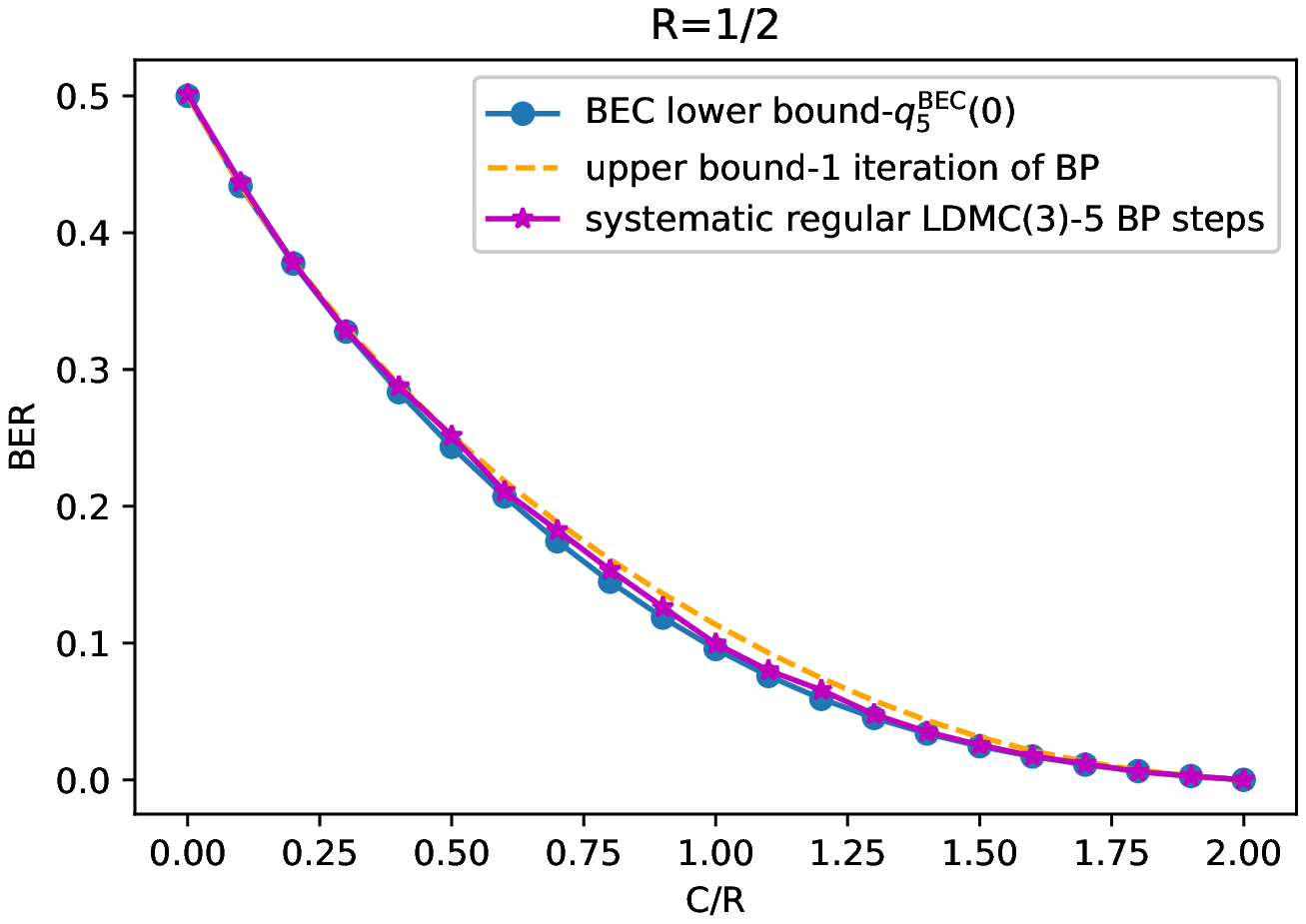}}
\caption{The performance of systematic regular LDMC(d) of rate $R=1/2$, $k=40000$, using 5 steps of BP along with the upper and lower bounds obtained from the erasure functions for (a) d=3 and (b) d=5. The upper and lower bounds match for systematic LDMC(3). The upper bound uses one iteration of (\ref{eq:qbec}) with the erasure function as in (\ref{eq:sys_reg}) and initialized at $x_0=1-\alpha R$. The initial point is the fraction of unerased bits observed in the systematic portion of the code.}
\label{fig:ldmc_tight_bounds}
\end{figure}

\subsection{Upper bound for systematic $\mathrm{LDMC}(d)$ as $d\to \infty$}
\label{sec:upper_bound}
Now we consider the case where the node degree tends to infinity for systematic LDMC(d) of rate $\frac{1}{2}$. To get an upper bound for LDMC codes in this case, we can analyze one step of BP. To do this, we first need to understand what a typical majority to bit message looks like as degree increases. 

Consider a majority $X$ of $d+1$ bits $S_0,\cdots,S_d$. Let $r_i:=\pi_{S_i}(0)/\pi_{S_i}(1)$, $i>0$ be the incoming messages to the local neighborhood. Let $r_0:=\pi_{S_0}(0)/\pi_{S_0}(1)$ be the outgoing message. Then the BP update rule when $X=0$ is as follows:
\begin{equation}
r_0=1+\frac{\sum_{|I|=d/2}\prod_{i\in I} 1/r_i}{\sum_{|I|< d/2}\prod_{i\in I} 1/r_i}.
\label{eq:BP_update_large_d}
\end{equation}
Set $\alpha:=C/R$. Initially, around $p:=\alpha/2$ fraction of the bits are returned by the channel. We have that of the $d-1$ nodes that the target bit is connected to, around $dp$ are recovered perfectly. In this case, roughly $dp/2$ send a message of $r_i=\infty$ and the rest send $r_i=0$. There are around $(1-p)d$ nodes that are undecided and send a message of $1$ into the local neighborhood. Then if we group the terms in the numerator that contain the strong $1/r_i=\infty$ messages with the terms that send the uninformative $r_i=1$, we get the dominating terms in both the numerator and denominator of (\ref{eq:BP_update_large_d}). Let $S'$ be the subset of nodes with $r_i=1$. Given that $|S'|\approx d(1-p)$, the outing message is asymptotically
\[
r_0\approx 1+\frac{\sum_{I\subset S',|I| =d(1-p)/2} 1}{\sum_{I\subset S',|I| \le (d-2)(1-p)/2} 1}= 1+\frac{{d(1-p)\choose d(1-p)/2}}{\sum_{j\le (d-2)(1-p)/2 }{d(1-p)\choose j}}.
\]
By Stirling's approximation, the numerator behaves as
\[
2^{d(1-p)}\sqrt{\frac{2}{d(1-p)\pi}},
\]
and the denominator is roughly
\[
2^{d(1-p)}/2.
\]
Then the triangle to bit message when $X=0$ behaves like
\[
1+2\sqrt{\frac{2}{d(1-p)\pi}}.
\]
Some of the incoming messages to the target bit will cancel each other and the rest will amplify. If $N_0$ is the number of majorities that evaluate to $0$ and $N_1$ is the number of majorities that evaluate to 1, then the decoding error is
\begin{equation}
\frac{1}{1+(1+2\sqrt{\frac{2}{d(1-p)\pi}})^{|N_0-N_1|}}.
\label{eq:dec_err}
\end{equation}
If we integrate this expression w.r.t the distribution of $N_0-N_1$ then we get the average error at $x_0$. One can show that the probability that a node agrees with its majority is:
\[
\frac{1}{2}(1+\sqrt{\frac{2}{\pi d}}).
\]
Note that $N_0-N_1$ is asymptotically normal by the CLT. When $X=0$, $N_0-N_1$ has mean $d'\alpha\sqrt{\frac{2}{\pi d}}$ and variance $d'\alpha$ where $d'=d(1-r)$. When $r=1/2$ we get $d'=d/2$ and initially we have $p=1/2$. Thus $N_0-N_1\sim d'\alpha\sqrt{\frac{2}{\pi d(1-p)}}+\sqrt{d'\alpha} Z$ where $Z$ is standard normal. We can write this as $N_0-N_1\sim \sqrt{d'\alpha}(\sqrt{d'\alpha}\sqrt{\frac{2}{\pi d}}+Z)$.

We can integrate (\ref{eq:dec_err}) w.r.t to this density to find the average decoding error after one iteration of BP. Setting $d'=d(1-r)$ and taking the limit as $d\to\infty$, we find that 
\begin{equation}
\lim_{d\to \infty} \int_{-\infty}^\infty \frac{1}{1+(1+2\sqrt{\frac{2}{d(1-p)\pi}})^{|\sqrt{d\alpha/2}(z+\sqrt{\frac{\alpha}{\pi}})|}}
f(z) dz= \int_{-\infty}^\infty \frac{1}{1+e^{|2\sqrt{\frac{2\alpha(1-r)}{\pi(1-p)}}z+\frac{4\alpha(1-r)}{\pi\sqrt{1-p}}|}}
f(z) dz.
\label{eq:dynamics_large_d}
\end{equation}
This integral gives an upper bound on the decoding error of BP in the asymptotic regime of large $d$. Fig.~\ref{fig:upper-bound-sys} shows the above bound versus the empirical performance of systematic LDMC(17). BP converges fast for systematic LDMCs, which explains the accuracy of this one step prediction. 

\begin{figure}[ht]
\centering
\includegraphics[width=0.5\textwidth]{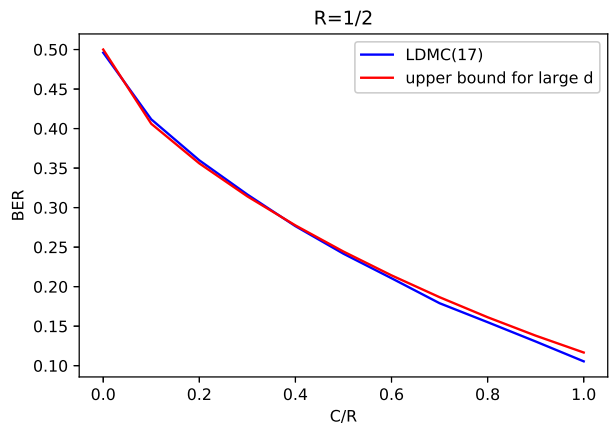}
\caption{The empirical performance of systematic LDMC(17) after 5 iterations along with the predicted on step error in the large $d$ limit obtained from (\ref{eq:dynamics_large_d}) based on one step of BP.}
\label{fig:upper-bound-sys}
\end{figure}

\section{Improving LDGMs using LDMCs}
\label{sec:comp_ldgm}
In this section we study LDGMs and code optimization. Recall that LDGM($d$) is the  ensemble of degree $d$ generated by the parity function. We show that a joint design over LDGMs and LDMCs can uniformly improve the performance of LDGMs for all noise levels in some simple settings.  

As discussed in the introduction, LDGMs are some of the most widely used families of linear codes. They are known to be good both in the sense of coding \cite{mackay1999good} and compression \cite{wainwright2005lossy}. In fact, \cite{mackay1999good} shows that LDGM($d$) (for odd\footnote{When $d$ is even the all one vector is in the kernel of the generator matrix. This implies that BER is 1/2.} $d\ge 3$) enjoys, from a theoretical perspective, almost all the good properties of random codes.   Indeed as shown in Fig.~\ref{fig:ldgm_map}, even relatively short LDGMs can achieve reasonably small error under MAP decoding. As the codes get longer, and the degrees grow, the error can be made arbitrarily small for all $C/R>1$. From a practical perspective, however, their decoding is problematic. The problem is that MAP decoding is not easy to implement in practice even for moderate size codes. BP decoding is not feasible either since for such codes, as generated, BP has a trivial local minima in which all bits remain unbiased. One may hope that adding a small number of degree 1 nodes would enable BP to get around this initial fixed point and achieve near optimal performance. Unfortunately, this is not the case. While small perturbations in degree distribution can sometimes lead to huge boosts in performance (e.g. see Fig.~\ref{fig:LDGM_LDMC} for the degree 2 case), it is often not enough to reach MAP level performance. In general. boosting BP performance for LDGMs is a non-trivial task that often involves some careful code optimization with many relevant parameters. We briefly discuss this matter next.

\begin{figure}[ht]
\centering
\includegraphics[width=0.5\textwidth]{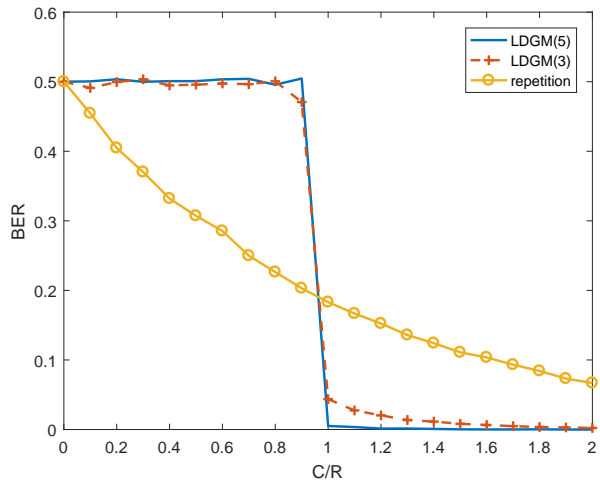}
\caption{The empirical performance of LDGM($d$) for $d=1,3,5$ under (bitwise) MAP decoding for rate $R=1/2$ and $k=1000$. The bitwise decoding is obtained by computing $\hrank$. LDGMs quickly achieve a threshold-like performance as their degree increases. 
}
\label{fig:ldgm_map}
\end{figure}

To understand how LDGM($d$) behaves under BP, we first construct its erasure function and then appeal to Proposition \ref{prop:approx_gives_lower_bound}. With the notation of Fig.~\ref{fig:comp_trees}, we note that a parity check $\Delta_j$ of degree $d$ can determine a target bit $X_0$ if all of its $d-1$ leaf bits $\partial_j X_0$ in the local neighborhood are unerased. Otherwise, it sends an uninformative message. Thus if $1-q$ is the probability of erasure coming into the local neighborhood after some iterations of BP,  then at the next iteration the parity check determines the bit with probability $q^{d-1}$. This gives the $i$-th erasure polynomial
$
E^\BEC_i(q)=1/2(1-q^{d-1})^i.
$
Since the variable node degrees are Poisson distributed (with parameter $\alpha d$), we obtain the following erasure function
\begin{equation}
E_{\mathrm{LDMG}(d)}^\BEC(\alpha,q)=\frac{1}{2}\sum_i \PP(\mathrm{Poi}(\alpha d)=i)(1-q^{d-1})^i=\frac{1}{2}e^{-\alpha dq^{d-1}}.
\end{equation}
We can now study, as before, the local dynamics of error under BP. Let $q^\BEC_\ell(x_0)$ be computed as done in Proposition \ref{prop:approx_gives_lower_bound}. We note that $q^\BEC_\ell(0)$ is in this case (asymptotically) exact for predicting BP error, meaning that whenever the computational graph is a tree the average error is equal to (and not just lower bounded by) $(1-q^\BEC_\ell)/2$. This is due to the fact that parity checks preserve the BEC structure of the input messages, hence, we can write $$\delta^\BP_\ell=\frac{1-q^\BEC_\ell(0)}{2}+o(1).$$
Thus for BP to make any progress during decoding, the following  (necessary) contraction constraint must be satisfied:
\begin{equation}
D(\alpha,q)\triangleq\frac{1-q}{2}-E_{(d)}^\BEC(\alpha,q)> 0.
\label{eq:Dq}    
\end{equation}
In other words, we simply want the outgoing error to be less than the error flowing in. We call the above function the D-function of the ensemble. Similarly, we can define the truncated D-function:
\begin{equation}
D_{\le d}(\alpha,q)\triangleq\frac{1-q}{2}-E_{\le d}^\BEC(\alpha,q).
\label{eq:Dq_truncated}    
\end{equation}
For iterative decoding to take off, we need $D(0)> 0$. Then the first point where $D(q)=0$ occurs determines the limiting performance of BP. Fig.~\ref{fig:Dq} shows the results for LDMC(3) and a mixed LDGM ensemble, which is defined as follows. Let $\Lambda$ be a degree distribution over check degrees. An LDGM ensemble is said to be $\Lambda$-mixed if each check node in the code is selected i.i.d from $\Lambda$.

\begin{figure}[ht]
\centering
\subfloat[$D(\alpha,q)$ for mixed LDGM at $\alpha=1.1$]{\includegraphics[width=0.5\textwidth]{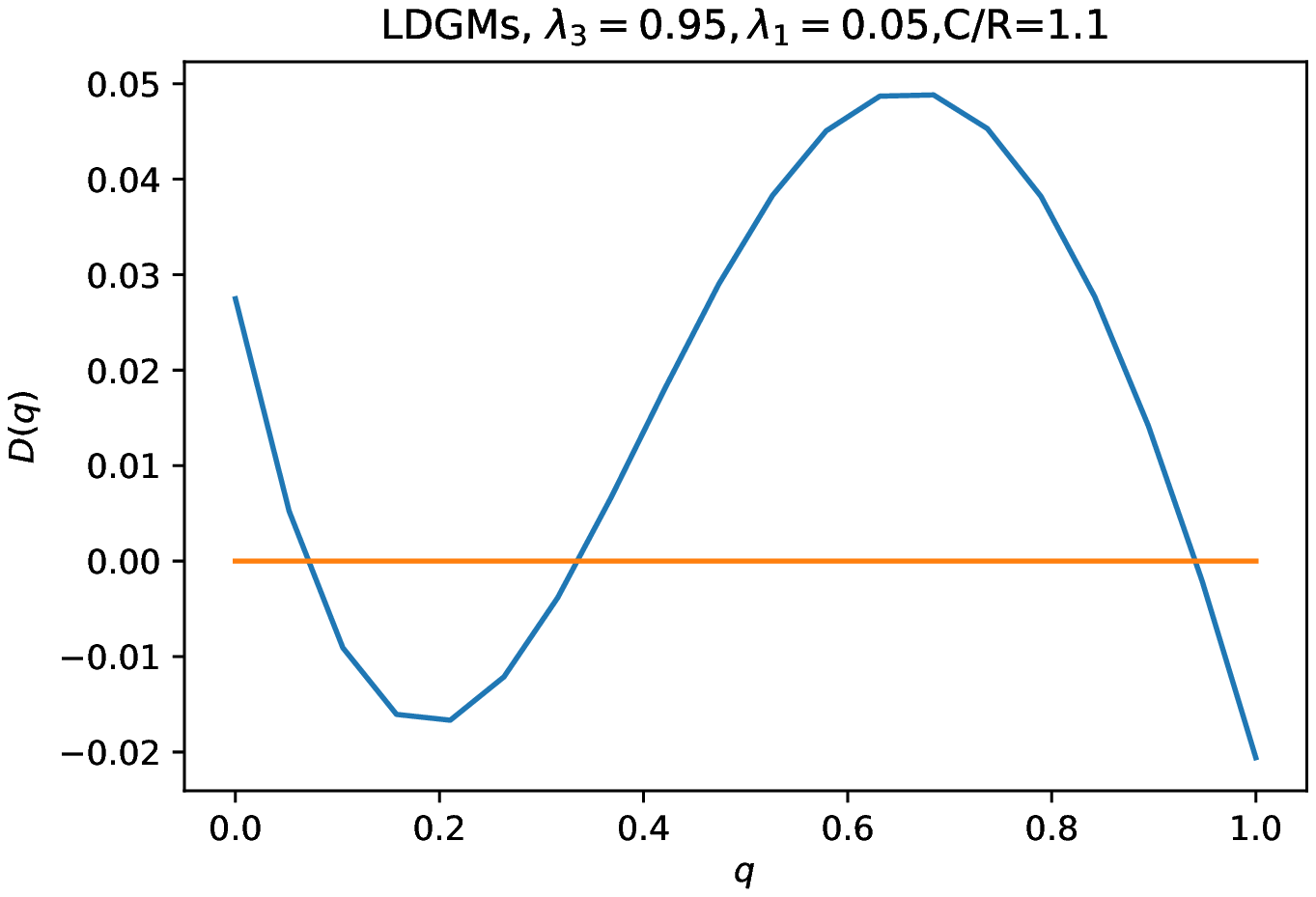}}
\subfloat
[$D_{\le 10}(\alpha,q)$ for LDMC(3) at $\alpha=1.1$]{
\includegraphics[width=0.5\textwidth]{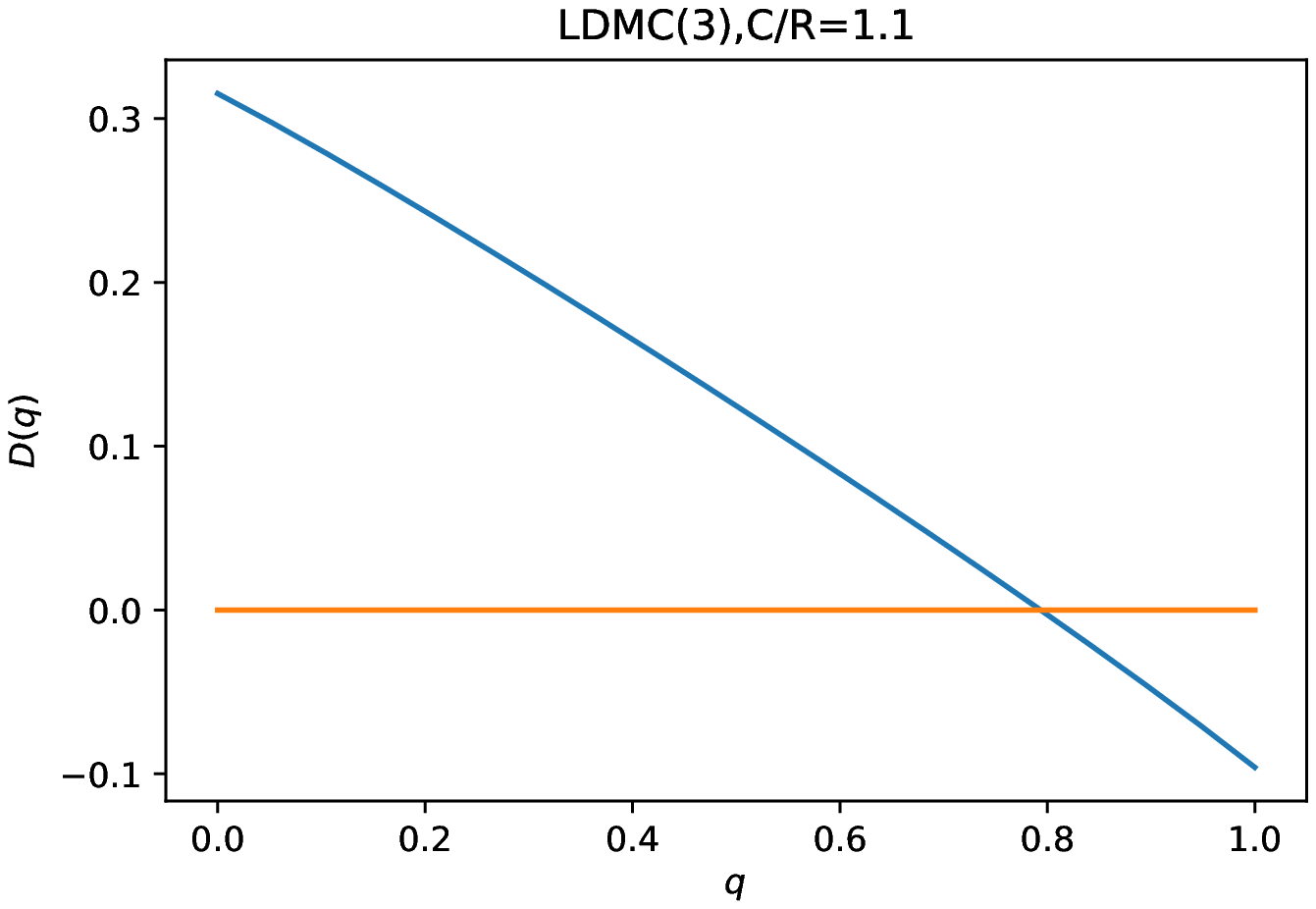}}
\caption{The D-functions from (\ref{eq:Dq})-(\ref{eq:Dq_truncated}) viewed as a function of $q$ for $\alpha=1.1$. The first zero of $D(q)$ determines the fixed points of BP dynamics in (\ref{eq:qbec}). The fixed points remain stable with respect to small variations in the truncation degree $d=10$. Two ensembles are considered: a) A mixed LDGM ensemble using $\lambda_1=0.05$ fraction of degree 1 nodes and $\lambda_3=0.95$ fraction of degree 3 nodes. The degree 1 nodes are needed so that $D(0)>0$ is satisfied. It can be seen that BP has fixed point near 0. Thus small perturbations in degree distributions cannot help BP reach the MAP level of performance shown in Fig.~\ref{fig:ldgm_map}, which corresponds to the right most zero of the D-function. More sophisticated optimization is required to improve the performance. b) The LDMC(3) ensemble. As expected from the observations in Figs. \ref{fig:BP_fixed}-\ref{fig:ldmc_d5vsd3}, $D(q)$ has a unique fixed point. See also Conjecture \ref{conj:monotone}. Furthermore, the relatively large value of $D(0)$ suggests that the convergence is fast for this ensemble. }
\label{fig:Dq}
\end{figure}

 It is easy to define the erasure function of such an ensemble in terms of the erasure function of its regular components. Let $\lambda_i:=\PP(\Lambda=i)$. Suppose that $\Lambda$ has finite support with cardinality $m$. Then the erasure polynomial of an $\Lambda$-mixed LDGM ensemble is simply
\begin{equation}
E_\Lambda^\BEC(\alpha,q)=2^{m-1}\prod_{i=1}^m E_{\mathrm{LDGM}(i)}^\BEC(\alpha \lambda_i,q).
\label{eq:Ebec_mixed}    
\end{equation}
We note that this expression is half the probability that a variable node receives an uninformative message from each component of the code in the ensemble. The code optimization problem can now be formulated in terms of the dynamical system in (\ref{eq:qbec}) associated with this E-function. Suppose that we want to run $\ell$ iterations of BP to decode an LGDM. Let $q^\BEC_{\ell,\alpha}(0)$ be density of unerased bits after $\ell$ iterations with $C/R=\alpha$. If we are interested in minimizing the BP error at two different $C/R$'s, say $\alpha_1$ and $\alpha_2$, then the following optimization problem becomes relevant
\begin{align*}
\textup{maximize}_{\Lambda}\,\,& q_{l,\alpha_1}^{\BEC}(0)+q_{l,\alpha_2}^{\BEC}(0)\\
&\sum_{i} \lambda_i=1\\
&\lambda_i\ge 0.
\end{align*}

This is a non-convex optimization problem. We can solve it up to local optimality using gradient descent. Solving for $\alpha_1=0.9,\alpha_2=1.1$ over LDGM($d$) with $d\le 3$, we find that $\lambda_1=0.08,\lambda_2=0.22,\lambda_3=0.7$. The D-curve and the corresponding performance are shown in Fig.~\ref{fig:LDGM_LDMC_p9}. The same figure demonstrates the performance after we simply remove the lower degree parities by setting $\lambda_1=0,\lambda_2=0$, and replace them with an LDMC(3). We can also jointly optimize over the LDGM/LDMC(3) ensemble.

Since LDMC(3) dominates the repeition code everywhere, we expect this new LDGM/LDMC ensemble to have lower error than the pure LDGM ensemble. Another possible solution for LDGM optimization is $\lambda_1=0.1099,\lambda_2=0.1409,\lambda_3=0.7492$. The performance is shown in Table \ref{table:LDGM_LDMC}, which demonstrates that the joint design uniformly dominates the LDGM ensembles. We note that for LDMCs we compute the performance by simulating over $10^5$ data bits across 10 trials using 80 iterations of BP, whereas for LDGMs we compute the theoretical limits from the fixed points of (\ref{eq:qbec}). The attainable performance in practice is slightly worse than fixed point prediction and often requires minor adjustments to the degree distributions. We can see that the LDGM family exhibits a sharp transition at the end point $C/R=0.9$ while the combined ensembles degrades more smoothly beyond this point while maintaining smaller error everywhere else. Fig.\ref{fig:LDGM_LDMC_p9}b shows that the D-curve with optimal parameters almost touches the $x$-axis for some small $q$ when $\alpha=0.9$. This is an artifact of the optimal designs. Such proximity with zero induces a near fixed point, from which BP requires many iterations to escape until it reaches good performance. 

As another example, setting $\ell:=40,\alpha_1:=0.7,\alpha_2:=1.1$ in the above optimization problem and solving for $\lambda_i$'s in the joint LDGM(d)/LDMC(3) ensemble (with $d\le 3$) gives $\lambda_1=0.0,\lambda_2=0.0,\lambda_3=0.42$ and  $\lambda_{\mathrm{LDMC(3)}}=0.58$. For the LDGM ensemble (with $d\le 3$) we get $\lambda_1=0.003,\lambda_2=0.988,\lambda_3=0.009$. The results are shown in Fig.~\ref{fig:LDGM_LDMC}. We note that the LDGM ensemble has a small fraction of degree 3 and degree 1 nodes. This fact seems to be consistent w.r.t to changes in $\ell$ or the initial points of gradient descent. We further note that an LDGM ensemble with $\lambda_2=1$ has a $\BER$ of 1/2 for all $C/R$. The small fraction of degree $1$ nodes are needed both to initialize the iterative decoding and to break the symmetries in the solution space of LDMC(2).


\begin{figure}[ht]
\centering
\subfloat[LDGM vs  LDGM/LDMC.]{\includegraphics[width=0.55\textwidth]{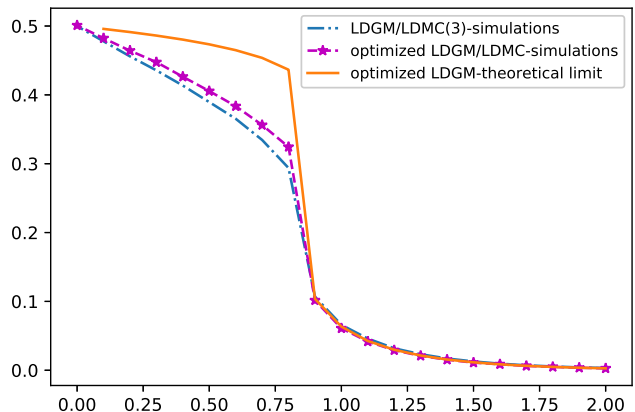}}
\subfloat
[D-curve at $\alpha=0.9$ for the optimized LDGM.]{
\includegraphics[width=0.55\textwidth]{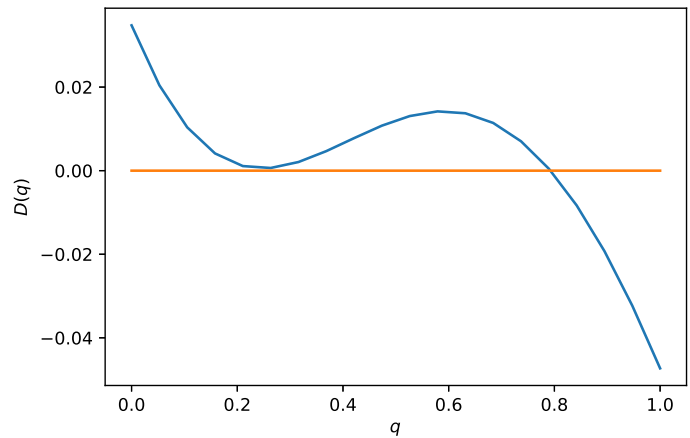}}\\
\subfloat
[BP error vs time steps at $\alpha=0.9$.]{
\includegraphics[width=0.55\textwidth]{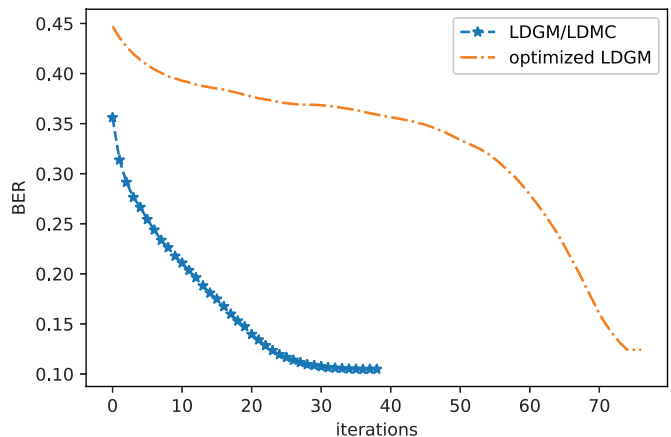}}

\caption{a) Code performance is compared for optimized LDGMs (over degrees $\le 3$), combined  LDGM/LDMC, and optimized LDGM/LDMC ensembles. The LDGM ensemble is optimized for $\alpha_1=0.9$ and $\alpha_2=1.1$. It has parameters $\lambda_1=0.08$, $\lambda_2=0.22$, $\lambda_3=0.7$. The LDGM curve shows the theoretical limit predicted by density evolution and corresponds to the fixed points of the dynamical system in (\ref{eq:qbec}). The LDGM/LDMC curve is obtained by simulations using 60 steps of BP for $k=10^5$ data bits. This ensemble is created by replacing the degree 1 and 2 components of the optimized ensemble with LDMC(3), i.e., it has parameters $\lambda_3=0.7$ and $\lambda_{\mathrm{LDMC}(3)}=0.3$. Since LDMC(3) dominates repetition for all noise levels, it is reasonable to expect that the combined LDGM/LDMC ensemble performs better. The optimized LDGM/LDMC curve shows simulated performance using 100 BP iterations for $k=10^5$ data bits. The ensemble has parameters $\lambda_3=0.75$ and $\lambda_{\mathrm{LDMC}(3)}=0.25$ obtained by optimizing the combined ensemble at the same $C/R$ points mentioned above. b) The D-curve for the optimized LDGM ensemble is shown. The low values of $D$ (relative to the position of the fixed point) and the near zero point at $q\approx 0.25$ indicate that BP requires many iterations to converge. This behavior is typical for optimal designs shown at the bottom figure. The fixed point near $0.8$ is compatible with the error of $0.105$ obtained at $\alpha=0.9$ on the left. c) The progress of BP error is shown at $C/R=0.9$. The nearly flat region in the error curve can be explained by the presence of a near fixed point in the D-curve. Overall, the experiments suggest that the joint ensembles converge  faster and achieve better performance.}
\label{fig:LDGM_LDMC_p9}
\end{figure}

\begin{table*}[t]
\centering
\begin{tabular}[width=0.75\textwidth]{|l|l|l|l|l|l|l|l|}
\hline
$C/R$&2 & 1.5  &1.2& 1.1& 1& 0.9&0.8  \\
 \hline
 optimized LDGM I& 0.00279 &0.0115  &   0.0297&0.0425 & 0.06336&0.10367&0.43652\\
 \hline
 optimized LDGM II& 0.00268  &0.01118  & 0.02915& 0.04183& 0.06278&0.10459&0.42548 \\
 \hline
 optimized LDGM/LDMC(3)& 0.00277&0.01107&0.02906&0.04173&0.06093&0.01016&0.3244\\
\hline 
\end{tabular}
\caption{The empirical performance of optimized LDGM/LDMC(3) ensemble is compared with theoretical limits of two optimized LDGM ensembles. The LDGM/LDMC(3) ensemble has parameters $\lambda_{\mathrm{LDMC}(3)}=0.25$ and $\lambda_{3}=0.75$. The performance is obtained by averaging over $10^5$ data bits across 10 trials using 100 iterations of BP. The optimized LDGM I ensemble has parameters $\lambda_1=0.08,\lambda_2=0.22,\lambda_3=0.7$ and LDGM II has parameters $\lambda_1=0.1099,\lambda_2=0.1409$,$\lambda_3=0.7492$. The theoretical bounds are computed using fixed points of (\ref{eq:qbec}).}
\label{table:LDGM_LDMC}
\end{table*}

\begin{figure}[ht]
\centering
\includegraphics[width=0.5\textwidth]{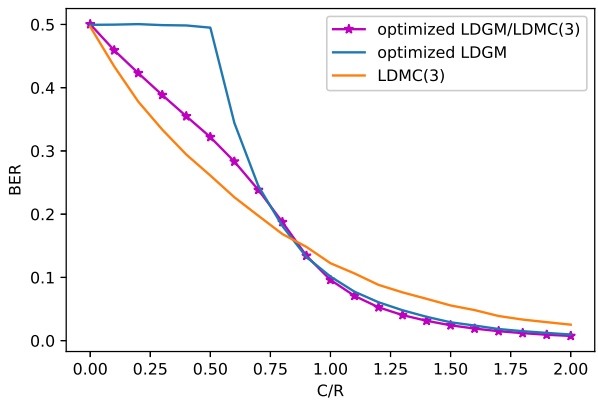}
\caption{BER curves for optimized LDGM/LDMC is compared with optimized LDGM. All curves are obtained by simulations on codes with  $k=200000$ data bits. The code optimization is carried over codes with degrees $d\le 3$. The optimized joint design has degree distribution $\lambda_1=0.0,\lambda_2=0.0,\lambda_3=0.42$ and  $\lambda_{\mathrm{LDMC(3)}}=0.58$. The optimized LDGM has degree distribution  $\lambda_1=0.003,\lambda_2=0.988,\lambda_3=0.009$. The codes are optimized to minimize the sum of BERs at $\alpha_1=0.7$ and $\alpha_2=1.1$. }
\label{fig:LDGM_LDMC}
\end{figure}


\section{Codes as channel transforms}
\label{sec:channel_transform}
In this section we study LDMCs from the perspective of a channel transform. This notion arises when one employs a concatenated code. Concatenated codes are the codes that act on pre-coded information. This means that the input to the code is not an arbitrary point in the alphabet space, but rather the codeword of an outer code. This technique is often used to design codes with high performance and low decoding complexity. For instance, to approach the capacity of the erasure channel with LDPCs one needs to use high degree variable nodes. These in turn create short cycles in the computational graph of the BP decoder, which is problematic for accuracy of BP. To mitigate the impact of cycles, one needs to use very large codes and many iterations of BP, leading to long delays in the communication system as well as an expensive decoding procedure. A common method to circumvent these difficulties is to employ a two (or more) layer design. A low complexity inner code $f_i:\mathcal{A}^k\to\mathcal{A}^n$ is used to reduce the channel error, without necessarily correcting any erasure pattern. Then an outer error correcting code $f_o:\mathcal{A}^{m}\to\mathcal{A}^k$ cleans up the remaining error. The outer code here can be an LDPC but one that faces a weakened channel, hence, it requires fewer BP iterations and can be made to be shorter. It can also be a (short) error correcting code that relies on syndrom decoding. In either case, the overall communication path looks like the following $$\mathcal{A}^{m}\stackrel{f_o}{\to} \mathcal{A}^{k}\underbrace{\stackrel{f_i}{\to} \mathcal{A}^{n}\stackrel{\mathrm{BEC}_\epsilon}{\to} Y\stackrel{g_i}{\to}}_{Q(\epsilon):\text{channel transform}} \mathcal{B}^k\stackrel{g_o}{\mapsto} \mathcal{A}^m.$$ Here $Y$ is the outcome of the channel, $g_i$ is inner decoder and $g_o$ is the outer decoder. The domain of the outer decoder is chosen to be different from the alphabet of the message space on purpose. This is to accommodate various decoding messages that maybe transmitted from the inner decoder to the outer decoder. Two common choices in the literature are: 1) hard decision decoding ($\mathcal{B}=\mathcal{A}$);  in this case the inner decoder can only transmit a hard decision on each bit to the outer decoder corresponding to its best estimate of what the bit value is. 2) soft-decision decoding( $\mathcal{B}=\mathbb{R}$); in this case the inner decoder is allowed to send the bitwise probabilities of error to the outer decoder. In either case, we can view the action of the inner code together with its decoder as one channel $Q$. 

For hard decision decoding, it is clear that the channel (after interleaving) is a BSC with crossover probability equal to BER. For soft-decision decoding, the output of the channel transform is a sequence of probabilities. We view this channel as a product channel that sends the marginals on every bit to the outer decoder $Q(\epsilon):\mathcal{A}^k\to \prod_{i=1}^k \pi_i$. In practice, often an interleaver is placed between the inner and outer decoder to ensure that the bit errors are not correlated, hence, it makes sense to model the action of the inner code with a product channel. To study the performance of codes as channel transforms under soft decision decoding we introduce the notion of soft information
\[
\iota(\epsilon)\triangleq1-\EE[\frac{1}{k} \sum_{i=1}^k h_i(\epsilon)],
\]
where $h_i=h_b(\pi_i)$ is the binary entropy of the $i$-th marginal produced by $Q(\epsilon)$. The soft information can be seen as the average per-bit information sent from the inner code to the hard decision (outer) decoder. If the inner code is wrapped with an interleaver, $\iota$ will closely approximate the capacity of the inner channel $Q(\epsilon)$. In this case, two information bits of the inner code are likely to fall in different blocks of the outer error correcting code. Hence, the possible dependencies between the bits is not relevant.

\subsection{H-functions}
We shall also be interested in estimating the amount of soft information that an inner code produces. Such estimates can prove fruitful in code design under soft decoding in the same manner that the bounds from E-functions were used in code optimization in the previous section. To obtain these estimates, we first define the notion of E-entropies:
\begin{definition}[E-entropies]
Take the setup of Definition \ref{def:erasure_function}. Let $Y_i$'s be $\BEC$ induced observations as before. Define the $d$-th erasure entropy polynomial to be
\[
H_d^\BEC(q)\triangleq\EE[h_b(S_0|M_1,\cdots,M_d,Y^{(1)},\cdots,Y^{(d)})],
\]
where the expectation is taken w.r.t the ensemble distribution. Similarly, the $d$-th truncated erasure entropy polynomial $H^\BEC$ is defined as
\[
H_{\le d}^\BEC(\alpha,q)\triangleq\sum_{k\le d} \P(\mathrm{Deg}=k)H_k^\BEC(q).
\]
The (untruncated) erasure entropy function $H^\BEC(q,\alpha)$ is defined as the limit (in $d$) of $H_{\le d}^\BEC(q,\alpha)$. The corresponding error entropies $H_d^\BSC$ are defined in an analogous way to $E^\BSC$ by replacing $\BEC$ channels  with $\BSC$s as done in Definition \ref{def:error_function}. Similarly, we have the $d$-th truncated error entropy function:
\[
H_{\le d}^\BSC(\alpha,q)\triangleq\sum_{k\le d} \P(\mathrm{Deg}=k)H_k^\BEC(q)+ \PP(\mathrm{Deg}>d).
\]
Finally, the (untruncated) error entropy function $H^\BSC(\alpha,q)$ is defined as the limit of $H^\BSC_{\le d}(\alpha,q)$. 
\label{def:H_func}
\end{definition}
We define the $\chi^2$-entropy functions $\mathcal{H}$ analogously.
\begin{definition}[$\chi^2$-entropies]
Take the setup of Definition \ref{def:erasure_function}. Let $Y_i$'s be $\BEC$ induced observations as before. Define the $d$-th erasure $\chi^2$-entropy polynomial to be
\[
\mathcal{H}_d^\BEC(q)\triangleq\EE[1-I_{\chi^2}(S_0;M_1,\cdots,M_d,Y^{(1)},\cdots,Y^{(d)})],
\]
where the expectation is taken w.r.t the ensemble distribution. Similarly, the $d$-th truncated erasure $\chi^2$-entropy polynomial $\mathcal{H}^\BEC$ is defined as
\[
\mathcal{H}_{\le d}^\BEC(\alpha,q)\triangleq\sum_{k\le d} \P(\mathrm{Deg}=k)\mathcal{H}_k^\BEC(q).
\]
The (untruncated) erasure $\chi^2$-entropy function $\mathcal{H}^\BEC(q,\alpha)$ is defined as the limit (in $d$) of $\mathcal{H}_{\le d}^\BEC(q,\alpha)$. The corresponding error $\chi^2$-entropies $\mathcal{H}_d^\BSC$ are defined in an analogous manner to Definition \ref{def:H_func}. 
\label{def:H_chi_func}
\end{definition}

\subsection{Bounds via channel comparison lemmas}
By applying the less noisy parts of Lemma \ref{lem:lemma_A} and Lemma \ref{lem:lemma_B}, we can prove the following bounds on soft information. The idea is that in each iteration of BP we replace the $T_\ell$ channels discussed in the proof of Proposition \ref{prop:approx_gives_lower_bound}, with the least (resp. most) noisy channels while matching $\chi^2$-capacities to obtain upper (resp. lower) bounds on soft information.

\begin{proposition}
Consider the dynamical systems
\begin{equation}
q^\BEC_{t+1}(x_0)=1-\
\mathcal{H}_{\le d}^\BEC(\alpha,q^\BEC_{t}),
\label{eq:qbec_chisoft}
\end{equation}
and
\begin{equation}
q^\BSC_{t+1}(x_0)=1/2-1/2\sqrt{1-\mathcal{H}_{\le d}^\BSC(\alpha,q^\BEC_{t})},
\label{eq:qbsc_chisoft}
\end{equation}
initialized at $x_0$ with $\alpha=C/R$. 
Let $\iota_\ell^{\mathrm{BP}}$ be the soft information output of a ensemble under $\BP$ after $\ell$ iterations. Likewise, let $\iota^{*}$ be the soft information output under the optimal decoder. Then 
\[
q^\BEC_\ell(1)-o(1)\ge \iota^{*}\ge \iota_\ell^{\mathrm{BP}}.
\]
Furthermore, 
\[
q^\BEC_\ell(0)+o(1)\ge \iota_\ell^{\mathrm{BP}}\ge 1-h_b(q^\BSC_\ell(1/2))+o(1)
\]
with $o(1)\to 0$ and $k\to \infty$.\\
\label{prop:approx_gives_lower_bound_soft}
\end{proposition}
\begin{proof}
The proof is similar to that of Proposition \ref{prop:approx_gives_lower_bound}. The idea is that we replace E-functions with $\mathcal{H}$-functions and then in each iteration of BP we replace the $T_\ell$ channels  with the least (resp. most) noisy channels while matching $\chi^2$-capacities to obtain upper (resp. lower) bounds on soft information. 
\end{proof}
To provide more context for the open question in Remark  \ref{remark:mc_vs_ln}, we further study the soft information dynamics while matching capacities. We have the following conjectured bound for majority codes (or perhaps more generally for binary monotone functions).

\begin{conjecture}
For a code ensemble generated by a monotone function, consider the dynamical systems
\begin{equation}
q^\BEC_{t+1}(x_0)=1-\
H_{\le d}^\BEC(\alpha,q^\BEC_{t}),
\label{eq:qbec_soft}
\end{equation}
and
\begin{equation}
q^\BSC_{t+1}(x_0)=h^{-1}(H_{\le d}^\BSC(\alpha,q^\BEC_{t})),
\label{eq:qbsc_soft}
\end{equation}
initialized at $x_0$ with $\alpha=C/R$. Then
$$
q_\ell^\BEC(0)+o(1)\ge \iota_\ell^{\BP}\ge 1-h_b(q_t^\BSC(1/2))+o(1).$$
\label{conj:qbec}
\end{conjecture}
We skip the computation of H-functions for LDMCs as it is similar to that of E-functions. We simply need to map messages to $\chi^2$-entropies (instead of error probabilities), which is a straightforward extension. 
Fig.~\ref{fig:soft_info_bounds} shows the results. It is natural to conjecture that the $H$-dynamics give the correct upper and lower bounds for monotone functions. We further observe that the associated dynamics for soft information have a unique fixed point similar to what we observed for the error dynamics.
\begin{figure}[ht]
\centering
\subfloat[][soft information bounds for LDMC(3)]{\includegraphics[width=0.5\textwidth]{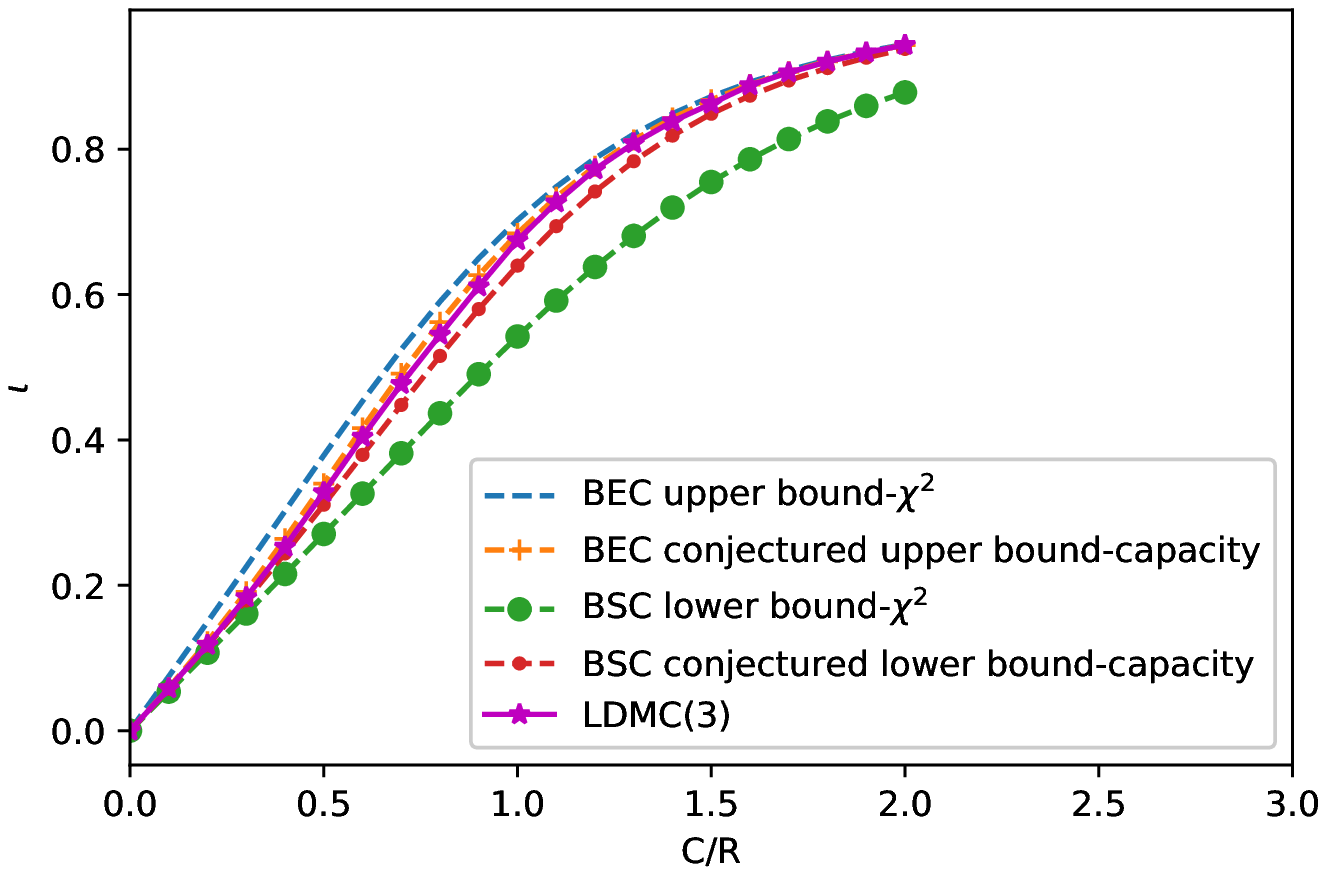}}
\subfloat[][The bound of (\ref{eq:qbec_chisoft}) after 5 iterations]{\includegraphics[width=0.5\textwidth]{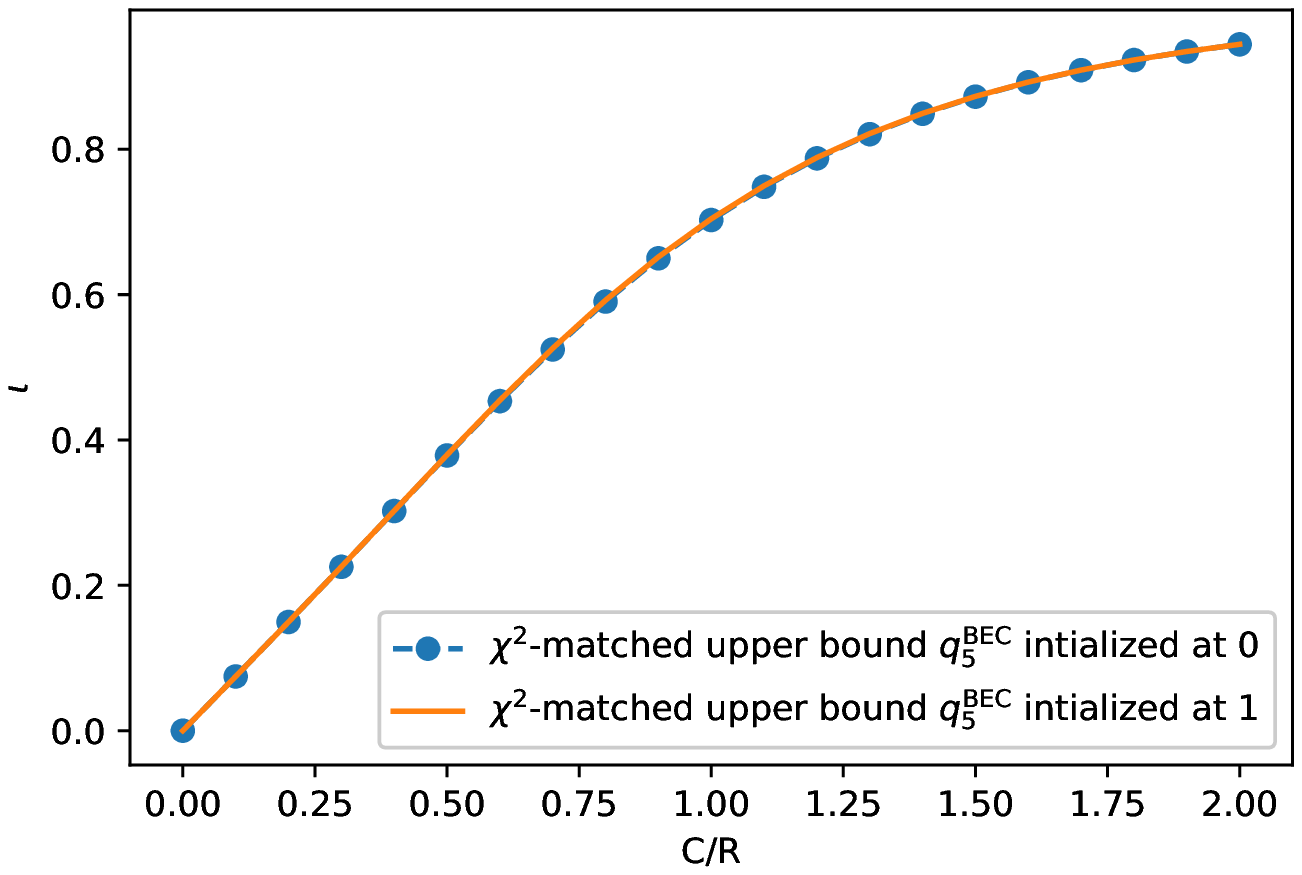}}
\caption{a) Empirical soft information for LDMC(3) with $k=40000$ compared with the lower and upper bounds obtained while matching capacity and $\chi^2$-capacity. The $\chi^2$-bounds are obtained from (\ref{eq:qbec_chisoft})-(\ref{eq:qbsc_chisoft}). The conjectured bounds are obtained by matching capacities as described in (\ref{eq:qbec_soft})-(\ref{eq:qbsc_soft}). b) The soft information dynamics converge to a unique fixed point, independent of the initial condition. The convergence is fast and nearly achieved after only 5 iterations.}
\label{fig:soft_info_bounds}
\end{figure}

We end this section by comparing the soft information of LDMC(3) with linear codes. We note that for a linear code $\iota(\epsilon)=1-2\mathrm{BER}(\epsilon)$. Thus we can use the bounds of Theorem \ref{thm:two_point_converse} together with Proposition \ref{prop:hrank} to obtain similar bounds on soft information.  The results are shown in Fig.~\ref{fig:soft_info}.

\begin{figure}[ht]
\centering
\subfloat[][$R=1/2$]{\includegraphics[width=0.5\textwidth]{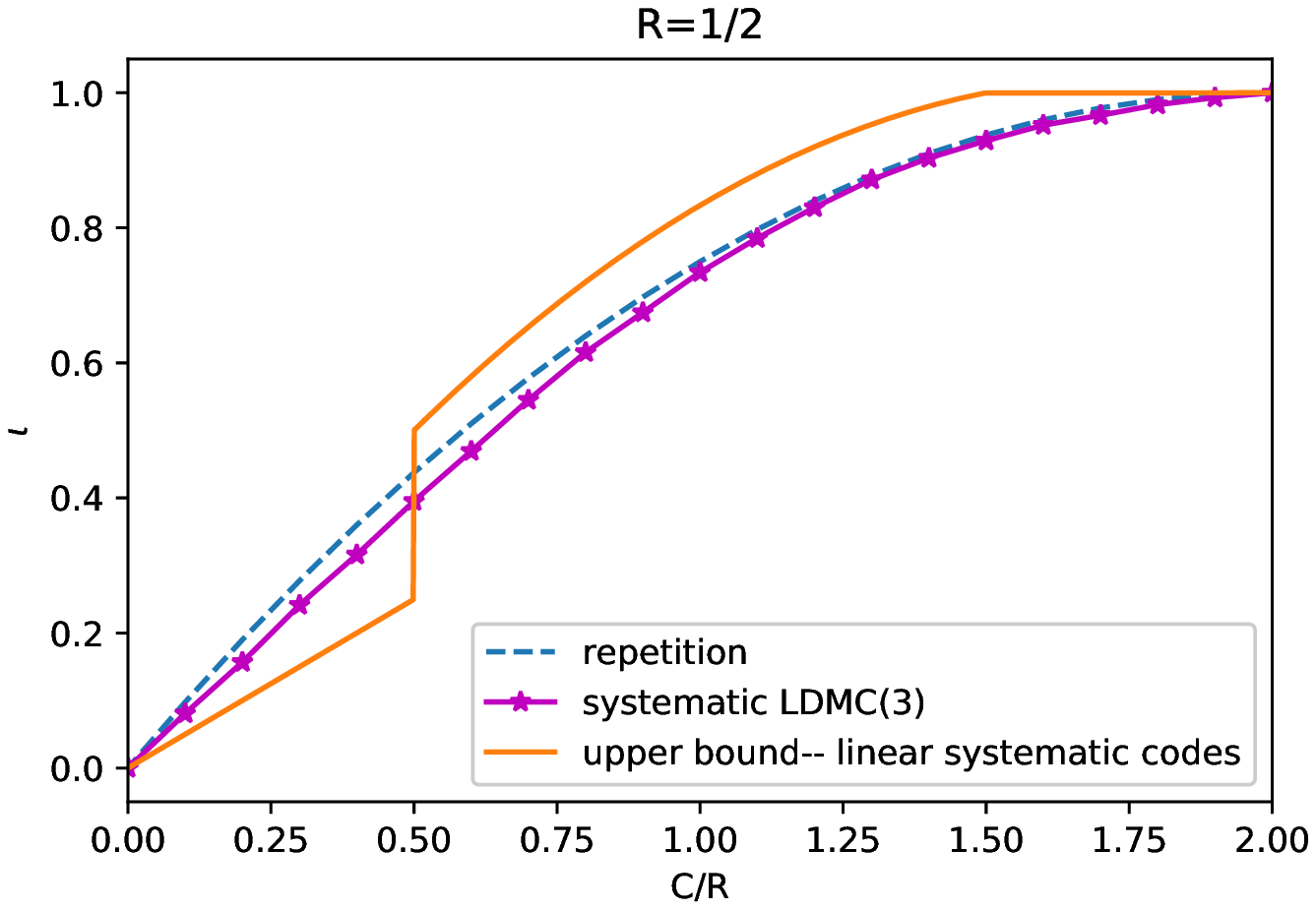}}
\subfloat[][$R=1/5$]{\includegraphics[width=0.5\textwidth]{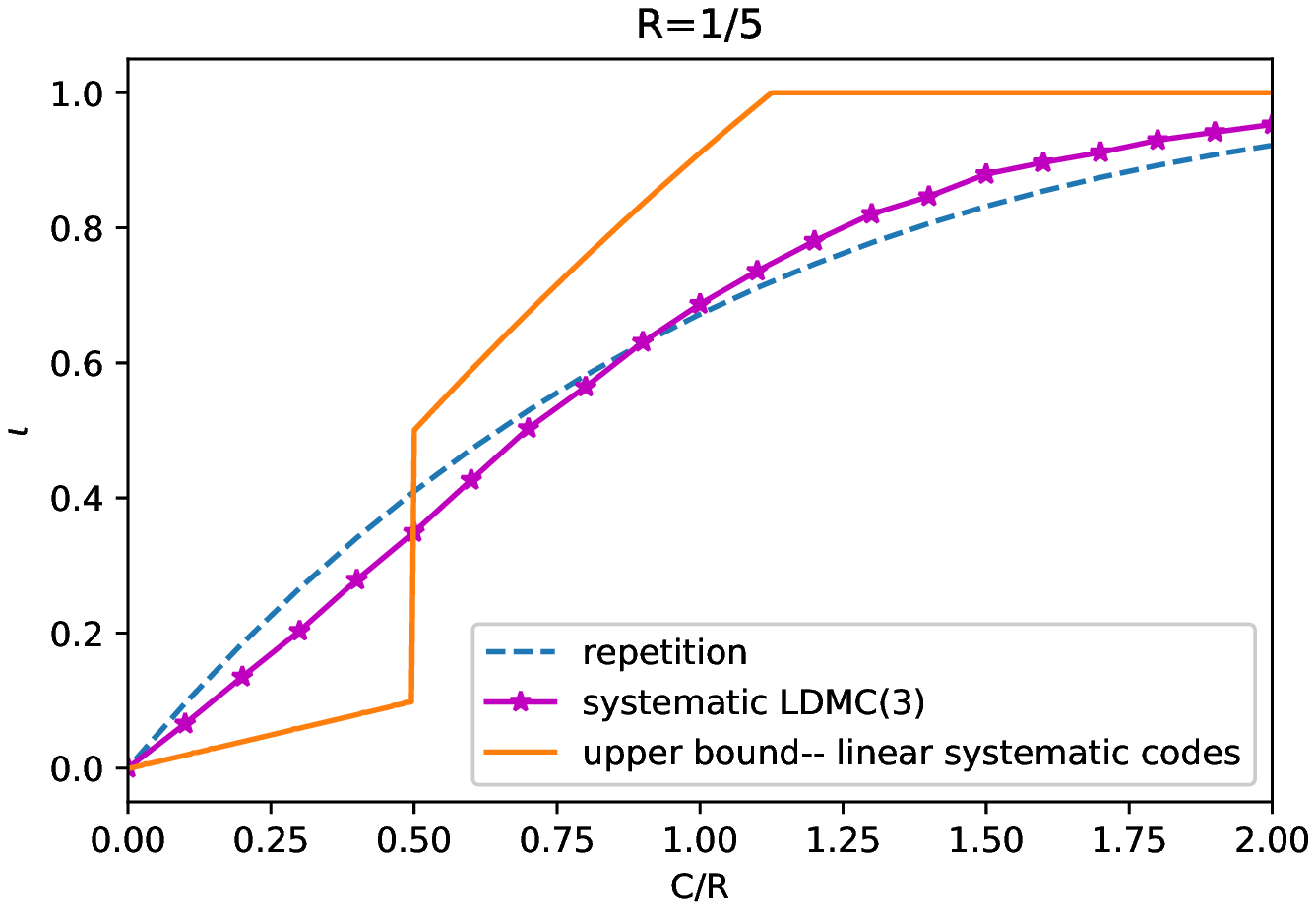}}
\\
\subfloat[][$R=2/3$]{
\includegraphics[width=0.5\textwidth]{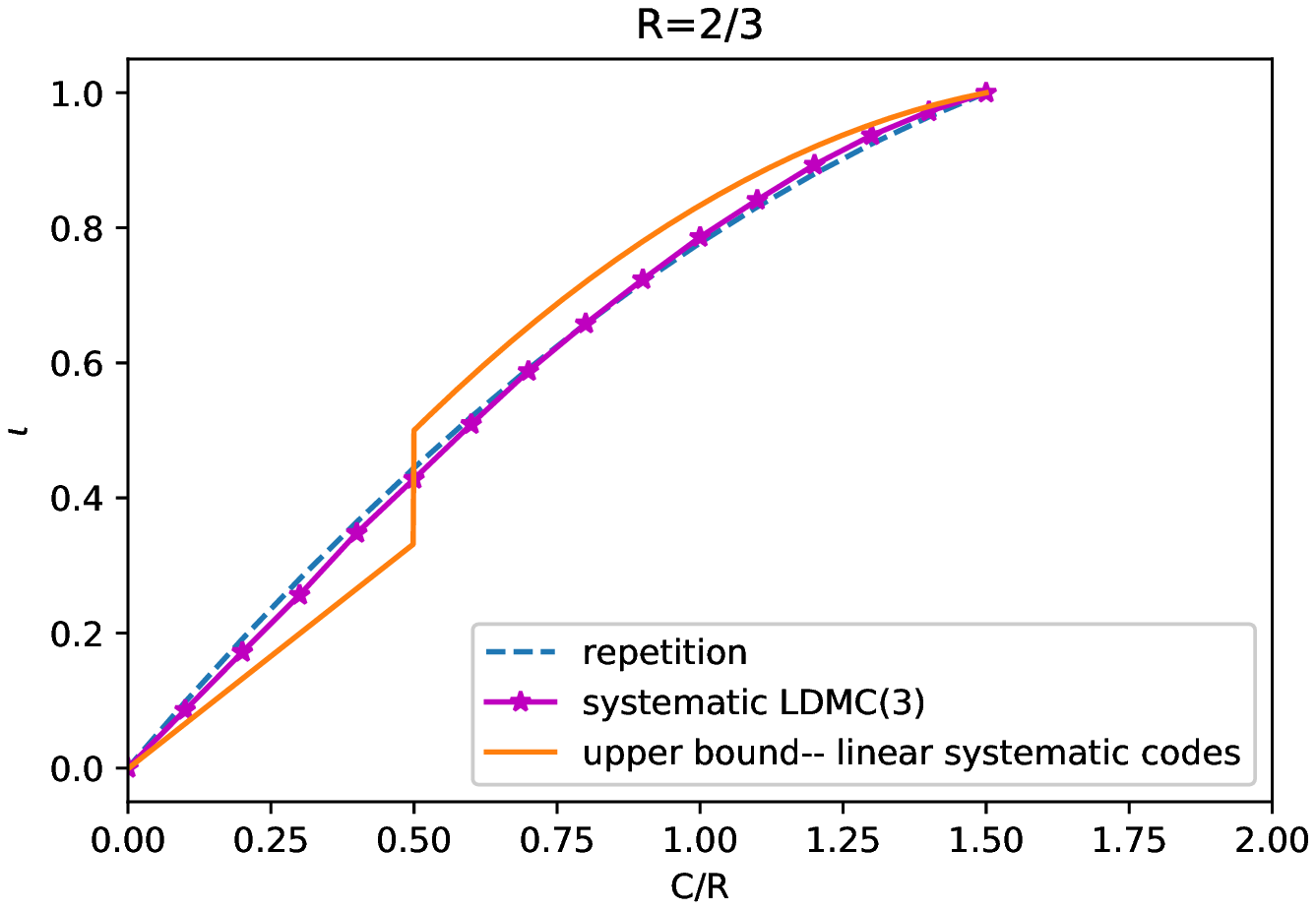}}
\subfloat[][$R\approx 0$]{
\includegraphics[width=0.5\textwidth]{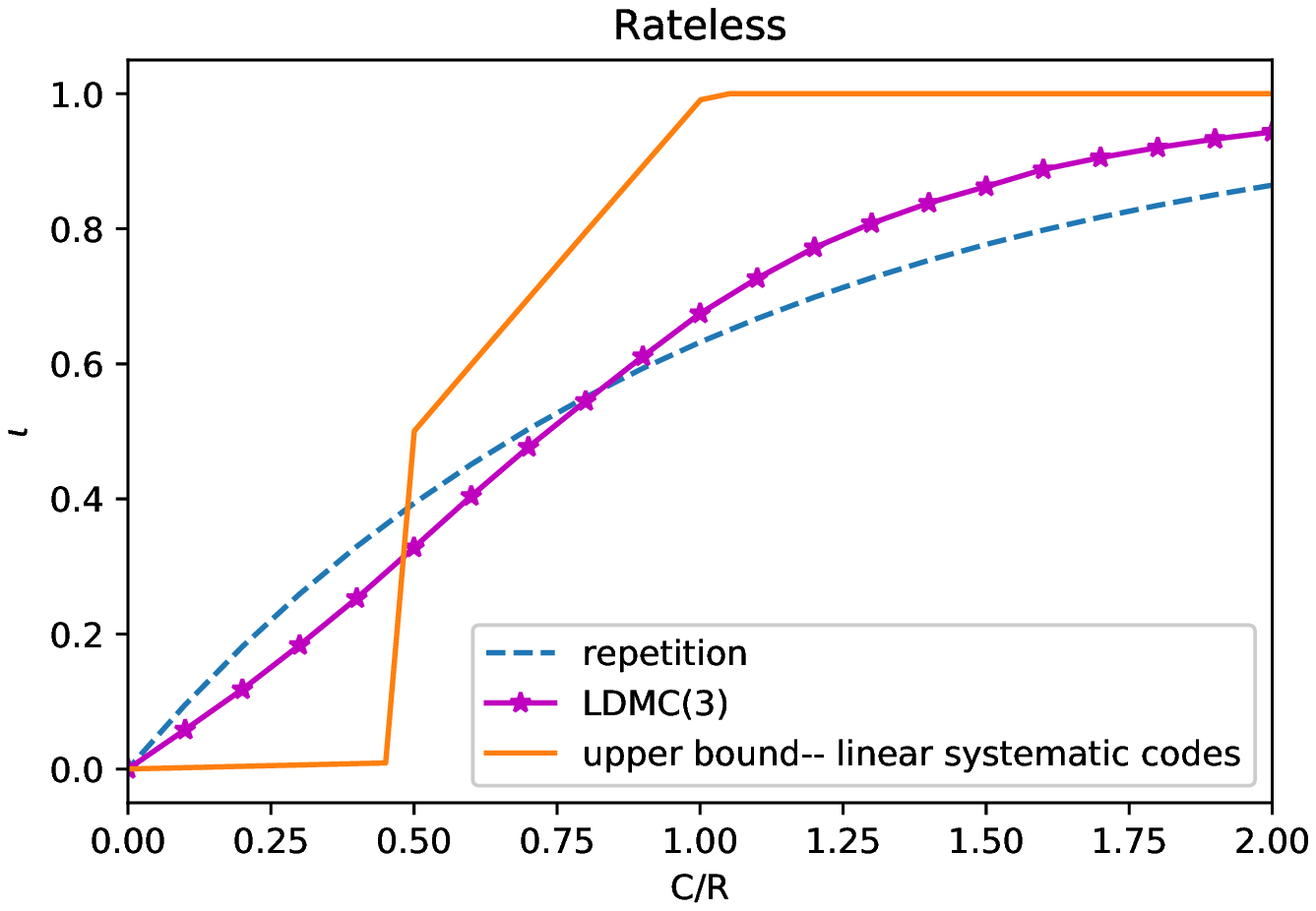}}
\caption{Empirical soft information for LDMC(3) with $k=20000$ compared with repetition and linear codes satisfying $\BER=0.25$ at $C/R=0.5$ for different rates. In (a)-(c) we use a systematic LDMC(3) and in (d) we use the  LDMC(3) ensemble.}
\label{fig:soft_info}
\end{figure}

\subsection{A comment on Fano type bounds}
\label{sec:fano}
 Let $T_\ell$ be the tree channel discussed in the proof of Proposition \ref{prop:approx_gives_lower_bound}. In the previous section, we gave bounds on $I(S_0;T_\ell)$ using the channel comparison lemmas. It is natural to ask how well these bounds bounds combined with a Fano type inequality  would predict the error dynamics. Fano's inequality states that
\begin{equation}
P_e(T_\ell)\ge h^{-1}(1-I(S_0;T_\ell)).
\label{eq:fano_ber}
\end{equation}
Hence, one might expect that a good upper bound on $I(S_0;T_\ell)$ would result in a tight bound on BER. Likewise, it is easy to translate a bound on BER to a bound on soft information using a degradation argument. Since $T_\ell$ is degraded w.r.t $\BEC_{1-2\BER}$, we have
\begin{equation}
I(S_0;T_\ell)\le 1-2\BER.
\label{eq:ber_to_I}    
\end{equation}
As shown in Fig.~\ref{fig:Fano}, such bounds are looser than those obtained from direct arguments as in Proposition \ref{prop:approx_gives_lower_bound} and \ref{prop:approx_gives_lower_bound_soft}. 

\begin{figure}[ht]
\centering
\subfloat[][Comparing the upper bounds in (\ref{eq:ber_to_I}) and (\ref{eq:qbec_chisoft})]{\includegraphics[width=0.5\textwidth]{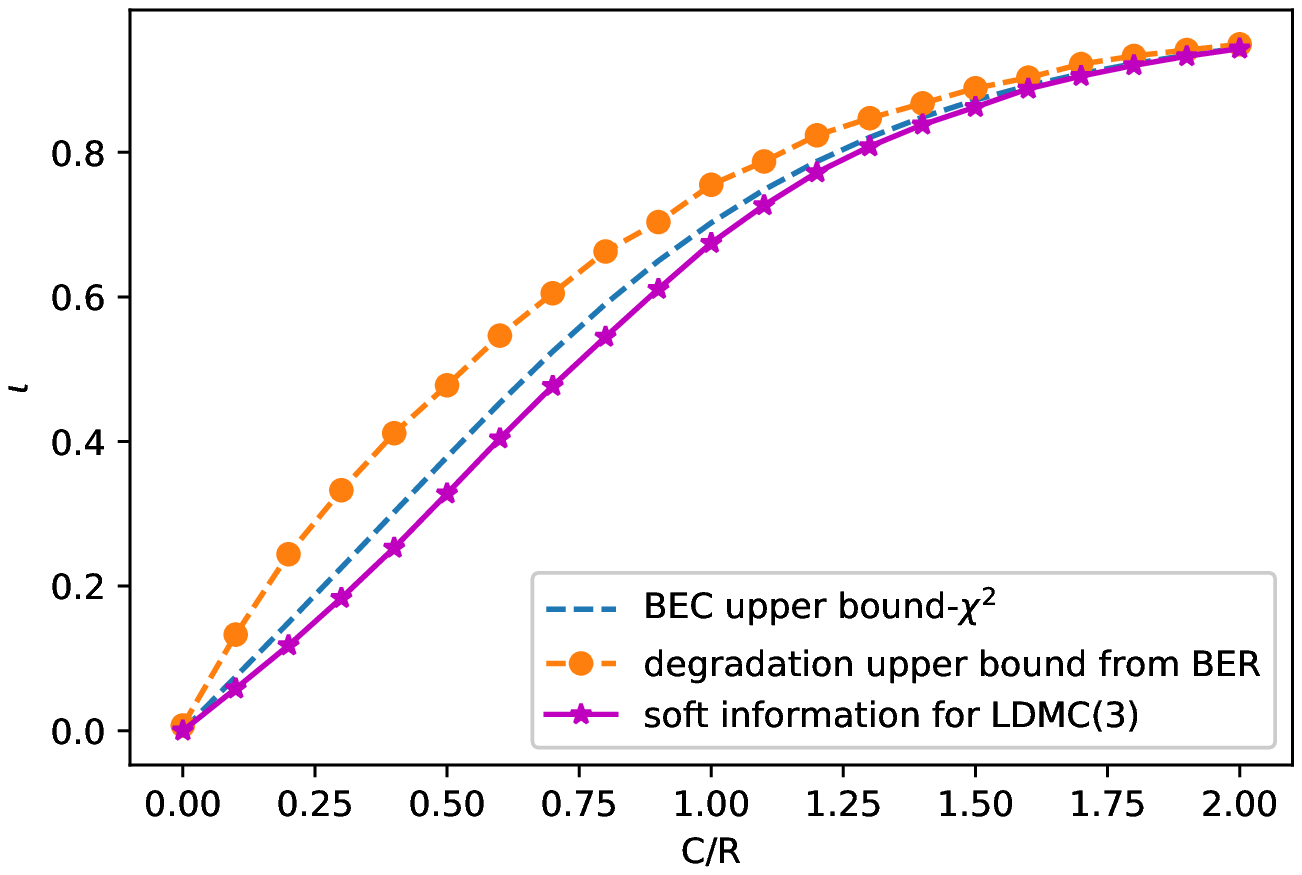}}
\subfloat[][Comparing the lower bounds in (\ref{eq:fano_ber}) and (\ref{eq:qbec}).]{
\includegraphics[width=0.5\textwidth]{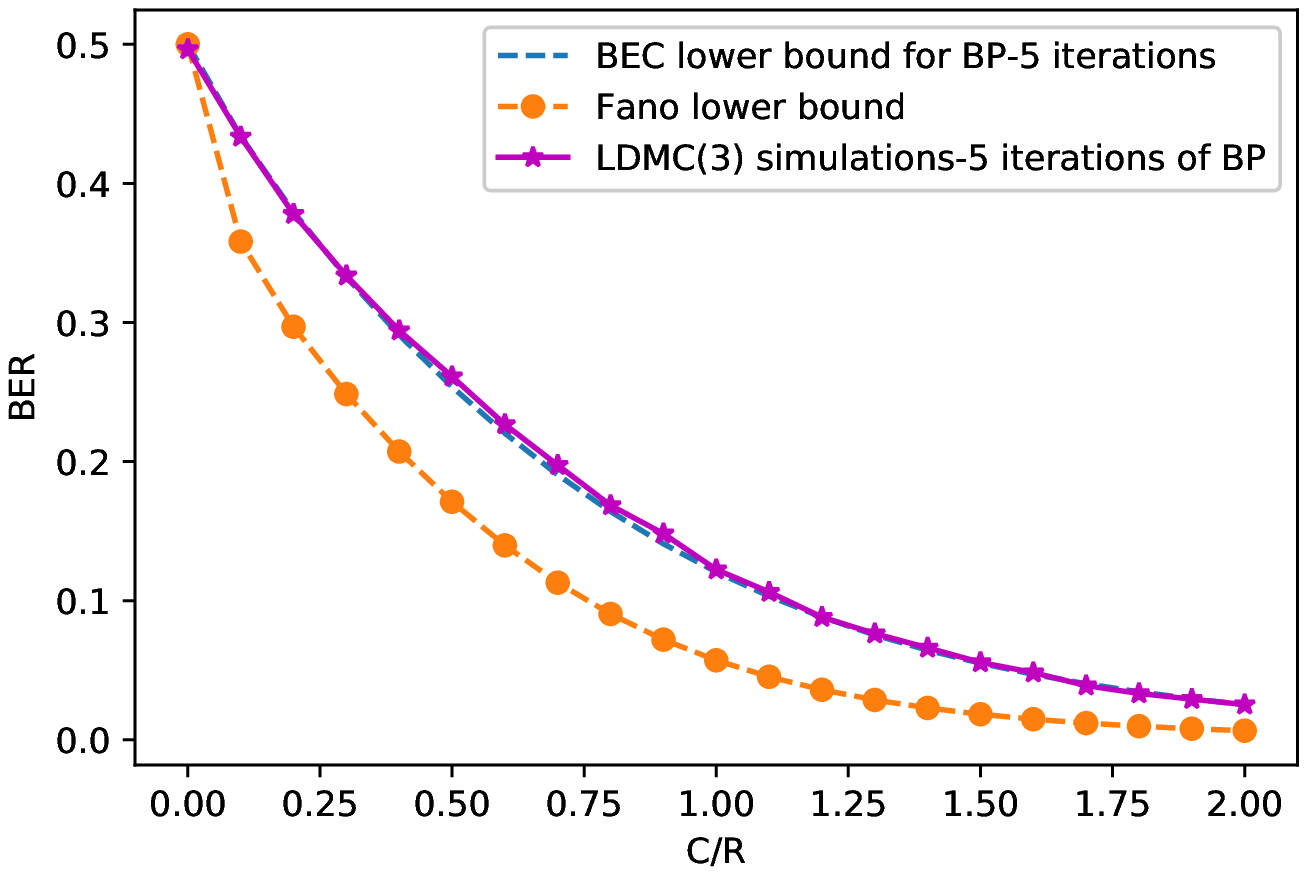}}
\caption{a) The direct bound on soft information obtained from channel comparison is tighter than that of (\ref{eq:ber_to_I}) implied from BER. b) The direct bound on BER is tighter than that inferred from Fano's inequality (\ref{eq:fano_ber}) and the upper bound on soft information obtained by matching $\chi^2$-capacity as shown in Fig.~\ref{fig:soft_info_bounds}a. }
\label{fig:Fano}
\end{figure}

\IEEEpeerreviewmaketitle



%

\appendices

\section{Complementary  proofs}\label{apx:complem}

\begin{proof}[Proof of Lemma~\ref{lem:lemma_A}]

The first part of lemma is well-known. The BEC part is called the erasure decomposition lemma  ~\cite[Lemma 4.78]{richardson2008modern} and the BSC part is a consequence of data processing (or, more precisely, hard decision decoding) ~\cite[Problem 4.55]{richardson2008modern}. 
The BSC half of the second part has been shown in~\cite[Appendix]{sasoglu2011polar}. The rest of the statements appear to be
new. 

Let $Y_\delta$ denote output of a $\BSC_\delta$ applied to input $X$. Then $\BMS$ $W$ can be represented as $X\mapsto
(Y_\Delta,\Delta)$ where $\Delta\in[0,1/2]$ is a random variable independent of $X$. To prove the BEC part of the second
claim, we need to show that for any input distribution $P_X=\Ber(p)$ we have
$$ I(X; Y_\Delta, \Delta) \le I(X; Y_E)\,,$$
where $Y_E$ is the output of a $\BEC_{1-C}$. Note that $I(X;Y_E) = C H(X)$. Thus, we need to show that for any
distribution of $\Delta$ and for any $p$ the following inequality holds:
\begin{equation}\label{eq:la_2}
	\EE[h_b(p\ast \Delta)] - \EE[h_b(\Delta)] \le (1-\EE[h_b(\Delta)]) h_b(p)\,,
\end{equation}
where $h_b$ denotes the base-2 binary entropy function and $a\ast b = a(1-b)+(1-a)b$ is the binary convolution function. A
result known as Mrs. Gerber's Lemma (MGL), cf.~\cite{wyner1973theorem}, which states that the parametric function
\begin{equation}\label{eq:la_mgl}
	h_b(p \ast \delta) \quad \mbox{vs.}\quad h_b(\delta), \qquad \delta \in [0,1/2]  
\end{equation}is convex. 
Consequently, the function must be below the chord connecting its endpoints, i.e. 
$$ h_b(p \ast \delta) \le (1-h_b(\delta)) h_b(p) + h_b(\delta)\,.$$
Clearly, the latter inequality implies~\eqref{eq:la_2} after taking expectation over $\delta$. Note also that the
BSC part of the second claim also follows from~\eqref{eq:la_mgl}. Indeed, from convexity we have
\begin{equation}\label{eq:la_3}
	\EE[h_b(p \ast \Delta)] \ge h_b(p \ast \delta_{eff})\,,
\end{equation}
where $\delta_{eff}$ is chosen so that $h_b(\delta_{eff})=\EE[h_b(\Delta)]$. In turn,~\eqref{eq:la_3} is equivalent to the
first relation in~\eqref{eq:la_mc}.

To prove the third part of the Lemma, i.e.~\eqref{eq:la_ln}, take
$P_X = \Ber(p)$ and $Q_X = \Ber(q)$ and let $P_Y$, $Q_Y$ be the output distributions induced by $W$. Similarly, let
$P_{Y_B}, Q_{Y_B}$ and $P_{Y_E}, Q_{Y_E}$ be the distributions induced by the equal-$\chi^2$-capacity $\BSC$ and $\BEC$,
respectively. We need to show (using (\ref{eq:less_noisy_kl})) that 
	$$ D(P_{Y_B} \| Q_{Y_B}) \le D(P_{Y} \| Q_Y)  \le D(P_{Y_E} \| Q_{Y_E})\,.$$
First notice that
	\begin{align} C_{\chi^2}(\BSC_\delta) &= (1-2\delta)^2 \\
	   C_{\chi^2}(\BEC_\delta) &= 1-\delta 
\end{align}
After representing $\BMS$ as a mixture of $\BSC$'s we have $\eta = \EE[(1-2\Delta)^2]$. 
Introducing the binary divergence function $d(a\|b) = D(\Ber(a) \| \Ber(b))$ we need to show: For any distribution of
$\Delta\in[0,1/2]$ and $p,q\in [0,1]$ we have
\begin{equation}\label{eq:la_4}
	d(p \ast \delta_{eff} \| q \ast \delta_{eff}) \le \EE[D(p \ast \Delta \| q \ast \Delta)] 
		\le \EE[(1-2\Delta)^2] d(p\|q)\,,
\end{equation}		
where $\delta_{eff}=\frac{1-\sqrt{\eta}}{2}$ is defined to satisfy
	$$ (1-2\delta_{eff})^2 = \EE[(1-2\Delta)^2]\,.$$
Note that the right-most inequality in~\eqref{eq:la_4} follows froma a well-known fact 
that the strong data-processing contraction coefficient $\eta_{KL}(W)$ equals $\EE[(1-2\Delta)^2]$ (e.g. this follows
from the proof of~\cite[Theorem 21]{polyanskiy2017strong}).

To prove~\eqref{eq:la_4} we will establish \textit{a variant of the MGL}, possibly of separate interest. Namely, 
we will show that for any fixed $p,q \in [0,1]$ the function
\begin{equation}\label{eq:la_5}
		d(p\ast \delta \| q \ast \delta) \quad \mbox{vs.} \quad (1-2\delta)^2 \qquad \delta \in[0,1/2] 
\end{equation}	
is convex.  Clearly,~\eqref{eq:la_5} would imply both sides of~\eqref{eq:la_4}.

To show~\eqref{eq:la_5} we proceed directly. Change parametrization to $x=(1-2\delta)^2$ and thus $\delta = \delta(x) =
\frac{1-\sqrt{x}}{2}$. Letting $d(x; p,q) = d(p\ast \delta(x) \| q \ast \delta(x))$ we find
\begin{align} \partial_x d(x; p,q) &= -\frac{1}{4\sqrt{x}} a(x; p,q)\\
	a(x; p,q) & = \left(\ln {p\ast \delta\over 1-p\ast \delta} + \ln {1-q\ast \delta \over q\ast \delta}\right)(1-2p) 
	+ \left( {1-p\ast \delta \over 1-q\ast \delta} - {p\ast \delta\over q\ast \delta}\right) (1-2q)  \,.
\end{align}
For convenience, let us introduce 
	$$ s \eqdef p \ast \delta, \quad \sigma \eqdef q \ast \delta\,.$$
Differentiating again, we get that convexity constraint $\partial_x^2 d(x; p,q) \ge 0$ is equivalent to the following
inequality:
\begin{equation}\label{eq:la_6}
		2 a(x; p,q) + \sqrt{x} b(x; p,q) \ge 0\,, 
\end{equation}	
where (we used $1-2p = {1-2s\over 1-2\delta}$ and $1-2q = {1-2\sigma\over 1-2\delta}$) 
$$ b(x; p, q) = {1\over (1-2\delta)^2} \left( {(1-2s)^2\over s(1-s)} + (1-2\sigma)^2 {s(1-\sigma)^2 + (1-s)\sigma^2
\over 
	\sigma^2 (1-\sigma)^2}  + {4s\over \sigma(1-\sigma)} (1-2s)(1-2\sigma)\right)\,.$$
Noticing that $\sqrt{x} = 1-2\delta$ and multiplying~\eqref{eq:la_6} by $(1-2\delta)$ we get that we need to verify 
\begin{equation}\label{eq:la_7}
	2 (1-2s) \ln {s (1-\sigma)\over (1-s) \sigma} + {2(1-2\sigma)\over \sigma(1-\sigma)}(\sigma+s-1) + {(1-2s)^2\over
s(1-s)} + (1-2\sigma)^2 {s(1-2\sigma)+\sigma^2\over \sigma^2 (1-\sigma)^2} \ge 0\,.
\end{equation}
Note that this inequality needs to hold for all values of $s,\sigma \in[\delta,1/2]$. However, due to arbitrariness of
$\delta$ and since it does not appear in~\eqref{eq:la_7} (this is crucial), we need to simply establish~\eqref{eq:la_7}
on the unit square $[0,1/2]^2$. 

Here again, we reparameterize
	$$ u \eqdef 1-2s, \quad v \eqdef 1-2\sigma $$
so that $(u,v) \in [0,1]^2$ now range over the unit square. 
Then~\eqref{eq:la_7} is rewritten as (after dividing by $u$) 
\begin{equation}\label{eq:la_8}
	f(u,v) \eqdef 2 \ln {(1-u)(1+v)\over (1+u)(1-v)} - 4 {v\over u(1-v^2)} (u+v) + {4u\over 1-u^2} + {4 v^2\over u(1-v^2)^2} (2(1-u)v +
(1-v)^2) \ge 0\,.
\end{equation}
It is easy to check that this inequality holds when either $u=0+,1-$ or $v=0+,1-$. Thus, we only need to rule out violations
inside the $[0,1]^2$. Taking derivative over $u$ of $f(u,v)$ we get
$$ \partial_u f = 0 \quad \iff \quad {2u^4\over(1-u^2)^2} = {2v^4 \over (1-v^2)^2}\,,$$
since $t \mapsto {t\over 1-t}$ is monotone, this implies that minimum of $f(u,v)$ is attained at $u=v$. But $f(u,u)=0$.
Thus, we find
	$$ \min_{u,v} f(u,v) = \min_u f(u,u) = 0\,.$$
This concludes the proof of~\eqref{eq:la_8} and, hence, of~\eqref{eq:la_5}.
\end{proof}

\begin{proof}[Proof of Lemma~\ref{lem:lemma_B}]
If $\tilde W_i$'s are degraded w.r.t $W_i$, then we have a Markov chain $X_0- (Y,Y_1^m)- (Y,\tilde Y_1^m)$. This proves the first part. 

To prove the second part, we may assume by induction that $W_i=\tilde W_i$ for all $i\ge 2$ (i.e. only
  one channel is replaced). Now suppose we have 
  \begin{equation}\label{eq:R1}U\ci (X_1^m,Y_1^m,\tilde Y_1^m, Y)|X_0\end{equation} We want to show
\[
I(U;Y,\tilde Y_1, Y_2^m)\le I(U;Y, Y_1,Y_2^m) 
\]
or equivalently
\begin{equation}
I(U;\tilde Y_1|Y, Y_2^m)\le I(U; Y_1|Y,Y_2^m). 
\label{eq:ln_proof}
\end{equation}
The desired inequality follows from the definition of the less noisy order (in the conditional universe where $(Y,Y_2^m)$ is observed) if we can show $U-X_0-X_1-(Y_1,\tilde Y_1)$ form a Markov chain conditionally on $(Y,Y_2^m)$. Note that this is equivalent (by d-separation) to showing that the conditional independence assertions of the Lemma are representable by the following directed acyclic graphical model (DAG)

\begin{center}
\tikzset{node distance=1.5cm, auto}
\begin{tikzpicture}[auto]
  \node (U){$U$};
  \node (X0)[right of=U,node distance=1.5cm] {$X_0$};
  \node (X1) [right of=X0,node distance=1.5cm]{$X_1$};
  \node (Y1) [right of=X1,node distance=1.5cm]{$(Y_1,\tilde Y_1)$};
  \node (A) [below of=X0,node distance=1cm]{};
  \node (Y) [right of=A,node distance=.75cm]{$Y$};
  \node (X2) [below of=Y,node distance=1cm]{$X_2^m$};
  \node (Y2) [right of=X2,node distance=1.5cm]{$Y_2^m$};
  \path [line] (X0) --node[above]{} (Y);  
  \path [line] (U) --node[above]{} (X0);  
  \path [line] (X0) --node[above]{} (X1);  
  \path [line] (X1) --node[above]{} (Y1);  
  \path [line] (X1) --node[above]{} (Y);  
  \path [line] (X2) --node[above]{} (Y);  
  \path [line] (X2) --node[above]{} (Y2);  
\end{tikzpicture}
\end{center}
We first recall (from d-separation) that \begin{equation}\label{eq:R2}A\ci (B,C)|D \implies A\ci B|(C,D)\end{equation} for arbitrary random variables $(A,B,C,D)$ (cf. (S3) in ~\cite[Chapter 3]{lauritzen1996graphical}). From (\ref{eq:R1}) and $(\ref{eq:R2})$, we see that the main assertion of interest in the above DAG is \begin{equation}\label{eq:R3}U\ci (Y_1,\tilde Y_1)|(X_1,Y,Y^m_2)\end{equation} To prove this assertion, we note that
\begin{align*}
p_{UY_1\tilde Y_1|X_1YY^m_2}&=\sum_{x_0,x_2^m} p_{UY_1\tilde Y_1|X_0X_2^mX_1YY_2}p_{X_0X_2^m|X_1YY_2^m}\\
&\stackrel{\textup{by } (\ref{eq:R1})}{=} \sum_{x_0,x_2^m} p_{U|X_0}p_{Y_1\tilde Y_1|X_0X_2^mX_1YY_2}p_{X_0X_2^m|X_1YY_2^m}\\
&\propto \sum_{x_0,x_2^m} p_{U|X_0}p_{YY_1\tilde Y_1|X_0X_2^mX_1Y_2}p_{X_0X_2^m|X_1YY_2^m}\\
&= \sum_{x_0,x_2^m} p_{U|X_0}p_{Y_1\tilde Y_1|X_1}p_{Y|X_0X_1X_2}p_{X_0X_2^m|X_1YY_2^m}\\
&=p_{Y_1\tilde Y_1|X_1}\sum_{x_0,x_2^m} p_{U|X_0} p_{Y|X_0X_1X_2}p_{X_0X_2^m|X_1YY_2^m} \\
&=h(y_1,\tilde y_1,x_1)g(u,x_1,y,y_2^m).
\end{align*}
This proves (\ref{eq:R3}) and the desired inequality in (\ref{eq:ln_proof}) follows. 

\end{proof}

\section{Erasure polynomials for LDMC(3)}
\label{apx:polys}
Here we include the $d$-th erasure polynomials $E_d^\BEC(q)$ for $d\le 10$ in Python form for LDMC(3). These polynomials are generated using the procedure described in Section \ref{sec:Ecurve_ldmc3} and are used to produce the bounds in Figs. \ref{fig:exit_deg48}-\ref{fig:ldmc_tight_bounds} and Table \ref{table:E_A_BER}, as well as for code optimization in Section \ref{sec:comp_ldgm}. 

\begin{longtable}{c}
$E^{\BEC}_d$\\
 \hline
 $d=0$\\ 0.5\\
 \hline
  $d=1$\\ 0.25\\
 \hline
 $d=2$\\ 0.125*q**4 - 0.25*q**3 + 0.25*q**2 - 0.25*q + 0.25\\
 \hline
 $d=3$\\
 0.1875*q**6 - 0.46875*q**5 + 0.46875*q**4 - 0.1875*q**3 + 4.440892e-16*q**2 - 0.09375*q + 0.15625\\
 \hline 
$d=4$\\  
\makecell{0.46875*q**8 - 1.9375*q**7 + 3.71875*q**6 - 4.3125*q**5 +\\ 3.28125*q**4 - 1.6875*q**3 + 0.65625*q**2 - 0.3125*q + 0.15625}\\
\hline
\makecell{$d=5$\\0.9375*q**10 - 4.58007812500001*q**9 + 10.087890625*q**8 - 13.0859375*q**7\\ + 10.976562500*q**6 - 6.15234375*q**5 + 2.24609375*q**4 - 0.4296875*q**3\\ + 0.0390625*q**2 - 0.126953125*q + 0.103515625}\\
\hline
$d=6$\\ \makecell{2.2900390625*q**12 - 14.455078125*q**11 + 42.9462890624997*q**10\\ - 79.5214843749996*q**9 + 102.12890625*q**8 - 95.5664062499994*q**7\\ + 66.5722656249995*q**6 - 34.7460937499998*q**5 + 13.5791015624999*q**4\\ - 3.99414062499997*q**3 + 0.981445312499996*q**2 - 0.310546875*q + 0.103515625}\\
\hline
$d=7$\\ \makecell{5.05517578125*q**14 - 36.368896484375*q**13 + 121.872802734375*q**12\\ - 251.26171875*q**11 + 354.7236328125*q**10 - 361.612548828124*q**9\\ + 274.061279296874*q**8 - 156.953124999999*q**7 + 68.3422851562493*q**6\\ - 22.3791503906245*q**5 + 5.18676757812476*q**4 - 0.68359374999993*q**3\\ + 0.0820312499999894*q**2 - 0.131591796874999*q + 0.070556640625}\\
\hline
$d=8$\\\makecell{12.2824707031249*q**16 - 104.389648437499*q**15 + 421.901855468746*q**14\\ - 1078.21191406248*q**13 + 1953.12304687496*q**12 - 2661.41503906243*q**11\\ + 2821.08544921866*q**10 - 2368.68652343741*q**9 + 1587.56103515619*q**8\\ - 849.672851562463*q**7 + 361.467285156233*q**6 - 121.303710937495*q**5\\ + 31.8554687499987*q**4 - 6.56933593749975*q**3 + 1.18603515624997*q**2\\ - 0.282226562499999*q + 0.070556640625}\\
\hline
$d=9$\\\makecell{28.517944335937*q**18 - 270.209632873528*q**17 + 1213.44797515864*q**16\\ - 3426.34039306621*q**15 + 6803.52593994091*q**14 - 10066.1437683096*q**13\\ + 11474.6510925279*q**12 - 10284.0617065415*q**11 + 7337.84271240106*q**10 \\- 4200.8187103263*q**9 + 1938.88133239697*q**8 - 722.934997558378*q**7\\ + 216.904724121012*q**6 - 51.2509460448973*q**5 + 8.86129760741628*q**4\\ - 0.913879394530327*q**3 + 0.115905761718657*q**2 - 0.122840881347652*q\\ + 0.0489273071289062}\\
\hline
$d=10$\\\makecell{69.4315452575683*q**20 - 742.947502136231*q**19 + 3808.84984970093*q**18\\ - 12453.0257034302*q**17 + 29158.0880355837*q**16 - 52039.3605651862*q**15\\ + 73535.9429168715*q**14 - 84300.01968384*q**13 + 79613.6392593413*q**12\\ - 62487.5754547145*q**11 + 40908.980049135*q**10 - 22326.4821624763*q**9\\ + 10117.3594665529*q**8 - 3780.88119506833*q**7 + 1154.60105895993*q**6\\ - 285.082305908195*q**5 + 56.3512229919426*q**4 - 8.95305633544941*q**3\\ + 1.28042221069343*q**2 - 0.244636535644538*q + 0.0489273071289063}\\
\hline
\end{longtable}

\section*{Acknowledgment}
The authors would like to thank Prof. Frank Kschischang for many interesting discussions, and especially for suggesting
the name ``low-density majority codes''.

\ifCLASSOPTIONcaptionsoff
  \newpage
\fi



\bibliographystyle{IEEEtran}
\bibliography{IEEEabrv,journal_draft_v5}

\end{document}